\documentclass[12pt]{amsart} 
\usepackage[english]{babel}
\usepackage[utf8]{inputenc}
\usepackage{amsmath}
\usepackage{graphicx}
\usepackage{subfigure}
\usepackage{amssymb}
\usepackage{amsthm}
\usepackage{bm}
\newcommand{\nm}{\noalign{\smallskip}}
\newcommand{\ds}{\displaystyle}
\usepackage{tikz-cd}
\usepackage{mathrsfs}
\usepackage[colorinlistoftodos]{todonotes}
\usepackage{enumitem}
\usepackage{yfonts}
\usepackage{ dsfont }
\usepackage{bbm}
\usepackage{geometry}
\usepackage{setspace}
\usepackage{tikz}
\usepackage{pgfplots}
\usepackage{cite}
\pgfplotsset{compat=newest}
\usepgfplotslibrary{fillbetween}
\usetikzlibrary{positioning}
\usetikzlibrary{decorations.pathreplacing}
\usetikzlibrary{decorations,arrows}
\usetikzlibrary{decorations.markings}
\usetikzlibrary{patterns}
\usetikzlibrary{quotes,angles}
\usetikzlibrary{arrows}

\numberwithin{equation}{section}
\newtheorem{theorem}{Theorem}[section]
\newtheorem{lemma}[theorem]{Lemma}

\newtheorem{definition}[theorem]{Definition}

\newtheorem{hypothesis}[theorem]{Hypothesis}
\newtheorem{corollary}[theorem]{Corollary}

\newtheorem{remark}[theorem]{Remark}
\newtheorem{proposition}[theorem]{Proposition}
\newtheorem{assumption}[theorem]{Assumption}
\allowdisplaybreaks[4]

\newcommand*{\rom}[1]{\expandafter\@slowromancap\romannumeral #1@}

\title[Bulk-Edge Correspondence for Finite Disordered Systems]{Bulk-Edge Correspondence for Finite Two-dimensional Ergodic Disordered Systems}

\author{Habib Ammari} %
\address[H. Ammari]{Department of Mathematics, ETH Z\"{u}rich, R\"{a}mistrasse 101, CH-8092 Z\"{u}rich, Switzerland; Hong Kong Institute for Advanced Study, City University of Hong Kong, Kowloon Tong, Hong Kong}
\email{habib.ammari@math.ethz.ch}

\author{Jiayu Qiu} %
\address[J. Qiu]{Department of Mathematics, ETH Z\"{u}rich, R\"{a}mistrasse 101, CH-8092 Z\"{u}rich, Switzerland}
\email{jiayu.qiu@sam.math.ethz.ch}

\date{}

\begin{document}

\begin{abstract}
In this paper, we rigorously prove the bulk-edge correspondence for finite two-dimensional ergodic disordered systems. Specifically, we focus on the short-range Hamiltonians with ergodic disordered on-site potentials. We first introduce the bulk and edge indices, which are both well-defined within the Aizenman-Molchanov mobility gap. On the one hand, the bulk index is the usual Hall conductance, which is a well-studied quantized topological number. On the other hand, the edge index, which characterizes the averaged angular momentum of edge modes in the mobility gap, is uniquely associated with finite systems. Our main result proves that as the sample size tends to infinity, the edge index converges to the bulk index almost surely. Our findings provide a rigorous foundation for the bulk-edge correspondence principle for finite disordered systems. The existence of the Aizenman-Molchanov mobility gap is proved by the geometric decoupling method, introduced by Aizenman and Molchanov [Comm. Math. Phys., 1993], under a rational assumption on the distribution of the random potential. For completeness, all assumptions are checked on a prototypical model for (quantum) anomalous Hall physics.
\end{abstract}

\maketitle

\bigskip

\par

\bigskip

\\

\tableofcontents

\section{Introduction}

Since the discovery of the celebrated quantum Hall effect (QHE) and its topological origin \cite{Klitzing80qhe,vonKlitzing2020forty_years}, the topological phase of matter has gained great interest within both the physics and mathematics communities. The most remarkable properties of topological materials are rooted in the fundamental concept of the bulk-edge correspondence (BEC) principle. This principle asserts that there necessarily exist localized waves at the edge of the system if the bulk medium possesses nontrivial topological characters, which take quantized values and originate from the ground state topology of the system \cite{qizhang11topo_insulator,bernivig13topo_insulator}. Due to their topological nature, these special edge modes are robust against disorder present in the system, enabling a wide range of applications. Motivated by the study in condensed matter systems, it has been argued and demonstrated that the BEC principle is universal and ubiquitous in various physical systems, such as photonic crystals and others \cite{haldane08realization,wang2009observation,ozawa19topological_photonics}.

Since the seminal work by Hatsugai \cite{Hatsugai93chern}, the BEC principle has been established rigorously across various set-ups, primarily focusing on electronic systems described by the (magnetic or non-magnetic) Schrödinger operator \cite{Kellendonk02landau+ktheory,Kellendonk2004landau+ktheory,Kellendonk04landau+functional,taarabt2014landau+Ktheory,cornean2021landau+functional,shapiro2022shrodinger+functional,drouot2021microlocal}, Dirac Hamiltonians \cite{bal2019dirac+functional,bal2023dirac+microlocal}, tight-binding models \cite{graf2018shortrange+transfer,avila2013shortrange+transfer,graf2013shortrange+scattering,shapiro2022shrodinger+functional,graf2005equality,tauber2022chiral_finite_chain,qiu2025generalized,braverman2018spectralflowfamilytoeplitz,bourne2017ktheory,kubota2017ktheory}, and classical wave systems \cite{lin2022transfer,thiang2023transfer,ammari2024toeplitz_1,ammari2024toeplitz_2,
ammari2024toeplitz_3,SWP1,SWP2,SWP3,SWP4,qiu2025bec_finite}. The techniques used in these studies cover a broad spectrum, including K-theory, scattering theory, microlocal analysis, functional analysis, Toeplitz theory and coarse geometry, etc.; we refer the reader to the monograph \cite{prodan2016ktheory} and introduction of \cite{cornean2021landau+functional,qiu2025generalized} for an overview, and the review paper \cite{bal2024continuous} for an exhibition of recent results. Note that while extensive research has established the BEC principle for infinite domains, such as half-plane models with edges or interface models consisting of two separated half-planes, the mathematical framework for finite geometries remains underexplored. Recent works \cite{bal2023edge_curved,ludewig2020shortrange+coarse,drouot2024bec_curvedinterfaces} have addressed more general domains in which the edge or interface may not be straight, while still requiring the infinite size of domains. We note that the work \cite{tauber2022chiral_finite_chain} has dealt with the BEC principle for finite one-dimensional chains. Motivated by the growth of topological physics and the practical reality that experiments are performed on finite two-dimensional samples, a rigorous theory of the BEC principle for finite high-dimensional systems is needed. In our recent work \cite{qiu2025bec_finite}, we proved the BEC principle for finite two-dimensional domains for the first time, where the bulk medium is assumed to be periodic (i.e., without disorder) and possess a spectral gap. In such cases, it is well-known that the topological character of the bulk medium is the Chern number. Remarkably, by introducing an edge index that is uniquely associated with finite systems and characterizes the circulation of waves near the sample edge, it was shown in \cite{qiu2025bec_finite} that the edge index converges to the Chern number as the size of the domain grows to infinity. Although this result establishes the BEC principle for finite systems, its practical significance is limited by the strong assumption that the bulk medium is periodic. As robustness against disorder constitutes the most remarkable feature of topological materials, it is highly desirable to extend the idea developed in \cite{qiu2025bec_finite} to the disorder regime, which motivates this work.

In this paper, we rigorously prove the bulk-edge correspondence for finite two-dimensional ergodic disordered systems. Specifically, we focus on short-range Hamiltonians with ergodic disordered on-site potentials (Assumptions \ref{asmp_Hamiltonian_deter_part} and \ref{asmp_random_potential}). For such Hamitonians, a natural characterization of Anderson localization in a given frequency range is the Aizenman-Molchanov mobility gap (Hypothesis \ref{thm_anderson_localization}), which generalizes the usual spectral gap in the presence of disorder. We first introduce the bulk and edge indices in such mobility gaps. First, the bulk index is taken as the usual Hall conductance, which takes quantized values in the mobility gap and recovers the Chern number in the absence of disorder. On the other hand, the edge index is the same as in \cite{qiu2025bec_finite}, which measures the averaged angular momentum of modes in the mobility gap, minus the contribution of bulk-localized modes caused by disorder; see Definition \ref{prop_refined_bulk_index}. Hence, the refined edge index measures purely the circulation of edge modes in the finite disordered system. Our main result (Theorem \ref{thm_bec_disorder}) proves that, as the sample size tends to infinity, the edge index converges to the bulk index almost surely, thereby establishing the BEC principle for finite disordered systems and extending the periodic result to the disorder regime. Moreover, the existence of the Aizenman-Molchanov mobility gap is proved by the geometric decoupling method introduced by Aizenman et al. \cite{aizenman1993localization_elementary_derivation,aizenman1994localization,aizenman1998localization,aizenman2001finite_volime} under a rational assumption on the distribution of the random potential (Theorem \ref{thm_anderson_localization}). For completeness, we check all assumptions for the Qi–Wu–Zhang (QWZ) model in Appendix \ref{appendixA}, a prototypical Hamiltonian for (quantum) anomalous Hall physics. We also check in Appendix \ref{appendixB} that the bulk medium is topologically nontrivial in the presence of weak disorder. Our results in this paper provide a rigorous foundation for the BEC principle for finite disordered systems and support emerging theories of two-dimensional topological alloy-type materials.

\section{Set-up and Main Results}

We consider a tight-binding periodic Hamiltonian $\mathcal{H}_0$ acting on $\mathcal{X}=\ell^2(\mathbb{Z}^2)\times \mathbb{C}^d$, where $\mathbb{C}^d$ represents internal degree of freedom (DOF) such as the sublattices and spins. Denote $\mathcal{H}_0(\bm{n},\bm{m})\in \mathbb{C}^{d\times d}$ by its matrix elements. The basic properties of $\mathcal{H}_0$ are summarized as follows.
\begin{assumption} \label{asmp_Hamiltonian_deter_part}
i) (Periodicity) $\mathcal{H}_0(\bm{n}+\bm{e}_i,\bm{m}+\bm{e}_i)=\mathcal{H}_0(\bm{n},\bm{m})$;

ii) (Short-rangeness) $\mathcal{H}_0(\bm{n},\bm{m})=0$ if $\|\bm{n}-\bm{m}\|_{1}\geq 2$.
\end{assumption}

Throughout this paper, we will use the $\ell^1$ norm for $\mathbb{Z}^2$ vectors and abbreviate the notion as $|\bm{n}|:=\|n\|_{1}$. The norm notation without subscripts is preserved for (Frobenius) norm for $n\times n$ matrices, i.e., $\|A\|:=\|A\|_{F}$. The norm of other linear spaces always comes with specified subscripts.

\begin{remark}
Note that by the arbitrariness of the internal DOF $d\geq 1$, all periodic Hamiltonians with finite range are covered by Assumption \ref{asmp_Hamiltonian_deter_part}. We also remark that the assumption of finite range can be relaxed: with more technical effort, the main result of this paper is expected to hold for general Hamiltonians with exponentially decaying kernels; this extension is left for the interested reader.
\end{remark}

Next, we introduce the random (disordered) potential. Its properties are summarized as follows.
\begin{assumption}[Random potential]
\label{asmp_random_potential}
Let $V=\text{diag}(V_1,\cdots ,V_d)\in \mathbb{R}^{d\times d}$ be a diagonal matrix, and $\{\omega_{\bm{n}}^{(i)}\}_{\bm{n}\in\mathbb{Z}^{2},1\leq i\leq d}$ be a sequence of i.i.d. random variables with a common bounded density $\mu\in L^{\infty}(\mathbb{R})$ of compact support. Then the random potential operator is defined as
\begin{equation*}
V_{\omega}\psi\in \mathcal{B}(\mathcal{X}),\quad (V_{\omega}\psi)(\bm{n}):=\text{diag}(\omega_{\bm{n}}^{(1)}V_1,\cdots ,\omega_{\bm{n}}^{(d)}V_d)\cdot \psi(\bm{n}).
\end{equation*}
\end{assumption}
The random operator $V_{\omega}$ is covariant with respect to the spatial translation; see Section \ref{sec_covariant_tpuv} for more details. Assumptions \ref{asmp_Hamiltonian_deter_part} and \ref{asmp_random_potential} will be in place throughout the paper, so we do not repeat them. The main object of study is the transport property of the following random Hamiltonian:
\begin{equation*}
\mathcal{H}_{\omega}:=\mathcal{H}+V_{\omega}\in \mathcal{B}(\mathcal{X}).
\end{equation*}
The random variable $\omega$ belongs to the product space $\Omega=\mathbb{R}^{\mathbb{Z}^2 d}$ equipped with the probability measure $\mathbb{P}$, which is generated by the pre-measure induced by $\mu$ on the Borel cylinder sets in $\Omega$. Due to randomness, the spectrum of $\mathcal{H}_{\omega}$ is expected to exhibit Anderson localization. That is, within a given range of spectral variables, the spectrum of $\mathcal{H}_{\omega}$ is a pure point spectrum with exponentially localized eigenfunctions; see the monographs \cite{kirsch2007invitation,stollmann2001caught,stolz2011introduction} for an overview of Anderson localization. An elegant description of such a localization regime is based on the fractional momentum of the Green function $G_\omega$ of $\mathcal{H}_{\omega}$ introduced by Aizenman and Molchanov \cite{aizenman1993localization_elementary_derivation}, which is supposed (and verified later) as follows.
\begin{hypothesis}[Aizenman–Molchanov mobility gap] \label{hypo_anderson_localization}
An open interval $\mathcal{I}\subset \mathbb{R}$ is an Aizenman–Molchanov mobility gap if it satisfies the following condition: for any $s\in (0,1)$, there exist $\alpha,C>0$ such that
\begin{equation*}
\mathbb{E}(\|G_\omega(\bm{n},\bm{m};\lambda+i\epsilon)\|^s)\leq Ce^{-\alpha|\bm{n}-\bm{m}|}
\end{equation*}
for all $\lambda\in \mathcal{I}$ and $\epsilon>0$. Here, $G_\omega(\bm{n},\bm{m};z)=\langle\mathbbm{1}_{\{\bm{n}\}},(\mathcal{H}_{\omega}-z)^{-1}\mathbbm{1}_{\{\bm{m}\}}\rangle_{\ell^2(\mathbb{Z}^2)}$ denotes the Green function of $\mathcal{H}_{\omega}$.
\end{hypothesis}
As discussed in \cite[Remark 1]{graf2005equality} (see also \cite{aizenman1993localization_elementary_derivation,aizenman1998localization,graf2005equality,stolz2011introduction}), Hypothesis \ref{hypo_anderson_localization} implies that the spectral projection of $\mathcal{H}_{\omega}$ with end point lying within the mobility gap (ground state projection) is exponentially localized.
\begin{proposition} \label{prop_localization_ground_state}
There exists $C>0$ such that, for any $\bm{n},\bm{m}\in\mathbb{Z}^2$,
\begin{equation}
\label{eq_expectation_all_B1_function}
\mathbb{E}\big(\sup_{f\in B_1(\mathcal{I})}\|f(\mathcal{H}_{\omega})(\bm{n},\bm{m})\|  \big)\leq Ce^{-\alpha|\bm{n}-\bm{m}|}.
\end{equation}
Here, $B_1(\mathcal{I})$ denotes the set of Borel measurable functions $f$ that are constant outside $\mathcal{I}$ with $|f(x)|\leq 1$ for every $x\in\mathbb{R}$. Consequently, there exists $\nu>0$ such that for $\mathbb{P}-$a.s. $\omega$
\begin{equation*}
\sup_{f\in B_1(\mathcal{I})} \sum_{\bm{n},\bm{m}\in\mathbb{Z}^2}  \|f(\mathcal{H}_{\omega})(\bm{n},\bm{m})\|(1+|n|^{-\nu})e^{\alpha |\bm{m}-\bm{n}|}\leq C <\infty,
\end{equation*}
where $C=C(\omega)$ may depend on the realization. In particular, setting $P_{\omega,\lambda}:=\mathbbm{1}_{(-\infty,\lambda)}(\mathcal{H}_{\omega})$, it holds almost surely that
\begin{equation*}
\sum_{\bm{n},\bm{m}\in\mathbb{Z}^2}\|P_{\omega,\lambda}(\bm{n},\bm{m})\|(1+|n|^{-\nu})e^{\alpha |\bm{m}-\bm{n}|}\leq C <\infty,
\end{equation*}
for any $\lambda\in \mathcal{I}$.
\end{proposition}

Importantly, Proposition \ref{prop_localization_ground_state} implies that the Hall conductance is well-defined and quantized within the mobility gap. The proof can be found, for example, in \cite{Avron1994charge,graf2005equality,Schulz-Baldes01homotopy}.
\begin{proposition} \label{prop_quantization_index}
Define
\begin{equation} \label{eq_bulk_index}
\mathcal{E}_{bulk,\omega,\lambda}:=-i\cdot\mathcal{T}\big(P_{\omega,\lambda}\big[[P_{\omega,\lambda},x_1],[P_{\omega,\lambda},x_2]\big]\big),
\end{equation}
where $\mathcal{T}$ denotes the trace per unit of volume (TPUV) defined as
\begin{equation*}
\mathcal{T}(A):=\lim_{L\to\infty}\frac{1}{|\Lambda_L|}\text{Tr}_{\Lambda_L}(A),\quad \Lambda_L=[-L,L]\times [-L,L],
\end{equation*}
whenever the limit exists. Then for $\mathbb{P}-$a.e. $\omega$,
\begin{equation*}
\mathcal{E}_{bulk,\omega,\lambda}=\mathbb{E}(\mathcal{E}_{bulk,\omega,\lambda})\in\mathbb{Z}
\end{equation*}
and is constant in $\lambda\in\mathcal{I}$.
\end{proposition}

\begin{figure}
\begin{center}
\begin{tikzpicture}[scale=0.8]
\draw[dashed] ({-2},{2})--({2},{2});
\draw[dashed] ({-2},{-2})--({2},{-2});
\draw[dashed] ({2},{-2})--({2},{2});
\draw[dashed] ({-2},{-2})--({-2},{2});
\draw[fill=black,opacity=1] ({1/2},{1/2}) ellipse(0.1 and 0.1);
\draw[fill=black,opacity=1] ({1+1/2},{1/2}) ellipse(0.1 and 0.1);
\draw[fill=black,opacity=1] ({-1+1/2},{1/2}) ellipse(0.1 and 0.1);
\draw[fill=black,opacity=1] ({-2+1/2},{1/2}) ellipse(0.1 and 0.1);

\draw[fill=red,opacity=0.5] ({-3+1/2},{1/2}) ellipse(0.1 and 0.1);
\draw[fill=red,opacity=0.5] ({-4+1/2},{1/2}) ellipse(0.1 and 0.1);
\draw[fill=red,opacity=0.5] ({-3+1/2},{1/2+1}) ellipse(0.1 and 0.1);
\draw[fill=red,opacity=0.5] ({-4+1/2},{1/2+1}) ellipse(0.1 and 0.1);
\draw[fill=red,opacity=0.5] ({-3+1/2},{1/2+2}) ellipse(0.1 and 0.1);
\draw[fill=red,opacity=0.5] ({-4+1/2},{1/2+2}) ellipse(0.1 and 0.1);
\draw[fill=red,opacity=0.5] ({-3+1/2},{1/2-1}) ellipse(0.1 and 0.1);
\draw[fill=red,opacity=0.5] ({-4+1/2},{1/2-1}) ellipse(0.1 and 0.1);
\draw[fill=red,opacity=0.5] ({-3+1/2},{1/2-2}) ellipse(0.1 and 0.1);
\draw[fill=red,opacity=0.5] ({-4+1/2},{1/2-2}) ellipse(0.1 and 0.1);
\draw[fill=red,opacity=0.5] ({-3+1/2},{1/2-3}) ellipse(0.1 and 0.1);
\draw[fill=red,opacity=0.5] ({-4+1/2},{1/2-3}) ellipse(0.1 and 0.1);

\draw[fill=red,opacity=0.5] ({2+1/2},{1/2}) ellipse(0.1 and 0.1);
\draw[fill=red,opacity=0.5] ({3+1/2},{1/2}) ellipse(0.1 and 0.1);
\draw[fill=red,opacity=0.5] ({2+1/2},{1/2+1}) ellipse(0.1 and 0.1);
\draw[fill=red,opacity=0.5] ({3+1/2},{1/2+1}) ellipse(0.1 and 0.1);
\draw[fill=red,opacity=0.5] ({2+1/2},{1/2+2}) ellipse(0.1 and 0.1);
\draw[fill=red,opacity=0.5] ({3+1/2},{1/2+2}) ellipse(0.1 and 0.1);
\draw[fill=red,opacity=0.5] ({2+1/2},{1/2-1}) ellipse(0.1 and 0.1);
\draw[fill=red,opacity=0.5] ({3+1/2},{1/2-1}) ellipse(0.1 and 0.1);
\draw[fill=red,opacity=0.5] ({2+1/2},{1/2-2}) ellipse(0.1 and 0.1);
\draw[fill=red,opacity=0.5] ({3+1/2},{1/2-2}) ellipse(0.1 and 0.1);
\draw[fill=red,opacity=0.5] ({2+1/2},{1/2-3}) ellipse(0.1 and 0.1);
\draw[fill=red,opacity=0.5] ({3+1/2},{1/2-3}) ellipse(0.1 and 0.1);

\draw[fill=red,opacity=0.5] ({1/2},{1/2+2}) ellipse(0.1 and 0.1);
\draw[fill=red,opacity=0.5] ({1/2+1},{1/2+2}) ellipse(0.1 and 0.1);
\draw[fill=red,opacity=0.5] ({1/2-1},{1/2+2}) ellipse(0.1 and 0.1);
\draw[fill=red,opacity=0.5] ({1/2-2},{1/2+2}) ellipse(0.1 and 0.1);

\draw[fill=red,opacity=0.5] ({1/2},{1/2-3}) ellipse(0.1 and 0.1);
\draw[fill=red,opacity=0.5] ({1+1/2},{1/2-3}) ellipse(0.1 and 0.1);
\draw[fill=red,opacity=0.5] ({-1+1/2},{1/2-3}) ellipse(0.1 and 0.1);
\draw[fill=red,opacity=0.5] ({-2+1/2},{1/2-3}) ellipse(0.1 and 0.1);

\draw[fill=black,opacity=1] ({1/2},{1/2+1}) ellipse(0.1 and 0.1);
\draw[fill=black,opacity=1] ({1+1/2},{1/2+1}) ellipse(0.1 and 0.1);
\draw[fill=black,opacity=1] ({-1+1/2},{1/2+1}) ellipse(0.1 and 0.1);
\draw[fill=black,opacity=1] ({-2+1/2},{1/2+1}) ellipse(0.1 and 0.1);
\draw[fill=black,opacity=1] ({1/2},{1/2-1}) ellipse(0.1 and 0.1);
\draw[fill=black,opacity=1] ({1+1/2},{1/2-1}) ellipse(0.1 and 0.1);
\draw[fill=black,opacity=1] ({-1+1/2},{1/2-1}) ellipse(0.1 and 0.1);
\draw[fill=black,opacity=1] ({-2+1/2},{1/2-1}) ellipse(0.1 and 0.1);
\draw[fill=black,opacity=1] ({1/2},{1/2-2}) ellipse(0.1 and 0.1);
\draw[fill=black,opacity=1] ({1+1/2},{1/2-2}) ellipse(0.1 and 0.1);
\draw[fill=black,opacity=1] ({-1+1/2},{1/2-2}) ellipse(0.1 and 0.1);
\draw[fill=black,opacity=1] ({-2+1/2},{1/2-2}) ellipse(0.1 and 0.1);

\node[left,scale=1.2] at ({-3.5},{1/2}) {$\cdots$};
\node[left,scale=1.2] at ({-3.5},{1/2+1}) {$\cdots$};
\node[left,scale=1.2] at ({-3.5},{1/2-1}) {$\cdots$};
\node[left,scale=1.2] at ({-3.5},{1/2-2}) {$\cdots$};
\node[left,scale=1.2] at ({-3.5},{1/2-3}) {$\cdots$};
\node[left,scale=1.2] at ({-3.5},{1/2+2}) {$\cdots$};
\node[left,scale=1.2] at ({5},{1/2+1}) {$\cdots$};
\node[left,scale=1.2] at ({5},{1/2-1}) {$\cdots$};
\node[left,scale=1.2] at ({5},{1/2}) {$\cdots$};
\node[left,scale=1.2] at ({5},{1/2-2}) {$\cdots$};
\node[left,scale=1.2] at ({5},{1/2-3}) {$\cdots$};
\node[left,scale=1.2] at ({5},{1/2+2}) {$\cdots$};

\node[above,scale=1.2] at ({1/2},{2.8}) {$\vdots$};
\node[above,scale=1.2] at ({1/2+1},{2.8}) {$\vdots$};
\node[above,scale=1.2] at ({1/2-1},{2.8}) {$\vdots$};
\node[above,scale=1.2] at ({1/2-2},{2.8}) {$\vdots$};
\node[above,scale=1.2] at ({1/2+2},{2.8}) {$\vdots$};
\node[above,scale=1.2] at ({1/2+3},{2.8}) {$\vdots$};
\node[above,scale=1.2] at ({1/2-3},{2.8}) {$\vdots$};
\node[above,scale=1.2] at ({1/2-4},{2.8}) {$\vdots$};

\node[below,scale=1.2] at ({1/2},{-2.5}) {$\vdots$};
\node[below,scale=1.2] at ({1/2+1},{-2.5}) {$\vdots$};
\node[below,scale=1.2] at ({1/2+2},{-2.5}) {$\vdots$};
\node[below,scale=1.2] at ({1/2+3},{-2.5}) {$\vdots$};
\node[below,scale=1.2] at ({1/2-1},{-2.5}) {$\vdots$};
\node[below,scale=1.2] at ({1/2-2},{-2.5}) {$\vdots$};
\node[below,scale=1.2] at ({1/2-3},{-2.5}) {$\vdots$};
\node[below,scale=1.2] at ({1/2-4},{-2.5}) {$\vdots$};

\end{tikzpicture}
\caption{A finite sample $\Lambda_L$ (dashed square with black sites).}
\label{fig_finite_sample}
\end{center}
\end{figure}

By the independence shown above, we will omit the subscript $\lambda$ in $\mathcal{E}_{bulk,\omega,\lambda}$ whenever no confusion arises. As being well-known, the quantized index $\mathcal{E}_{bulk,\omega}$ is a topological number, which recovers the first Chern number of the eigenbundle $P_{\omega,\lambda}$ when $\mathcal{H}_{\omega}$ is non-random and periodic\footnote{e.g. $\omega$ degenerates to the Dirac distribution on a single point.} and $\mathcal{I}$ is a real spectral gap in the sense that $\mathcal{I}\cap \sigma(\mathcal{H}_{\omega})=\emptyset$; see \cite{Avron1994charge,Bellissard94non_commutative_QHE} and the references therein. In such cases, the BEC principle suggests that $\mathcal{I}_{bulk,\omega}\neq 0$ and the bulk medium is truncated into pieces (e.g., into two half-planes or finite subdomains), there necessarily exist special eigenmodes of the Hamiltonian that are localized near the edge of the medium. Mathematically, this is formulated and proved as follows in the case of finite subdomains \cite{qiu2025bec_finite}.\footnote{Here, we only display the result for the special cases where the subdomains are finite rectangles, while the original result \cite{qiu2025bec_finite} holds for arbitrary single-connected domains with regular boundaries.} Define the restricted Hamiltonian to the finite box $\Lambda_L$ as
\begin{equation} \label{eq_box_Hamiltonian_simple}
\mathcal{H}_{\omega,L}^{sim}:=\iota^{*}_{\Lambda_L}\mathcal{H}_{\omega}\iota_{\Lambda_L}\in \mathcal{B}(\mathcal{X}_L),\quad \mathcal{X}_L:=\ell^2(\Lambda_L)\otimes \mathbb{C}^d,
\end{equation}
where $\iota_{\Lambda_L}:\mathcal{X}_L\to \mathcal{X}$ is the canonical inclusion operator and $\iota^{*}_{\Lambda_L}$ denotes its adjoint.\footnote{This way of defining Hamiltonians on finite domains is also referred to as imposing simple boundary conditions on the boundary $\partial \Lambda_L$ \cite{kirsch2007invitation}, as indicated by the superscript.} Next, we introduce the following edge index for the box Hamiltonian $\mathcal{H}_{\omega,L}^{sim}$:
\begin{equation} \label{eq_edge_index}
\mathcal{E}_{edge,\omega,\mathcal{I},\rho,L}:=\frac{i}{2|\Lambda_L|}\text{Tr}_{\Lambda_L}\big(\mathcal{M}_{\omega,L}\rho^{\prime}(\mathcal{H}_{\omega,L}^{sim})\big),
\end{equation}
where $\mathcal{M}_{\omega,L}$ is the orbital angular momentum operator defined as
\begin{equation*}
\mathcal{M}_{\omega,L}:=[\mathcal{H}_{\omega,L}^{sim},x_2]x_1-[\mathcal{H}_{\omega,L}^{sim},x_1]x_2,
\end{equation*}
and $\rho\in C^{\infty}$ satisfies $\rho(\lambda)=0$ for $\lambda>\sup \mathcal{I}$ and $\rho(\lambda)=1$ for $\lambda<\inf \mathcal{I}$.\footnote{i.e., $\rho^{\prime}$ is a mollified version of the Dirac function locating in the mobility gap $\mathcal{I}$.} The edge index \eqref{eq_edge_index} is well-defined as the trace of a matrix on the finite-dimensional space $\mathcal{X}_{L}$. Being sloppy in notation, the result of \cite{qiu2025bec_finite} can be stated as follows.
\begin{theorem}[BEC principle for finite systems \cite{qiu2025bec_finite}]
\label{thm_bec_non_random_finite}
When $\mathcal{H}_{\omega}$ is deterministic and periodic, and $\mathcal{I}\cap \sigma(\mathcal{H}_{\omega})=\emptyset$, it holds for all $\rho$ (satisfying the conditions following \eqref{eq_edge_index}) that
\begin{equation*}
\lim_{L\to\infty}\mathcal{E}_{edge,L}=\mathcal{E}_{bulk}=\mathcal{C}_1(P_{\omega,\lambda}),
\end{equation*}
where $\mathcal{C}_1(P_{\omega,\lambda})$ denotes the first Chern number associated with the eigenbundle $\bm{\kappa}\mapsto P_{\omega,\lambda}(\bm{\kappa})$ and the latter is obtained by applying the Floquet-Bloch transform to the periodic operator $P_{\omega,\lambda}$.
\end{theorem}
Physically, the edge index \eqref{eq_edge_index} computes the expectation of averaged angular momentum carried by in-gap eigenmodes, which necessarily locates near the boundary $\partial \Lambda_L$ since the bulk medium is insulating (gapped); with this interpretation, Theorem \ref{thm_bec_non_random_finite} establishes the BEC principle for (non-random) finite systems on a rigorous mathematical ground. Mathematically, Theorem \ref{thm_bec_non_random_finite} indicates that the box operator $\mathcal{H}_{L}^{sim}$ must have in-gap eigenvalues with eigenfunctions localized near the boundary $\partial \Lambda_L$, as long as the bulk Hamiltonian has nonzero Chern number and the size of the domain is sufficiently large; see \cite[Corollary 1.3]{qiu2025bec_finite}.

However, when the disorder is taken into account, the situation is more complicated. The key point is that, except for the angular momentum carried by the edge states, the index \eqref{eq_edge_index} also counts the contribution of the bulk-localized states (with an energy located in the mobility gap) which emerges as a consequence of Anderson localization. To remedy it, we need to modify the definition of the edge index by subtracting the contribution of bulk-localized modes. The refined index is defined as follows, whose physical meaning is illustrated in Remark \ref{rmk_illustrate_refined_index}.
\begin{definition}[Refined edge index] \label{prop_refined_bulk_index}
For any open sub-interval $\mathcal{I}^{\prime}$ of the mobility gap $\mathcal{I}$ that includes the transition region of the density function $\rho$, i.e.,
\begin{equation*}
\text{supp}(\rho^{\prime}) \Subset \mathcal{I}^{\prime} \Subset \mathcal{I} ,
\end{equation*}
we define
\begin{equation} \label{eq_refined_bulk_index}
\begin{aligned}
\mathcal{E}_{edge,\omega,\mathcal{I},\mathcal{I}^{\prime},\rho,L}^{ref}&:=\mathcal{E}_{edge,\omega,\mathcal{I},\rho,L} - \mathcal{E}_{edge,\omega,\mathcal{I},\mathcal{I}^{\prime},\rho,L}^{cor,(1)}-\mathcal{E}_{edge,\omega,\mathcal{I},\mathcal{I}^{\prime},\rho,L}^{cor,(2)} .
\end{aligned}
\end{equation}
Here
\begin{equation} \label{eq_correction_term_1}
\mathcal{E}_{edge,\omega,\mathcal{I},\mathcal{I}^{\prime},\rho,L}^{cor,(1)}
:=\frac{i}{2\pi |\Lambda_L|}\text{Tr}_{\Lambda_L}\big( P_{\omega,loc} \int_{\mathbb{C}}dm(z)\partial_{\overline{z}}\tilde{\rho}(z) (\mathcal{H}_{\omega}-z)^{-1} \mathcal{M}_{\omega} (\mathcal{H}_{\omega}-z)^{-1}P_{\omega,loc} \big),
\end{equation}
where $\mathcal{M}_{\omega}:=[\mathcal{H}_{\omega},x_1]x_2-[\mathcal{H}_{\omega},x_2]x_1$ is the orbital angular momentum operator (associated with $\mathcal{H}_{\omega}$), $\tilde{\rho}$ is the almost-analytic extension of $\rho$ (see \eqref{eq_almost_analyticity}), and the projection $P_{\omega,loc}:=\mathbbm{1}_{\mathcal{I}^{\prime} }(\mathcal{H}_{\omega})$. On the other hand,
\begin{equation} \label{eq_correction_term_2}
\begin{aligned}
\mathcal{E}_{edge,\omega,\mathcal{I},\mathcal{I}^{\prime},\rho,L}^{cor,(2)}
&:=\frac{i}{2|\Lambda_L|} \text{Tr}_{\Lambda_L}\Big( P_{\omega,loc} \big( [\rho(\mathcal{H}_{\omega}),x_1]x_2 + \int d\lambda 
\rho^{\prime}(\lambda) T_{\lambda,12} \big) P_{\omega,loc}  \Big) \\
&\quad -\frac{i}{2|\Lambda_L|}\text{Tr}_{\Lambda_L}\Big( P_{\omega,loc} \big( [\rho(\mathcal{H}_{\omega}),x_2]x_1 + \int d\lambda 
\rho^{\prime}(\lambda) T_{\lambda,21} \big) P_{\omega,loc}  \Big),
\end{aligned}
\end{equation}
where
\begin{equation*}
T_{\lambda,ij}:=P_{\omega,\lambda}x_i P_{\omega,\lambda}^{\perp} x_jP_{\omega,\lambda}-P_{\omega,\lambda}^{\perp}x_iP_{\omega,\lambda} x_jP_{\omega,\lambda}^{\perp} .
\end{equation*}
\end{definition}
By the almost-analyticity of $\tilde{\rho}$, the complex integral in \eqref{eq_correction_term_1} converges absolutely (cf. \cite[Theorem 14.8]{zworski2012semiclassical}); hence the first correction term $\mathcal{E}_{edge,\omega,\mathcal{I},\mathcal{I}^{\prime},\rho,L}^{cor,(1)}$ is well-defined. On the other hand, for a.e. $\omega$, one can check that the operator $T_{\lambda,ij}$ is uniformly bounded in $\lambda$ due to the exponential localization of $P_{\omega,\lambda}^{\perp}$ and $P_{\omega,\lambda}$ (Proposition \ref{prop_localization_ground_state}) and recalling the Schur test; hence the $\lambda$-integral in \eqref{eq_correction_term_2} absolutely converges and the second correction term $\mathcal{E}_{edge,\omega,\mathcal{I},\mathcal{I}^{\prime},\rho,L}^{cor,(2)}$ is also well-defined.

With the refined edge index, we are now ready to present the main result of this paper, which establishes the bulk-edge correspondence for finite disordered systems.
\begin{theorem}[BEC principle for finite disordered systems] \label{thm_bec_disorder}
Under Hypothesis \ref{hypo_anderson_localization}, the following identity holds $\mathbb{P}-$almost surely
\begin{equation*}
\lim_{L\to\infty}\mathcal{E}_{edge,\omega,\mathcal{I},\mathcal{I}^{\prime},\rho,L}^{ref}=\mathcal{E}_{bulk,\omega}\in\mathbb{Z} .
\end{equation*}
\end{theorem}

The proof is presented in Section \ref{sec_bec}. It generalizes Theorem \ref{thm_bec_non_random_finite} to the disordered regime. From a practical perspective, Theorem \ref{thm_bec_disorder} provides a rigorous foundation for the BEC principle for finite disordered systems and we hope it could support emerging theories of two-dimensional topological alloy-type materials; see \cite{chan2024topological_alloy} and the references therein. We also note that by slightly modifying the proof in this paper, we can prove Theorem \ref{thm_bec_disorder} for simply-connected domains with more general shapes other than the simple box $\Lambda_L=[-L,L]\times [-L,L]$ considered in this paper. This is left to the interested reader.

\begin{remark} \label{rmk_illustrate_refined_index}
It is necessary to illustrate the correction terms in the refined edge index \eqref{eq_refined_bulk_index}. First, both $\mathcal{E}_{edge,\omega,\mathcal{I},\mathcal{I}^{\prime},\rho,L}^{cor,(1)}$ and $\mathcal{E}_{edge,\omega,\mathcal{I},\mathcal{I}^{\prime},\rho,L}^{cor,(2)}$ vanish {if the mobility gap is actually a spectral gap}: in this case, one just takes a switch function $\rho$ with $\rho^{\prime}$ supported in the spectral gap and then takes $\mathcal{I}^{\prime}\cap \sigma(\mathcal{H}_{\omega})=\emptyset$. As a consequence,
\begin{equation*}
P_{\omega,loc}=\mathbbm{1}_{\mathcal{I}^{\prime} }(\mathcal{H}_{\omega})=0.
\end{equation*}
Hence, $\mathcal{E}_{edge,\omega,\mathcal{I},\mathcal{I}^{\prime},\rho,L}^{cor,(1)}=\mathcal{E}_{edge,\omega,\mathcal{I},\mathcal{I}^{\prime},\rho,L}^{cor,(2)}=0$ by definition. Now, we examine the first correction term $\mathcal{E}_{edge,\omega,\mathcal{I},\mathcal{I}^{\prime},\rho,L}^{cor,(1)}$, with an emphasis on its physical meaning. In fact, as we see in Theorem \ref{thm_bec_disorder}, the quantity of our interests is its limiting value, i.e.,
\begin{equation} \label{eq_illustration_extra_term_1}
\begin{aligned}
&\lim_{L\to\infty}\frac{i}{2\pi |\Lambda_L|}\text{Tr}_{\Lambda_L}\big( P_{\omega,loc} \int_{\mathbb{C}}dm(z)\partial_{\overline{z}}\tilde{\rho}(z) (\mathcal{H}_{\omega}-z)^{-1} \mathcal{M}_{\omega} (\mathcal{H}_{\omega}-z)^{-1}P_{\omega,loc} \big) \\
&=\frac{i}{2\pi}  \mathcal{T} \big( P_{\omega,loc}\int_{\mathbb{C}}dm(z)\partial_{\overline{z}}\tilde{\rho}(z) (\mathcal{H}_{\omega}-z)^{-1} \mathcal{M}_{\omega} (\mathcal{H}_{\omega}-z)^{-1}P_{\omega,loc} \big) .
\end{aligned}
\end{equation}
Next, noting that the resolvent $(\mathcal{H}_{\omega}-z)^{-1}$ commutes with the projections $P_{\omega,loc}$, we {formally} move the resolvent $(\mathcal{H}_{\omega}-z)^{-1}$ cyclically to the right side and see that
\begin{equation*}
\begin{aligned}
& \mathcal{T} \big( P_{\omega,loc} \int_{\mathbb{C}}dm(z)\partial_{\overline{z}}\tilde{\rho}(z) (\mathcal{H}_{\omega}-z)^{-1} \mathcal{M}_{\omega} (\mathcal{H}_{\omega}-z)^{-1}P_{\omega,loc} \big) \\
&\overset{(i)}{=} \mathcal{T} \big( P_{\omega,loc} \int_{\mathbb{C}}dm(z)\partial_{\overline{z}}\tilde{\rho}(z)  \mathcal{M}_{\omega} (\mathcal{H}_{\omega}-z)^{-2}P_{\omega,loc} \big) \\
&\overset{(ii)}{=} -\pi i\cdot  \mathcal{T} \big( P_{\omega,loc} \mathcal{M}_{\omega} \rho^{\prime}(\mathcal{H}_{\omega}) P_{\omega,loc} \big) \\
&\overset{(iii)}{=} -\pi i \cdot \mathcal{T} \big( \mathcal{M}_{\omega} \rho^{\prime}(\mathcal{H}_{\omega} )P_{\omega,loc} \big) \\
&= -\pi i  \cdot \mathcal{T} \big(  \mathcal{M}_{\omega} \rho^{\prime}(\mathcal{H}_{\omega})\big),
\end{aligned}
\end{equation*}
where the cyclicity of TPUV is applied informally in (i) and (iii), and the identity (ii) follows from the Helffer-Sjöstrand formula. With the above calculation, we see that the limit of the correction term $\mathcal{E}_{edge,\omega,\mathcal{I},\mathcal{I}^{\prime},\rho,L}^{cor,(1)}$ in the refined edge index \eqref{eq_refined_bulk_index} is essentially
\begin{equation} \label{eq_illustration_extra_term_2}
\mathcal{T} \big(  \mathcal{M}_{\omega} \rho^{\prime}(\mathcal{H}_{\omega})\big),
\end{equation}
which measures the expectation of orbital angular momentum associated with the bulk-localized modes in the mobility gap, weighted by the function $\rho^{\prime}$. By subtracting it from the original edge index $\mathcal{E}_{edge,\omega,\mathcal{I},\rho,L}$, we obtain a quantity that only measures the circulation of edge modes.

However, we note that the above derivation is non-rigorous: the applications of cyclicity in (i) and (iii) are not justified as the orbital angular momentum operator $\mathcal{M}_{\omega}$ is not covariant with respect to space translation. For the same reason, we cannot justify that the TPUV in \eqref{eq_illustration_extra_term_2} is well-defined. This explains why we do not use the more concise expression \eqref{eq_illustration_extra_term_2} in Definition \ref{prop_refined_bulk_index}.

Next, we look at the other correction term $\mathcal{E}_{edge,\omega,\mathcal{I},\mathcal{I}^{\prime},\rho,L}^{cor,(2)}$. In fact, we anticipate that it converges to zero as $L\to\infty$. Proving this is very hard. The reason for which we expect $\mathcal{E}_{edge,\omega,\mathcal{I},\mathcal{I}^{\prime},\rho,L}^{cor,(2)}$ to vanish asymptotically can be illustrated as follows: Look at the formal limit of the first line of \eqref{eq_correction_term_2}, i.e.,
\begin{equation*} 
\mathcal{T}\Big( P_{\omega,loc} \big( [\rho(\mathcal{H}_{\omega}),x_1]x_2 + \int d\lambda 
\rho^{\prime}(\lambda) T_{\lambda,12} \big) P_{\omega,loc}  \Big).
\end{equation*}
Consider the extreme case that $\rho=\mathbbm{1}_{(-\infty,\lambda_0)}$ for some $\lambda_0\in \mathcal{I}$, and hence $\rho^{\prime}$ becomes the Dirac mass supported at $\lambda_0$. (This is, again, informal because we have assumed that $\rho\in C^{\infty}$.) In that case, we have
\begin{equation} \label{eq_illustration_extra_term_3}
\begin{aligned}
&\mathcal{T}\Big( P_{\omega,loc} \big( [\rho(\mathcal{H}_{\omega}),x_1]x_2 + \int d\lambda 
\rho^{\prime}(\lambda) T_{\lambda,12} \big) P_{\omega,loc}  \Big) \\
&=\mathcal{T}\Big( P_{\omega,loc} \big( [P_{\omega,\lambda_0},x_1]x_2 -  T_{\lambda_0,12} \big) P_{\omega,loc}  \Big) .
\end{aligned}
\end{equation}
Using the definition of $T_{\lambda,12}$, we see that
\begin{equation*}
\begin{aligned}
\mathcal{T}\big( P_{\omega,loc}  T_{\lambda_0,12}  P_{\omega,loc}  \big)
&=\mathcal{T}\Big( P_{\omega,loc}  \big( P_{\omega,\lambda_0}x_1 P_{\omega,\lambda_0}^{\perp} x_2P_{\omega,\lambda_0}-P_{\omega,\lambda_0}^{\perp}x_1P_{\omega,\lambda_0} x_2P_{\omega,\lambda_0}^{\perp} \big)  P_{\omega,loc}  \Big) \\
&=\mathcal{T}\Big( P_{\omega,loc}  \big( -P_{\omega,\lambda_0} [P_{\omega,\lambda_0}^{\perp},x_1] x_2P_{\omega,\lambda_0}+P_{\omega,\lambda_0}^{\perp}[P_{\omega,\lambda_0},x_1] x_2P_{\omega,\lambda_0}^{\perp} \big)  P_{\omega,loc}  \Big) \\
&=\mathcal{T}\Big( P_{\omega,loc}  \big( P_{\omega,\lambda_0} [P_{\omega,\lambda_0},x_1] x_2P_{\omega,\lambda_0}+P_{\omega,\lambda_0}^{\perp}[P_{\omega,\lambda_0},x_1] x_2P_{\omega,\lambda_0}^{\perp} \big)  P_{\omega,loc}  \Big) \\
&\overset{(i)}{=}\mathcal{T}\Big( P_{\omega,loc}\big( P_{\omega,\lambda_0}+P_{\omega,\lambda_0}^{\perp}\big)  [P_{\omega,\lambda_0},x_1] x_2   P_{\omega,loc}\Big) \\
&=\mathcal{T}\Big( P_{\omega,loc}  [P_{\omega,\lambda_0},x_1] x_2   P_{\omega,loc}\Big),
\end{aligned}
\end{equation*}
where the cyclicity is applied informally to derive (i). (This is also not justified because the operator $[P_{\omega,\lambda_0},x_1] x_2$ is not covariant.) As a consequence, the right side of \eqref{eq_illustration_extra_term_3} vanishes. Nevertheless, as this derivation is non-rigorous, we keep the refined index as in \eqref{eq_refined_bulk_index}.
\end{remark}

For completeness of this paper, we prove the validity of Hypothesis \ref{hypo_anderson_localization} for a specific class of random Hamiltonians, where the on-site potential has the ability of opening up a spectral gap\footnote{This is indeed a classical scenario in experimental studies of topological alloy-type materials; see \cite{chan2024topological_alloy} and the references therein.} and are described as follows.
\begin{assumption} \label{asmp_Dirac_eigenpair}
Consider the periodic Hamiltonian defined by $\mathcal{H}_{per}=\mathcal{H}_0+1\otimes V\in \mathcal{B}(\mathcal{X})$, with $\mathcal{H}_0$ and $V$ being introduced in Assumptions \ref{asmp_Hamiltonian_deter_part} and \ref{asmp_random_potential}, respectively. We assume that $\mathcal{H}_{per}$ has a band gap in the sense that
\begin{equation*}
(\lambda_0-\Delta,\lambda_0+\Delta)\cap \sigma(\mathcal{H}_{per})=\emptyset,
\end{equation*}
for some $\lambda_0\in \mathbb{R}$ and $\Delta>0$. Moreover, we assume that the amplitude of the potential satisfies
\begin{equation*}
\sup_{1\leq i\leq d}|V_i|\leq \Delta .
\end{equation*}
\end{assumption}
As shown in Appendix \ref{appendixA}, a typical model that satisfies Assumptions \ref{asmp_Hamiltonian_deter_part} and \ref{asmp_Dirac_eigenpair} is the famous Qi–Wu–Zhang (QWZ) model, which is a prototypical Hamiltonian for (quantum) anomalous Hall physics \cite{qi2006qwz_model,ezawa2024nonlinear_phase_transition}. For the random variable $\omega_{\bm{n}}^{(k)}$, we assume that their common density satisfies the following. 
\begin{assumption} \label{asmp_distribution_density}
Assume that $\text{supp}(\mu)=[0,1]$, and that there exist $\tau,C>0,\eta >5$ such that
\begin{equation} \label{eq_poly_distribution}
\int_{0}^{t}d\mu\leq Ct^{\eta}\quad (\forall t\in [0,\tau]).
\end{equation}
\end{assumption}
Then the existence of a mobility gap is justified as follows. The proof is presented in Section \ref{sec_anderson_loc}.
\begin{theorem}
\label{thm_anderson_localization}
Under Assumptions \ref{asmp_Dirac_eigenpair} and \ref{asmp_distribution_density}, there exists an interval $\mathcal{I}\subset (\lambda_0-\Delta,\lambda_0+\Delta)$ such that it is an Aizenman-Molchanov mobility gap in the sense of Hypothesis \ref{hypo_anderson_localization}.
\end{theorem}

Finally, in Appendix \ref{appendixB}, we prove that the bulk index $\mathcal{E}_{bulk,\omega,\lambda}$ is equal to the first Chern number associated with the spectral projection $\mathbbm{1}_{(-\infty,\lambda_0)}(\mathcal{H}_{per})$ of the deterministic periodic Hamiltonian. This indicates that the identity in Theorem \ref{thm_bec_disorder} is nontrivial and therefore completes our findings in this paper.

\begin{remark}
Assumption \ref{asmp_distribution_density}, especially the polynomial decay of the density near the origin, as supposed in \eqref{eq_poly_distribution}, is far from being optimal to prove Theorem \ref{thm_anderson_localization}. We note that \eqref{eq_poly_distribution} is posed to guarantee the Lifshitz tail estimate around $\lambda=\lambda_0$ (Proposition \ref{prop_lifshitz_tail}), which is the key element to establish the existence of mobility gap; see Section \ref{sec_lif_initial} for details. We note that, for a general random distribution, the Lifshitz tail estimate is extensively studied and verified for random Schrödinger Hamiltonians, whose proof relies heavily on the special structure of their deterministic part (i.e., the Laplacian operator) such as by the Dirichlet-Neumann bracketing method; see \cite{stolz2011introduction,kirsch2007invitation,stollmann2001caught} and the references therein. However, the Hamiltonian considered in this paper is general (e.g. the QWZ Hamiltonian), for which the Dirichlet-Neumann bracketing may not work. This is the reason for assuming \eqref{eq_poly_distribution}; see Remark \ref{rmk_beta_size} for a more detailed discussion. We note that similar assumptions on the fast decay of the distribution of the random variables near their extreme values have been assumed to study the Anderson localization \cite{barbaroux1997localization,stolz1998anderson}; see also the discussion in \cite{seelmann2020band}. On the other hand, we point out that the Lifshitz tail estimate is often assumed in the study of Anderson localization of general stochastic Hamiltonians (see, e.g. \cite{aizenman2001finite_volime}). It would be very interesting if one could prove Proposition \ref{prop_lifshitz_tail} without imposing \eqref{eq_poly_distribution}.
\end{remark}

\begin{remark}
We note that, in addition to being used for studying the BEC principle, Theorem \ref{thm_anderson_localization} is of its own interest in the study of Anderson localization. In fact,  the previous proofs of Anderson localization focus mainly on the case of large disorder or near band gaps of the unperturbed operator.\footnote{It is almost impossible to list all the literature in this large field of study; We refer the reader to the excellent monographs \cite{stolz2011introduction,kirsch2007invitation} for a review and \cite{aizenman1993localization_elementary_derivation,aizenman1994localization,aizenman1998localization,li2022anderson,ding2020localization,najar2006band_edge,combes1994localization,veselic2002localization,klopp2003note_anderson,nakamura2003anderson,prado2017dynamical,prado2021density,carmona1987anderson,chulaevsky2023anderson,cornean2019gapped_dirac}; see also the references therein.} In our case, there can be no gap for the unperturbed Hamiltonian $\mathcal{H}_{0}$. To the best of our knowledge, the present proof of Anderson localization under a random gap opening potential is new. In fact, as shown in Section \ref{sec_anderson_loc}, after establishing the Lifshitz-type estimate (Proposition \ref{prop_lifshitz_tail}) and an initial length estimate (Proposition \ref{prop_initial_length}), the proof of Theorem \ref{thm_anderson_localization} is completed by following the standard geometric decoupling approach.
\end{remark}

\section{Notation}
\noindent 1) Geometries: $\Lambda_L=[-L,L]\times [-L,L]\cap \mathbb{Z}^2$ (finite boxes), $\partial^{int}\Lambda_L$, $\partial^{ext}\Lambda_L$ and $\partial\Lambda_L$ (interior, exterior, and total boundary of $\Lambda_L$) are defined as
\begin{equation*}
    \partial^{int}\Lambda_L:=\{\bm{n}\in \mathbb{Z}^2:\, |n_1|=L\text{ or }|n_2|=L\},\quad \partial^{ext}\Lambda_L:=\{\bm{n}\in \mathbb{Z}^2:\, |n_1|=L+1\text{ or }|n_2|=L+1\},
\end{equation*}
and
\begin{equation*}
     \partial\Lambda_L:=\big\{(\bm{n},\bm{m})\in (\partial^{int}\Lambda_L\times \partial^{ext}\Lambda_L)\cup (\partial^{ext}\Lambda_L\times \partial^{int}\Lambda_L):\, |\bm{n}-\bm{m}|=1 \big\}.
\end{equation*}

\noindent 2) Hilbert spaces $\mathcal{X}_{L}=\ell^2(\Lambda_L)\times \mathbb{C}^d$, $\mathcal{X}=\ell^2(\mathbb{Z}^2)\times \mathbb{C}^d$.

\noindent 3) Hamiltonians: $\mathcal{H}_0$ (deterministic Hamiltonian; see Assumption \ref{asmp_Hamiltonian_deter_part}), $\mathcal{H}_{per}$ (periodically perturbed Hamiltonian; see Assumption \ref{asmp_Dirac_eigenpair}), $\mathcal{H}_{\omega}$ (random Hamiltonian), $\mathcal{H}_{\omega,L}^{sim}$ (random box Hamiltonian with simple boundary conditions), $\mathcal{H}_{\omega,L}^{\sharp}$ (random box Hamiltonian with periodic boundary conditions; see Section \ref{sec_lif_initial}).

\noindent 4) Norms: $\|\cdot\|=\|\cdot \|_{F}$ (Frobenius norm on $\mathbb{C}^{d\times d}$), $|\cdot|=\|\cdot\|_1$ ($\ell^1$ norm on $\mathbb{Z}^{2}$), $\|\cdot \|_{\mathcal{X}/\mathcal{X}_L}$ ($\ell^2$ norm on $\mathcal{X}/\mathcal{X}_L$, respectively), $\|\cdot\|_{\mathscr{T}_1}$ (trace norm on $\mathcal{X}$). With little abuse of notation, the operator norm on $\mathcal{X}/\mathcal{X}_L$ is also denoted by $\|\cdot \|_{\mathcal{X}/\mathcal{X}_L}$, respectively.

\noindent 5) Brackets: $\langle\cdot,\cdot\rangle_{\mathcal{V}}$ denotes the inner product of a Hilbert space $\mathcal{V}$. With a little abuse of notation, we will sometimes write $A(\bm{n},\bm{m})=\langle\mathbbm{1}_{\{\bm{n}\}},A\mathbbm{1}_{\{\bm{m}\}}\rangle_{\ell^2(\mathbb{Z}^d)}\in \mathbb{C}^{d\times d}$ as the kernel of operator $A\in\mathcal{B}(\mathcal{X})$. 

\noindent 6) Traces: Tr (the usual trace on the Hilbert space $\mathcal{X}$), $$\text{Tr}_{\Lambda_L}(\cdot)=\text{Tr}(\mathbbm{1}_{\Lambda_L}\cdot\mathbbm{1}_{\Lambda_L}),$$
\begin{equation} \label{def:t}
    \mathcal{T}(\cdot)=\lim_{L\to\infty}\frac{1}{|\Lambda_L|}\text{Tr}_{\Lambda_L}(\cdot)
\end{equation}
 (the trace per unit of volume, TPUV), and tr (the trace of $\mathbb{C}^{d\times d}$ matrices).

\section{Preliminaries} \label{sec_prelim}

\subsection{Some Algebraic Properties of Exponentially Localized Operators} \label{sec_trace_class}

In this section, we review some algebraic properties of exponentially localized operators, including composition and cyclicity properties. The proofs of Lemma \ref{lem_resolv_CT_estimate}-\ref{lem_composition_kernel} are found in \cite[Section 2]{drouot2024bec_curvedinterfaces}.

\begin{lemma} \label{lem_resolv_CT_estimate}
For any bounded self-adjoint operator $\mathcal{H}\in \mathcal{B}(\mathcal{X})$, its resolvent satisfies the Combes-Thomas estimate. Namely, there exists $C>0$ such that
\begin{equation*}
\|(\mathcal{H}-z)^{-1}(\bm{n},\bm{m})\|\leq \frac{C}{\text{dist}(z,\sigma(\mathcal{H}))}e^{-\text{dist}(z,\sigma(\mathcal{H}))|\bm{n}-\bm{m}|}
\end{equation*}
for any $\bm{n},\bm{m}\in\mathbb{Z}^2$ and $z\notin \sigma(\mathcal{H})$.
\end{lemma}

\begin{definition}[Exponentially localized operator] \label{def_exp_loc_operator}
The operator $A\in\mathcal{B}(\mathcal{X})$ is said to be exponentially localized if its integral kernel satisfies
\begin{equation*}
\|A(\bm{n},\bm{m})\|\leq Ce^{-p|\bm{n}-\bm{m}|}
\end{equation*}
with some $p>0$.
\end{definition}

\begin{lemma} \label{lem_position_commutator}
Let $A\in\mathcal{B}(\mathcal{X})$ have exponentially localized kernels in the sense of Definition \ref{def_exp_loc_operator}. Then the commutator $[A,x_i]$ ($i=1,2$) is also exponentially localized with the estimate
\begin{equation*}
\big\|[A,x_i](\bm{n},\bm{m})\big\|\leq \frac{C}{|p|}e^{-\frac{p}{2}|\bm{n}-\bm{m}|},
\end{equation*}
where $C>0$ is independent of $p$.
\end{lemma}

\begin{lemma} \label{lem_composition_kernel}
Let $A_i\in\mathcal{B}(\mathcal{X})$ ($i=1,2,\cdots,n$) have exponentially localized kernels in the sense of Definition \ref{def_exp_loc_operator}. Then the kernel of the composition operator $\Pi_{i=1}^{n}A_i$ is also exponentially localized with the estimate
\begin{equation*}
\big\|(\Pi_{i=1}^{n}A_i)(\bm{n},\bm{m})\big\|\leq \frac{C}{|p|^{2(n-1)}}e^{-\frac{p}{4}|\bm{n}-\bm{m}|},
\end{equation*}
where $C>0$ is independent of $p$.
\end{lemma}

\begin{lemma} \label{lem_finite_volume_cyclicity}
Suppose that $A,B\in\mathcal{B}(\mathcal{X})$ have exponentially localized kernels in the sense of Definition \ref{def_exp_loc_operator}. Then, for any $L>0$, it holds that
\begin{equation} \label{eq_finite_volume_cyclicity}
\big|\text{Tr}_{\Lambda_L}(AB-BA)\big|\leq \frac{C}{|p|^2}\Big(L^{\frac{3}{2}}+L^2e^{-\frac{p}{2}L^{\frac{1}{2}}}\Big)
\end{equation}
with $C>0$ being independent of $L$ and $p$. Consequently, the cyclicity of 
the TPUV holds in the sense that
\begin{equation*}
\mathcal{T}(AB-BA)=\lim_{L\to\infty}\frac{1}{|\Lambda_L|}\text{Tr}_{\Lambda_L}(AB-BA)=0.
\end{equation*}
\end{lemma}

\begin{proof}[Proof of Lemma \ref{lem_finite_volume_cyclicity}]
Note that, since $\Lambda_L$ is finite-dimensional, both $\mathbbm{1}_{\Lambda_L}AB\mathbbm{1}_{\Lambda_L}$ and $\mathbbm{1}_{\Lambda_L}BA\mathbbm{1}_{\Lambda_L}$ are trace-class by Definition \ref{def_exp_loc_operator}, which allows us to calculate their traces in position basis
\begin{equation*}
\begin{aligned}
\text{Tr}_{\Lambda_L}(AB-BA)&=\text{Tr}_{\Lambda_L}(AB)-\text{Tr}_{\Lambda_L}(BA) \\
&=\text{tr}\sum_{\bm{n}\in\Lambda_L}\sum_{\bm{m}\in\Lambda_L^{c}}A(\bm{n},\bm{m})B(\bm{m},\bm{n})-\text{tr}\sum_{\bm{n}\in\Lambda_L}\sum_{\bm{m}\in\Lambda_L^{c}}B(\bm{n},\bm{m})A(\bm{m},\bm{n}).
\end{aligned}
\end{equation*}
Hence, it suffices to show that both of these two sums are bounded by the right side of \eqref{eq_finite_volume_cyclicity}. We only prove it for the first sum, while the other one can be treated similarly. The idea is to decompose the sum over $\bm{m}$ into two regions:
\begin{equation} \label{eq_finite_volume_cyclicity_proof_1}
\begin{aligned}
\big| \text{tr}\sum_{\bm{n}\in\Lambda_L}\sum_{\bm{m}\in\Lambda_L^{c}}A(\bm{n},\bm{m})B(\bm{m},\bm{n}) \big|
&\leq d\sum_{\bm{n}\in\Lambda_L}\sum_{\bm{m}\in\Lambda_L^{c}}\|A(\bm{n},\bm{m})\| \|B(\bm{m},\bm{n})\| \\
&\leq C\sum_{\bm{n}\in\Lambda_L}\sum_{\bm{m}\in\Lambda_L^{c}}e^{-p|\bm{n}-\bm{m}|} \\
&= C\Big(\sum_{\bm{n}\in\Lambda_L}\sum_{\bm{m}\in \Lambda_{L+L^{1/2}}\cap\Lambda_L^{c}}e^{-p|\bm{n}-\bm{m}|} + \sum_{\bm{n}\in\Lambda_L}\sum_{\bm{m}\in (\Lambda_{L+L^{1/2}})^c}e^{-p|\bm{n}-\bm{m}|} \Big).
\end{aligned}
\end{equation}
For the first sum, we first sum over $\bm{n}$ and utilize the elementary bound
\begin{equation} \label{eq_finite_volume_cyclicity_proof_2}
\sum_{\bm{m}\in\mathbb{Z}^2}e^{-p|\bm{n}-\bm{m}|}\leq \frac{C}{|p|^2} \quad \text{for all }\bm{n}\in\mathbb{Z}^2.
\end{equation}
Hence, we see that
\begin{equation} \label{eq_finite_volume_cyclicity_proof_3}
\sum_{\bm{n}\in\Lambda_L}\sum_{\bm{m}\in \Lambda_{L+L^{1/2}}\cap\Lambda_L^{c}}e^{-p|\bm{n}-\bm{m}|}\leq \sum_{\bm{m}\in \Lambda_{L+L^{1/2}}\cap\Lambda_L^{c}} \frac{C}{|p|^2} \leq \frac{C}{|p|^2}L^{\frac{3}{2}}.
\end{equation}
For the second sum in \eqref{eq_finite_volume_cyclicity_proof_1}, we note that, for all $\bm{n}\in\Lambda_L$ and $\bm{m}\in (\Lambda_{L+L^{1/2}})^c$,
\begin{equation*}
\begin{array}{lll}
\ds
e^{-p|\bm{n}-\bm{m}|} &\leq&\ds  e^{-\frac{p}{2}\text{dist}(\bm{m},\Lambda_L)} e^{-\frac{p}{2}|\bm{n}-\bm{m}|} \\
\nm
&\leq & \ds  e^{-\frac{p}{2}L^{\frac{1}{2}}}e^{-\frac{p}{2}|\bm{n}-\bm{m}|}.
\end{array}
\end{equation*}
Then we perform the calculation as in the first part (with now summing over $\bm{m}$ first) and obtain 
\begin{equation} \label{eq_finite_volume_cyclicity_proof_4}
\begin{array}{lll}
\ds 
\sum_{\bm{n}\in\Lambda_L}\sum_{\bm{m}\in (\Lambda_{L+L^{1/2}})^c}e^{-p|\bm{n}-\bm{m}|}
&\leq & \ds e^{-\frac{p}{2}L^{\frac{1}{2}}}\sum_{\bm{n}\in\Lambda_L}\sum_{\bm{m}\in (\Lambda_{L+L^{1/2}})^c}e^{-\frac{p}{2}|\bm{n}-\bm{m}|}\\
\nm
&\leq&\ds  e^{-\frac{p}{2}L^{\frac{1}{2}}}\sum_{\bm{n}\in\Lambda_L}\frac{C}{|p|^2}\\
\nm
&\leq&\ds  \frac{C}{|p|^2} L^2e^{-\frac{p}{2}L^{\frac{1}{2}}}.
\end{array}
\end{equation}
Hence, the proof is complete by combining \eqref{eq_finite_volume_cyclicity_proof_1}, \eqref{eq_finite_volume_cyclicity_proof_3} and \eqref{eq_finite_volume_cyclicity_proof_4}.
\end{proof}

We also need to deal with exponentially localized kernels in a weaker sense, as being already seen in Proposition \ref{prop_localization_ground_state}. Specifically, we introduce the following definition. 
\begin{definition}
$A\in \mathcal{B}(\mathcal{X})$ is said to be $(\nu,\alpha,D)$ localized if
\begin{equation*}
\sum_{\bm{n},\bm{m}\in\mathbb{Z}^2}\|A(\bm{n},\bm{m})\|(1+|n|^{-\nu})e^{\alpha |\bm{m}-\bm{n}|}\leq D <\infty.
\end{equation*}
\end{definition}
\begin{lemma} \label{lem_composition_weak_localized}
Suppose that $A_i\in\mathcal{B}(\mathcal{X})$ is $(\nu_i,\alpha_i,D_i)$ localized ($i=1,2$). Then the composition $A_1A_2$ is $(\nu,\alpha,D)$ localized, where
\begin{equation*}
\nu=\max\{2\mu_1,2\mu_2\},\quad \alpha=\min\{\alpha_1/2,\alpha_2/2,1\},\quad D=\frac{CD_1D_2}{\alpha^{\frac{\gamma}{2}}}
\end{equation*}
with $C>0$ depending only on $\nu$.
\end{lemma}
\begin{proof}
For $\nu>0$ and $\alpha\in (0,1)$, we note that
\begin{equation*}
\begin{aligned}
&\sum_{\bm{n},\bm{k}\in\mathbb{Z}^2}(1+|\bm{n}|)^{-\nu}e^{\alpha|\bm{n-\bm{k}}|}\big\|\sum_{\bm{m}\in\mathbb{Z}^2}A_1(\bm{n},\bm{m})A_2(\bm{m},\bm{k})\big\| \\
&\leq \frac{C}{\alpha^{\frac{\nu}{2}}}\sum_{\bm{n},\bm{k}\in\mathbb{Z}^2}(1+|\bm{n}|)^{-\frac{\nu}{2}}(1+|\bm{k}|)^{-\frac{\nu}{2}}e^{2\alpha|\bm{n-\bm{k}}|}\big\|\sum_{\bm{m}\in\mathbb{Z}^2} A_1(\bm{n},\bm{m})A_2(\bm{m},\bm{k})\big\| \\
&\leq \frac{C}{\alpha^{\frac{\nu}{2}}}\sum_{\bm{n},\bm{k}\in\mathbb{Z}^2}(1+|\bm{n}|)^{-\frac{\nu}{2}}(1+|\bm{k}|)^{-\frac{\nu}{2}}e^{2\alpha|\bm{n-\bm{m}}|}e^{2\alpha|\bm{m-\bm{k}}|}\sum_{\bm{m}\in\mathbb{Z}^2} \|A_1(\bm{n},\bm{m})\| \|A_2(\bm{m},\bm{k})\|,
\end{aligned}
\end{equation*}
where the first inequality is derived by noticing that
\begin{equation*}
(1+|\bm{n}|)^{-\frac{\nu}{2}}e^{-\frac{\alpha}{2}|\bm{n}-\bm{k}|}\leq \frac{C}{\alpha^{\frac{\nu}{2}}}(1+|\bm{k}|)^{-\frac{\nu}{2}}
\end{equation*}
with $C>0$ depending only on $\gamma$, and the second follows from the triangle inequality. Now, we apply the Schwarz inequality and see that
\begin{equation*}
\begin{aligned}
&\sum_{\bm{n},\bm{k}\in\mathbb{Z}^2}(1+|\bm{n}|)^{-\nu}e^{\alpha|\bm{n-\bm{k}}|}\big\|\sum_{\bm{m}\in\mathbb{Z}^2}A_1(\bm{n},\bm{m})A_2(\bm{m},\bm{k})\big\| \\
&\leq \frac{C}{\alpha^{\frac{\gamma}{2}}}\sum_{\bm{n},\bm{k}\in\mathbb{Z}^2}(1+|\bm{n}|)^{-\frac{\nu}{2}}(1+|\bm{k}|)^{-\frac{\nu}{2}}
\Big(\sum_{\bm{m}\in\mathbb{Z}^2} \|A_1(\bm{n},\bm{m})\|^2 e^{4\alpha|\bm{n-\bm{m}}|}\Big)^{\frac{1}{2}}
\Big(\sum_{\bm{m}\in\mathbb{Z}^2}\|A_2(\bm{m},\bm{k})\|^2 e^{4\alpha|\bm{m-\bm{k}}|}\Big)^{\frac{1}{2}} \\
&\leq \frac{C}{\alpha^{\frac{\gamma}{2}}}
\Big(\sum_{\bm{n},\bm{m}\in\mathbb{Z}^2} \|A_1(\bm{n},\bm{m})\| (1+|\bm{n}|)^{-\frac{\nu}{2}}e^{2\alpha|\bm{n-\bm{m}}|}\Big)
\Big(\sum_{\bm{m},\bm{k}\in\mathbb{Z}^2}\|A_2(\bm{m},\bm{k})\|(1+|\bm{k}|)^{-\frac{\nu}{2}} e^{2\alpha|\bm{m-\bm{k}}|}\Big),
\end{aligned}
\end{equation*}
where the second inequality follows from $(\sum|a_i|)^{\frac{1}{2}}\leq \sum|a_i|^{\frac{1}{2}}$. This concludes the proof by setting $\nu=\max\{2\mu_1,2\mu_2\}$, $\alpha=\min\{\alpha_1/2,\alpha_2/2,1\}$.
\end{proof}

\subsection{Covariant Operators and Trace Per Unit of Volume}
\label{sec_covariant_tpuv}

Since the random Hamiltonian $\mathcal{H}_{\omega}$ is ergodic, its covariant property corresponding to the spatial translation is of critical importance in the study of the BEC principle. Here, we briefly sketch some basic properties of covariant operators and refer the interested reader to \cite{Bellissard94non_commutative_QHE,Avron1994charge} and the references therein for more details.

\begin{definition}[Covariant operator]
\label{def_covariance}
$A_{\omega}\in\mathcal{B}(\mathcal{X})$ is said to be covariant if it is measurable in $\omega\in \Omega$ and satisfies
\begin{equation} \label{eq_ergodicity_1}
\mathcal{A}_{T_{\bm{k}}\omega}=\mathcal{U}_{\bm{k}}\mathcal{A}_{\omega}\mathcal{U}_{\bm{k}}^{*},
\end{equation}
where $\mathcal{U}_{\bm{k}}:\psi(\cdot)\mapsto \psi(\cdot -\bm{k})$ is the translation map on $\mathcal{X}$ and $T$ is the (ergodic) action of the group $\mathbb{Z}^2$ on the probability space.
\end{definition}
In our set-up, i) the random potential $V_{\omega}$, ii) Hamiltonian $\mathcal{H}_{\omega}$, iii) the operator $f(\mathcal{H}_{\omega})$ for all bounded measurable functions $f$, and iv) commutators $[\mathcal{H}_{\omega},x_i]$ are all covariant (cf. \cite[Lemma 4.5]{kirsch2007invitation} and \cite{Schulz-Baldes01homotopy}).

The following result follows directly from the Birkhoff ergodic theorem (cf. \cite[Theorem 4.2]{kirsch2007invitation}).
\begin{lemma} \label{lem_covariant_TPUV_expectation}
Suppose that $A_{\omega}\in\mathcal{B}(\mathcal{X})$ is covariant and satisfies
\begin{equation*}
    \mathbb{E}\big(|A_{\omega}(0,0)| \big)<\infty .
\end{equation*}
Then
\begin{equation*}
\mathcal{T}(A_{\omega})=\mathbb{E}(A_{\omega}(0,0))
\end{equation*}
holds $\mathbb{P}-$almost surely. Here, $\mathcal{T}$ is the trace per unit of volume defined in \eqref{def:t}.
\end{lemma}

We also need the following result that justifies the cyclicity of $\mathcal{T}$ on covariant operators (cf. \cite[Lemma 3.21]{Klein2005linear_disorder}).
\begin{lemma} \label{lem_covariant_TPUV_cyclicity}
Suppose that $A_{\omega,i}\in\mathcal{B}(\mathcal{X})$ are covariant and
\begin{equation*}
\mathbb{E}\big(\|A_{\omega,i}\cdot\mathbbm{1}_{\{\bm{n}=0\}}\|^2_{HS}\big)<\infty
\end{equation*}
for $i=1,2$, where $\|\cdot\|_{HS}$ denotes the Hilbert-Schmidt norm. Then
\begin{equation*}
    \mathcal{T}(A_{\omega,1}A_{\omega,2})=\mathcal{T}(A_{\omega,2}A_{\omega,1})
\end{equation*}
holds $\mathbb{P}-$almost surely.
\end{lemma}

\section{Bulk-edge Correspondence for Finite Disordered Systems: Proof of Theorem \ref{thm_bec_disorder}}
\label{sec_bec}

Throughout this section, we will write $\mathcal{H}_{\omega,L}:=\mathcal{H}_{\omega,L}^{\text{sim}}$ because only box Hamiltonians with simple boundary conditions are concerned. This convention is 
removed in Section \ref{sec_anderson_loc}, where other types of boundary condition are also taken into account.

We first present the main skeleton of the proof of Theorem \ref{thm_bec_disorder}, which proceeds in several steps, then left the technical details to the following subsections. We begin with the original edge index \eqref{eq_edge_index} and show that its limit is equal to the bulk index \eqref{eq_refined_bulk_index} modulo the correction term in \eqref{eq_refined_bulk_index}. The first step is to show the following equivalent formulation of edge index in the language of Helffer-Sjöstrand formula. That is, for any realization $\omega\in\Omega$, it holds that
\begin{equation} \label{eq_edge_index_HS_formula}
\mathcal{E}_{edge,\omega,\mathcal{I},\rho,L}=\frac{1}{2\pi|\Lambda_L|}\int_{\mathbb{C}}dm(z)\partial_{\overline{z}}\tilde{\rho}(z)\text{Tr}_{\Lambda_L}\big( S_{\omega,L,12}(z)-S_{\omega,L,21}(z)  \big)
\end{equation}
with $\tilde{\rho}\in C_{c}^{\infty}(\mathbb{Z})$ being an almost-analytic extension of $\rho$ (cf. \cite{zworski2012semiclassical} and \cite[Section 2.2]{qiu2025generalized}) satisfying
\begin{equation} \label{eq_almost_analyticity}
\tilde{\rho}|_{\mathbb{R}}=\rho,\quad \partial_{\overline{z}}\tilde{\rho}=\mathcal{O}(|\text{Im }z|^{\infty}),\quad
\text{supp}(\tilde{\rho})\subset \{z\in \mathbb{C}:\, |\Im z|\leq 1\} ,
\end{equation}
and
\begin{equation*}
S_{\omega,L,ij}(z):=(\mathcal{H}_{\omega,L}-z)^{-1} [\mathcal{H}_{\omega,L},x_i] (\mathcal{H}_{\omega,L}-z)^{-1} [\mathcal{H}_{\omega,L},x_j] (\mathcal{H}_{\omega,L}-z)^{-1} .
\end{equation*}
The proof is given in Section \ref{sec_edge_index_HS_formula}. We note that this equivalent formulation is more general. In fact, a similar Helffer-Sjöstrand formulation can be used to define a topological index related to the spin transport in quantum systems, and show that it recovers the well-known $\mathbb{Z}_2$ index when the total spin charge is conserved; see \cite{qiu2025generalized} and also the references therein.

The next step is to justify the replacement of the box Hamiltonian $\mathcal{H}_{\omega,L}$ with the bulk Hamiltonian $\mathcal{H}_{\omega}$ in \eqref{eq_edge_index_HS_formula} by taking the average over $|\Lambda_L|$ and letting $L\to\infty$. In other words, we will show that, for $i=1,j=2$ and $i=2,j=1$,
\begin{equation} \label{eq_edge_to_bulk_HS}
\lim_{L\to\infty}\frac{1}{|\Lambda_L|}\int_{\mathbb{C}}dm(z)\partial_{\overline{z}}\tilde{\rho}(z)\text{Tr}_{\Lambda_L}\big( S_{\omega,L,ij}(z) \big)
=\lim_{L\to\infty}\frac{1}{|\Lambda_L|}\int_{\mathbb{C}}dm(z)\partial_{\overline{z}}\tilde{\rho}(z)\text{Tr}_{\Lambda_L}\big( S_{\omega,ij}(z) \big)
\end{equation}
with
\begin{equation*}
S_{\omega,ij}(z):=(\mathcal{H}_{\omega}-z)^{-1} [\mathcal{H}_{\omega},x_i] (\mathcal{H}_{\omega}-z)^{-1} [\mathcal{H}_{\omega},x_j] (\mathcal{H}_{\omega}-z)^{-1} .
\end{equation*}
Let us illustrate \eqref{eq_edge_to_bulk_HS} intuitively. Since the on-site random potential commutes with the position operator, we see that\footnote{We note that the identity does not hold for a box Hamiltonian with periodic boundary conditions, which contains hopping across boundaries. This implies that one cannot have a bulk-edge correspondence by defining the edge index \eqref{eq_edge_index} for periodic box Hamiltonians. Physically, it is reminiscent of the fact that the periodic boundary conditions are `penetrable'; see the discussion in \cite{mario19proof,qiu2025bec_finite}.}
\begin{equation*}
\mathbbm{1}_{\Lambda_L}[\mathcal{H}_{\omega,L},x_i]\mathbbm{1}_{\Lambda_L}=\mathbbm{1}_{\Lambda_L}[\mathcal{H}_{0},x_i]\mathbbm{1}_{\Lambda_L}=\mathbbm{1}_{\Lambda_L}[\mathcal{H}_{\omega},x_i]\mathbbm{1}_{\Lambda_L}.
\end{equation*}
Hence, when calculating the trace on the box $\Lambda_L$, $S_{\omega,ij}(z)$ differs from $S_{\omega,ij,L}(z)$ only by the resolvents. Remarkably, this difference is estimated by the resolvent identity
\begin{equation} \label{eq_edge_to_bulk_HS_proof_1}
\begin{aligned}
&(\mathcal{H}_{\omega}-z)^{-1}(\bm{n},\bm{m})-(\mathcal{H}_{\omega,L}-z)^{-1}(\bm{n},\bm{m}) \\
&=\sum_{(\bm{u},\bm{u}^{\prime})\in \text{supp}(\mathcal{F}^{(L)})}(\mathcal{H}_{\omega,L}-z)^{-1}(\bm{n},\bm{u})\mathcal{F}^{(L)}(\bm{u},\bm{u}^{\prime})(\mathcal{H}_{\omega}-z)^{-1}(\bm{u}^{\prime},\bm{m}).
\end{aligned}
\end{equation}
Here, the operator $\mathcal{F}^{(L)}$ records the hoppings between $\Lambda_L$ and its exterior
\begin{equation*}
\mathcal{F}^{(L)}:=\mathcal{H}_{\omega}-\mathcal{H}_{\omega}^{(L)},
\end{equation*}
whose kernel is supported on the boundary $\partial\Lambda_L$. Importantly, for $\Im z\neq 0$, the exponential decay of resolvents implies that the difference \eqref{eq_edge_to_bulk_HS_proof_1} converges to zero as $L\to\infty$. This is the fundamental idea underlying equality \eqref{eq_edge_to_bulk_HS}. The details are left to Section \ref{sec_HS_formula_edge_to_bulk}.

Note that the right side of \eqref{eq_edge_to_bulk_HS} is already a bulk quantity, which depends only on $\mathcal{H}_{\omega}$. In the last step, we will show that it equals the bulk index modulo an extra contribution from the mobility gap, which is exactly the extra term in the refined edge index \eqref{eq_refined_bulk_index}. To be precise, we will prove the following identity, which holds $\mathbb{P}-$almost surely
\begin{equation*}
\begin{aligned}
&\lim_{L\to\infty}\frac{1}{2\pi|\Lambda_L|}\int_{\mathbb{C}}dm(z)\partial_{\overline{z}}\tilde{\rho}(z)\text{Tr}_{\Lambda_L}\big( S_{\omega,12}(z) \big)-\lim_{L\to\infty}\frac{1}{2\pi|\Lambda_L|}\int_{\mathbb{C}}dm(z)\partial_{\overline{z}}\tilde{\rho}(z)\text{Tr}_{\Lambda_L}\big( S_{\omega,21}(z) \big) \\
&-\lim_{L\to\infty}\big( \mathcal{E}_{edge,\omega,\mathcal{I},\mathcal{I}^{\prime},\rho,L}^{cor,(1)}+\mathcal{E}_{edge,\omega,\mathcal{I},\mathcal{I}^{\prime},\rho,L}^{cor,(2)} \big) \\
&=\mathcal{E}_{bulk,\omega}.
\end{aligned}
\end{equation*}
This involves technical manipulations of traces and complex integrals; see Section \ref{sec_link_refined_bulk_index} for details.

\subsection{Step 1: Equivalent Formulation of the Original Edge Index} \label{sec_edge_index_HS_formula}

The proof of \eqref{eq_edge_index_HS_formula} begins by noting the Helffer-Sjöstrand formula (cf. \cite[Theorem 14.8]{zworski2012semiclassical}) and integrating by parts,
\begin{equation*}
\rho^{\prime}(\mathcal{H}_{\omega,L})
=\frac{1}{\pi i}\int_{\mathbb{C}}dm(z)\partial_{\overline{z}}\tilde{\rho^{\prime}}(z)(\mathcal{H}_{\omega,L}-z)^{-1}
=-\frac{1}{\pi i}\int_{\mathbb{C}}dm(z)\partial_{\overline{z}}\tilde{\rho}(z)(\mathcal{H}_{\omega,L}-z)^{-2}.
\end{equation*}
Hence, it follows that 
\begin{equation} \label{eq_edge_index_HS_formula_proof_1}
\begin{aligned}
&\mathcal{E}_{edge,\omega,\mathcal{I},\rho,L} \\
&=\frac{i}{2|\Lambda_L|}\cdot\text{Tr}_{\Lambda_L}\big(([\mathcal{H}_{\omega,L},x_2]x_1-[\mathcal{H}_{\omega,L},x_1]x_2)\rho^{\prime}(\mathcal{H}_{\omega,L})\big) \\
\nm 
&=-\frac{1}{2\pi |\Lambda_L|}\int_{\mathbb{C}}dm(z)\partial_{\overline{z}}\tilde{\rho}(z)
\Big(\text{Tr}_{\Lambda_L}\big([\mathcal{H}_{\omega,L},x_2]x_1(\mathcal{H}_{\omega,L}-z)^{-2} \big)-\text{Tr}_{\Lambda_L}\big([\mathcal{H}_{\omega,L},x_1]x_2(\mathcal{H}_{\omega,L}-z)^{-2} \big) \Big).
\end{aligned}
\end{equation}
Both the separation of traces and the interchange of the trace and the complex integral are justified by noting that $\text{Tr}_{\Lambda_L}$ is a finite sum over the diagonals and, for each $\bm{n}\in\Lambda_L$, 
\begin{equation} \label{eq_edge_to_bulk_HS_proof_5}
\Big| \big( [\mathcal{H}_{\omega,L},x_i]x_j (\mathcal{H}_{\omega,L}-z)^{-2} \big)(\bm{n},\bm{n}) \Big| \leq \frac{CL}{|\Im z|^{2}},
\end{equation}
($C>0$ is independent of $\omega,L,z$) whose composition with $\partial_{\overline{z}}\tilde{\rho}(z)$ is absolutely integrable by the almost-analyticity \eqref{eq_almost_analyticity}. To further rewrite \eqref{eq_edge_index_HS_formula_proof_1}, we note that from the identity
\begin{equation*}
(\mathcal{H}_{\omega,L}-z)x_i (\mathcal{H}_{\omega,L}-z)^{-2}
=[\mathcal{H}_{\omega,L},x_i](\mathcal{H}_{\omega,L}-z)^{-2}+x_i(\mathcal{H}_{\omega,L}-z)^{-1},
\end{equation*}
it follows that 
\begin{equation*}
x_i (\mathcal{H}_{\omega,L}-z)^{-2}
=(\mathcal{H}_{\omega,L}-z)^{-1}[\mathcal{H}_{\omega,L},x_i](\mathcal{H}_{\omega,L}-z)^{-2}+(\mathcal{H}_{\omega,L}-z)^{-1}x_i(\mathcal{H}_{\omega,L}-z)^{-1}.
\end{equation*}
Thus, \eqref{eq_edge_index_HS_formula_proof_1} is equal to
\begin{equation} \label{eq_edge_index_HS_formula_proof_2}
\begin{aligned}
\mathcal{E}_{edge,\omega,\mathcal{I},\rho,L}
&=-\frac{1}{2\pi |\Lambda_L|}\int_{\mathbb{C}}dm(z)\partial_{\overline{z}}\tilde{\rho}(z)
\text{Tr}_{\Lambda_L}\Big([\mathcal{H}_{\omega,L},x_2](\mathcal{H}_{\omega,L}-z)^{-1}[\mathcal{H}_{\omega,L},x_1](\mathcal{H}_{\omega,L}-z)^{-2} \\
&\quad\quad\quad\quad\quad\quad\quad\quad\quad\quad\quad\quad\quad +[\mathcal{H}_{\omega,L},x_2](\mathcal{H}_{\omega,L}-z)^{-1}x_1(\mathcal{H}_{\omega,L}-z)^{-1} \Big) \\
&\quad +\frac{1}{2\pi |\Lambda_L|}\int_{\mathbb{C}}dm(z)\partial_{\overline{z}}\tilde{\rho}(z)
\text{Tr}_{\Lambda_L}\Big([\mathcal{H}_{\omega,L},x_1](\mathcal{H}_{\omega,L}-z)^{-1}[\mathcal{H}_{\omega,L},x_2](\mathcal{H}_{\omega,L}-z)^{-2} \\
&\quad\quad\quad\quad\quad\quad\quad\quad\quad\quad\quad\quad\quad +[\mathcal{H}_{\omega,L},x_1](\mathcal{H}_{\omega,L}-z)^{-1}x_2(\mathcal{H}_{\omega,L}-z)^{-1} \Big).
\end{aligned}
\end{equation}
By cyclicity of traces (of operators on the finite-dimensional space $\mathcal{X}_L$), we obtain that
\begin{equation*}
\begin{aligned}
&\text{Tr}_{\Lambda_L}\big( [\mathcal{H}_{\omega,L},x_i](\mathcal{H}_{\omega,L}-z)^{-1}x_j(\mathcal{H}_{\omega,L}-z)^{-1} \big) \\
&=\text{Tr}_{\Lambda_L}\big( (\mathcal{H}_{\omega,L}-z)x_i(\mathcal{H}_{\omega,L}-z)^{-1}x_j(\mathcal{H}_{\omega,L}-z)^{-1} \big)-\text{Tr}_{\Lambda_L}\big( x_ix_j(\mathcal{H}_{\omega,L}-z)^{-1} \big) \\
&=\text{Tr}_{\Lambda_L}\big( x_i(\mathcal{H}_{\omega,L}-z)^{-1}x_j \big)-\text{Tr}_{\Lambda_L}\big( x_ix_j(\mathcal{H}_{\omega,L}-z)^{-1} \big) \\
&=0,
\end{aligned}
\end{equation*}
for all $\Im z\neq 0$. Hence, only the terms involving $(\mathcal{H}_{\omega,L}-z)^{-2}$ in \eqref{eq_edge_index_HS_formula_proof_2} are not trivial:
\begin{equation*}
\begin{aligned}
&\mathcal{E}_{edge,\omega,\mathcal{I},\rho,L} \\
&=-\frac{1}{2\pi |\Lambda_L|}\int_{\mathbb{C}}dm(z)\partial_{\overline{z}}\tilde{\rho}(z)
\text{Tr}_{\Lambda_L}\Big([\mathcal{H}_{\omega,L},x_2](\mathcal{H}_{\omega,L}-z)^{-1}[\mathcal{H}_{\omega,L},x_1](\mathcal{H}_{\omega,L}-z)^{-2}  \Big) \\
&\quad +\frac{1}{2\pi |\Lambda_L|}\int_{\mathbb{C}}dm(z)\partial_{\overline{z}}\tilde{\rho}(z)
\text{Tr}_{\Lambda_L}\Big([\mathcal{H}_{\omega,L},x_1](\mathcal{H}_{\omega,L}-z)^{-1}[\mathcal{H}_{\omega,L},x_2](\mathcal{H}_{\omega,L}-z)^{-2} \Big) \\
&=-\frac{1}{2\pi |\Lambda_L|}\int_{\mathbb{C}}dm(z)\partial_{\overline{z}}\tilde{\rho}(z)
\text{Tr}_{\Lambda_L}\Big((\mathcal{H}_{\omega,L}-z)^{-1} [\mathcal{H}_{\omega,L},x_2](\mathcal{H}_{\omega,L}-z)^{-1}[\mathcal{H}_{\omega,L},x_1](\mathcal{H}_{\omega,L}-z)^{-1}  \Big) \\
&\quad +\frac{1}{2\pi |\Lambda_L|}\int_{\mathbb{C}}dm(z)\partial_{\overline{z}}\tilde{\rho}(z)
\text{Tr}_{\Lambda_L}\Big((\mathcal{H}_{\omega,L}-z)^{-1} [\mathcal{H}_{\omega,L},x_1](\mathcal{H}_{\omega,L}-z)^{-1}[\mathcal{H}_{\omega,L},x_2](\mathcal{H}_{\omega,L}-z)^{-1} \Big) \\
&=\frac{1}{2\pi |\Lambda_L|}\int_{\mathbb{C}}dm(z)\partial_{\overline{z}}\tilde{\rho}(z)
\text{Tr}_{\Lambda_L}\big(S_{\omega,L,12}(z)\big)
-\frac{1}{2\pi |\Lambda_L|}\int_{\mathbb{C}}dm(z)\partial_{\overline{z}}\tilde{\rho}(z)
\text{Tr}_{\Lambda_L}\big(S_{\omega,L,21}(z)\big)
\end{aligned}
\end{equation*}
where the cyclicity of trace is applied for the second equality. This completes the proof of \eqref{eq_edge_index_HS_formula}.

\subsection{Step 2: Changing Hamiltonians from Box to Bulk} \label{sec_HS_formula_edge_to_bulk}

{\color{blue}Step 2.1.} We first prove that
\begin{equation} \label{eq_edge_to_bulk_HS_proof_2}
\begin{aligned}
&\lim_{L\to\infty}\frac{1}{|\Lambda_L|}\int_{\mathbb{C}}dm(z)\partial_{\overline{z}}\tilde{\rho}(z)\text{Tr}_{\Lambda_L}\big( S_{\omega,ij}(z) \big) \\
&=\lim_{L\to\infty}\frac{1}{|\Lambda_L|}\int_{\mathbb{C}}dm(z)\partial_{\overline{z}}\tilde{\rho}(z)\text{Tr}_{\Lambda_L}\Big((\mathcal{H}_{\omega}-z)^{-1} [\mathcal{H}_{\omega},x_i] (\mathcal{H}_{\omega}-z)^{-1} [\mathcal{H}_{\omega},x_j] (\mathcal{H}_{\omega}-z)^{-1} \Big) \\
&=\lim_{L\to\infty}\frac{1}{|\Lambda_L|}\int_{\mathbb{C}}dm(z)\partial_{\overline{z}}\tilde{\rho}(z)\text{Tr}_{\Lambda_L}\Big((\mathcal{H}_{\omega}-z)^{-1} \mathbbm{1}_{\Lambda_L}[\mathcal{H}_{\omega},x_i] \mathbbm{1}_{\Lambda_L}(\mathcal{H}_{\omega}-z)^{-1} \mathbbm{1}_{\Lambda_L}[\mathcal{H}_{\omega},x_j] \mathbbm{1}_{\Lambda_L}(\mathcal{H}_{\omega}-z)^{-1} \Big),
\end{aligned}
\end{equation}
i.e., inserting $\mathbbm{1}_{\Lambda_L}$ into $S_{\omega,ij}(z)$ does not change the limit. The idea is to insert the indicators one by one and to show that the limit does not change at each step. Here, we only prove the first step
\begin{equation} \label{eq_edge_to_bulk_HS_proof_3}
\begin{aligned}
&\lim_{L\to\infty}\frac{1}{|\Lambda_L|}\int_{\mathbb{C}}dm(z)\partial_{\overline{z}}\tilde{\rho}(z)\text{Tr}_{\Lambda_L}\Big((\mathcal{H}_{\omega}-z)^{-1} [\mathcal{H}_{\omega},x_i] (\mathcal{H}_{\omega}-z)^{-1} [\mathcal{H}_{\omega},x_j] (\mathcal{H}_{\omega}-z)^{-1} \Big) \\
&=\lim_{L\to\infty}\frac{1}{|\Lambda_L|}\int_{\mathbb{C}}dm(z)\partial_{\overline{z}}\tilde{\rho}(z)\text{Tr}_{\Lambda_L}\Big((\mathcal{H}_{\omega}-z)^{-1} \mathbbm{1}_{\Lambda_L}[\mathcal{H}_{\omega},x_i] (\mathcal{H}_{\omega}-z)^{-1} [\mathcal{H}_{\omega},x_j] (\mathcal{H}_{\omega}-z)^{-1} \Big),
\end{aligned}
\end{equation}
while the other three are checked similarly. Using the standard estimates in Section \ref{sec_prelim}, we see that the operators in \eqref{eq_edge_to_bulk_HS_proof_3} all have exponentially localized kernels for $\Im z\neq 0$; that is,
\begin{equation*}
\big\|(\mathcal{H}_{\omega}-z)^{-1}(\bm{n},\bm{m})  \big\|\leq \frac{C}{|\Im z|}e^{-|\Im z| |\bm{n}-\bm{m}|} ,
\end{equation*}
and 
\begin{equation*}
\big\|\big( [\mathcal{H}_{\omega},x_i] (\mathcal{H}_{\omega}-z)^{-1} [\mathcal{H}_{\omega},x_j] (\mathcal{H}_{\omega}-z)^{-1} \big) (\bm{n},\bm{m})  \big\|\leq \frac{C}{|\Im z|^{8}}e^{-\frac{|\Im z|}{4} |\bm{n}-\bm{m}|} .
\end{equation*}
Hence, the differences between both sides of \eqref{eq_edge_to_bulk_HS_proof_3} can be estimated by calculating the traces in position basis
\begin{equation*}
\begin{aligned}
&\Big|\text{Tr}_{\Lambda_L}\Big(
(\mathcal{H}_{\omega}-z)^{-1} [\mathcal{H}_{\omega},x_i] (\mathcal{H}_{\omega}-z)^{-1} [\mathcal{H}_{\omega},x_j] (\mathcal{H}_{\omega}-z)^{-1} \\
&\quad\quad\quad\quad -(\mathcal{H}_{\omega}-z)^{-1} \mathbbm{1}_{\Lambda_L}[\mathcal{H}_{\omega},x_i] (\mathcal{H}_{\omega}-z)^{-1} [\mathcal{H}_{\omega},x_j] (\mathcal{H}_{\omega}-z)^{-1} \Big)\Big| \\
&=\Big|\text{tr}\sum_{\bm{n}\in\Lambda_L}\sum_{\bm{m}\in\Lambda_L^{c}} 
(\mathcal{H}_{\omega}-z)^{-1}(\bm{n},\bm{m}) \big( [\mathcal{H}_{\omega},x_i] (\mathcal{H}_{\omega}-z)^{-1} [\mathcal{H}_{\omega},x_j] (\mathcal{H}_{\omega}-z)^{-1} \big) (\bm{m},\bm{n}) \Big| \\
&\leq C\sum_{\bm{n}\in\Lambda_L,\bm{m}\in\Lambda_L^{c}} \big\|(\mathcal{H}_{\omega}-z)^{-1}(\bm{n},\bm{m}) \big\| \big\|\big( [\mathcal{H}_{\omega},x_i] (\mathcal{H}_{\omega}-z)^{-1} [\mathcal{H}_{\omega},x_j] (\mathcal{H}_{\omega}-z)^{-1} \big) (\bm{m},\bm{n}) \big\| \\
&\leq \frac{C}{|\Im z|^{9}}\sum_{\bm{n}\in\Lambda_L,\bm{m}\in\Lambda_L^{c}}e^{-\frac{|\Im z|}{2} |\bm{n}-\bm{m}|} .
\end{aligned}
\end{equation*}
The right side is estimated elementarily as in the proof of Lemma \ref{lem_finite_volume_cyclicity}, whose result is
\begin{equation*}
\begin{aligned}
&\Big|\text{Tr}_{\Lambda_L}\Big(
(\mathcal{H}_{\omega}-z)^{-1} [\mathcal{H}_{\omega},x_i] (\mathcal{H}_{\omega}-z)^{-1} [\mathcal{H}_{\omega},x_j] (\mathcal{H}_{\omega}-z)^{-1} \\
&\quad\quad\quad\quad -(\mathcal{H}_{\omega}-z)^{-1} \mathbbm{1}_{\Lambda_L}[\mathcal{H}_{\omega},x_i] (\mathcal{H}_{\omega}-z)^{-1} [\mathcal{H}_{\omega},x_j] (\mathcal{H}_{\omega}-z)^{-1} \Big)\Big| \leq \frac{C}{|\Im z|^{11}}\big(L^{\frac{3}{2}}+L^2e^{-\frac{|\Im z|}{4}L^{\frac{1}{2}}}\big).
\end{aligned}
\end{equation*}
This estimate, together with the almost analyticity \eqref{eq_almost_analyticity} and the dominated convergence theorem, concludes the proof of \eqref{eq_edge_to_bulk_HS_proof_3}.

{\color{blue}Step 2.2.} Next, we note that
\begin{equation*}
\begin{aligned}
\text{Tr}_{\Lambda_L}\big( S_{\omega,L,ij}(z) \big)
&=\text{Tr}_{\Lambda_L}\Big((\mathcal{H}_{\omega,L}-z)^{-1} [\mathcal{H}_{\omega,L},x_i](\mathcal{H}_{\omega,L}-z)^{-1}[\mathcal{H}_{\omega,L},x_j](\mathcal{H}_{\omega,L}-z)^{-1}  \Big) \\
&=\text{Tr}_{\Lambda_L}\Big((\mathcal{H}_{\omega,L}-z)^{-1} \mathbbm{1}_{\Lambda_L}[\mathcal{H}_{\omega,L},x_i] \mathbbm{1}_{\Lambda_L}(\mathcal{H}_{\omega,L}-z)^{-1} \mathbbm{1}_{\Lambda_L}[\mathcal{H}_{\omega,L},x_j] \mathbbm{1}_{\Lambda_L}(\mathcal{H}_{\omega,L}-z)^{-1} \Big) \\
&=\text{Tr}_{\Lambda_L}\Big((\mathcal{H}_{\omega,L}-z)^{-1} \mathbbm{1}_{\Lambda_L}[\mathcal{H}_{\omega},x_i] \mathbbm{1}_{\Lambda_L}(\mathcal{H}_{\omega,L}-z)^{-1} \mathbbm{1}_{\Lambda_L}[\mathcal{H}_{\omega},x_j] \mathbbm{1}_{\Lambda_L}(\mathcal{H}_{\omega,L}-z)^{-1} \Big),
\end{aligned}
\end{equation*}
where the second identity follows from the fact that all the operators involved are defined on $\mathcal{X}_L$, and therefore, inserting $\mathbbm{1}_{\Lambda_L}$ does not affect the composition. Hence, with \eqref{eq_edge_to_bulk_HS_proof_2}, it suffices to prove the following equality to conclude the proof of \eqref{eq_edge_to_bulk_HS}:
\begin{equation*}
\begin{aligned}
&\lim_{L\to\infty}\frac{1}{|\Lambda_L|}\int_{\mathbb{C}}dm(z)\partial_{\overline{z}}\tilde{\rho}(z)\text{Tr}_{\Lambda_L}\Big((\mathcal{H}_{\omega,L}-z)^{-1} \mathbbm{1}_{\Lambda_L}[\mathcal{H}_{\omega},x_i] \mathbbm{1}_{\Lambda_L}(\mathcal{H}_{\omega,L}-z)^{-1} \mathbbm{1}_{\Lambda_L}[\mathcal{H}_{\omega,L},x_j] \mathbbm{1}_{\Lambda_L}(\mathcal{H}_{\omega,L}-z)^{-1} \Big) \\
&=\lim_{L\to\infty}\frac{1}{|\Lambda_L|}\int_{\mathbb{C}}dm(z)\partial_{\overline{z}}\tilde{\rho}(z)\text{Tr}_{\Lambda_L}\Big((\mathcal{H}_{\omega}-z)^{-1} \mathbbm{1}_{\Lambda_L}[\mathcal{H}_{\omega},x_i] \mathbbm{1}_{\Lambda_L}(\mathcal{H}_{\omega}-z)^{-1} \mathbbm{1}_{\Lambda_L}[\mathcal{H}_{\omega},x_j] \mathbbm{1}_{\Lambda_L}(\mathcal{H}_{\omega}-z)^{-1} \Big).
\end{aligned}
\end{equation*}
Again, the proof is achieved by  gradually changing the resolvent $(\mathcal{H}_{\omega,L}-z)^{-1}$ to $(\mathcal{H}_{\omega}-z)^{-1}$ and showing that the limit does not change at each step. We only show the proof of the first step, while the others are treated similarly:
\begin{equation} \label{eq_edge_to_bulk_HS_proof_4}
\begin{aligned}
&\lim_{L\to\infty}\frac{1}{|\Lambda_L|}\int_{\mathbb{C}}dm(z)\partial_{\overline{z}}\tilde{\rho}(z) \\
&\quad \text{Tr}_{\Lambda_L}\Big(
\big((\mathcal{H}_{\omega,L}-z)^{-1}-(\mathcal{H}_{\omega}-z)^{-1} \big)\mathbbm{1}_{\Lambda_L}[\mathcal{H}_{\omega},x_i] \mathbbm{1}_{\Lambda_L}(\mathcal{H}_{\omega}-z)^{-1} \mathbbm{1}_{\Lambda_L}[\mathcal{H}_{\omega},x_j] \mathbbm{1}_{\Lambda_L}(\mathcal{H}_{\omega}-z)^{-1} \Big)=0.
\end{aligned}
\end{equation}
With the resolvent identity \eqref{eq_edge_to_bulk_HS_proof_1} and the estimates presented in Section \ref{sec_prelim}, the operator kernels of the right side are bounded as follows. First,
\begin{equation*}
\begin{aligned}
\big\|\big((\mathcal{H}_{\omega,L}-z)^{-1}-(\mathcal{H}_{\omega}-z)^{-1} \big) (\bm{n},\bm{m})\big\|
&\leq \frac{C}{|\Im z|^{2}} \sum_{(\bm{u},\bm{u}^{\prime})\in \text{supp}(\mathcal{F}^{(L)})} e^{-|\Im z||\bm{n}-\bm{u}|}e^{-|\Im z||\bm{u}^{\prime}-\bm{m}|} \\
&\leq \frac{C}{|\Im z|^{2}} \sum_{\bm{u}\in\partial^{int}\Lambda_L} e^{-|\Im z||\bm{n}-\bm{u}|}e^{-|\Im z||\bm{u}-\bm{m}|} ,
\end{aligned}
\end{equation*}
where in the last step, we recall that $|\bm{u}-\bm{u}^{\prime}|=1$ for $(\bm{u},\bm{u}^{\prime})\in \text{supp}(\mathcal{F}^{(L)})$ and the fact that $\mathcal{F}^{(L)}$ is supported on the boundary $\partial \Lambda_L$. On the other hand, we have
\begin{equation*}
\big\|\big( [\mathcal{H}_{\omega},x_i] \mathbbm{1}_{\Lambda_L}(\mathcal{H}_{\omega}-z)^{-1} \mathbbm{1}_{\Lambda_L}[\mathcal{H}_{\omega},x_j] \mathbbm{1}_{\Lambda_L}(\mathcal{H}_{\omega}-z)^{-1} \big) (\bm{n},\bm{m})\big\| \leq  \frac{C}{|\Im z|^{8}}e^{-\frac{|\Im z|}{4}|\bm{n}-\bm{m}|}.
\end{equation*}
Hence, the trace in \eqref{eq_edge_to_bulk_HS_proof_4} is estimated as
\begin{equation*}
\begin{aligned}
& \Big|\text{Tr}_{\Lambda_L}\Big(
\big((\mathcal{H}_{\omega,L}-z)^{-1}-(\mathcal{H}_{\omega}-z)^{-1} \big)\mathbbm{1}_{\Lambda_L}[\mathcal{H}_{\omega},x_i] \mathbbm{1}_{\Lambda_L}(\mathcal{H}_{\omega}-z)^{-1} \mathbbm{1}_{\Lambda_L}[\mathcal{H}_{\omega},x_j] \mathbbm{1}_{\Lambda_L}(\mathcal{H}_{\omega}-z)^{-1} \Big) \Big| \\
& \leq \frac{C}{|\Im z|^{10}}\sum_{\substack{\bm{n},\bm{m}\in \Lambda_L \\ \bm{u}\in\partial^{int}\Lambda_L}}e^{-|\Im z||\bm{n}-\bm{u}|}e^{-|\Im z||\bm{u}-\bm{m}|}e^{-\frac{|\Im z|}{4}|\bm{m}-\bm{n}|}.
\end{aligned}
\end{equation*}
The right side is calculated as follows: we first sum over $\bm{n}$ then $\bm{m}$, where each step produces a factor of $1/|\Im z|^2$ and all exponential terms vanish at last\footnote{This is analogous to the proof of Lemma \ref{lem_finite_volume_cyclicity}.}; then we finally sum the unity over $\bm{u}\in\partial^{int}\Lambda_L$, whose result is simply $|\partial^{int}\Lambda_L|=\mathcal{O}(L)$. Therefore,
\begin{equation*}
\begin{aligned}
 \Big|\text{Tr}_{\Lambda_L}\Big(
\big((\mathcal{H}_{\omega,L}-z)^{-1}-(\mathcal{H}_{\omega}-z)^{-1} \big)\mathbbm{1}_{\Lambda_L}[\mathcal{H}_{\omega},x_i] \mathbbm{1}_{\Lambda_L}(\mathcal{H}_{\omega}-z)^{-1} \mathbbm{1}_{\Lambda_L}[\mathcal{H}_{\omega},x_j] \mathbbm{1}_{\Lambda_L}(\mathcal{H}_{\omega}-z)^{-1} \Big) \Big|\\
\nm
\leq \frac{CL}{|\Im z|^{14}}
\end{aligned}
\end{equation*}
with $C>0$ being independent of $L$ and $z$. Thus, by the almost analyticity \eqref{eq_almost_analyticity} and the dominated convergence theorem, the limit and the complex integral in \eqref{eq_edge_to_bulk_HS_proof_4} can be interchanged to obtain a vanishing integrand. This concludes the proof.

\subsection{Step 3: Link to the Bulk Index} \label{sec_link_refined_bulk_index}

We first fix some notation. Recall that in Definition \ref{prop_refined_bulk_index}, we have taken an open sub-interval $\mathcal{I}^{\prime}$ of the mobility gap $\mathcal{I}$ such that it includes the transition region of density function $\rho$, i.e.,
\begin{equation*}
\text{supp}(\rho^{\prime}) \Subset \mathcal{I}^{\prime} \Subset \mathcal{I} ,
\end{equation*}
and have defined $P_{\omega,loc}:=\mathbbm{1}_{\mathcal{I}^{\prime} }(\mathcal{H}_{\omega})$. Now, we introduce two complement projections
\begin{equation*}
\begin{aligned}
P_{\omega,-}:=\mathbbm{1}_{(-\infty,\inf \mathcal{I}^{\prime}]}(\mathcal{H}_{\omega}),\quad 
P_{\omega,+}:=\mathbbm{1}_{[\sup \mathcal{I}^{\prime},\infty)}(\mathcal{H}_{\omega}) .
\end{aligned}
\end{equation*}
The supports of these projections are illustrated in Figure \ref{fig_support_projection}. By Proposition \ref{prop_localization_ground_state}, these projections are exponentially localized $\mathbb{P}-$almost surely
\begin{equation} \label{eq_aux_projection_ae_localized}
\mathbb{E}\big(\|P(\bm{n},\bm{m})\|  \big)\leq Ce^{-\alpha|\bm{n}-\bm{m}|} \quad \text{ for all }\bm{n},\bm{m}\in\mathbb{Z}^2,\, P\in \{P_{\omega,+},P_{\omega,-},P_{\omega,loc}\}.
\end{equation}

\begin{figure}
    \centering
    \begin{tikzpicture}[scale=0.8]
\draw[->] (-8,0)--(8,0);
\draw[very thick, red] (-8,4)--(-2,4);
\draw[very thick, red] (2,2)--(8,2);
\draw [very thick, red] plot [smooth,samples=400, tension=1] coordinates {(-2,4) (-1,3.8) (0,3) (1,2.2) (2,2)};
\draw[dashed] (-8,2)--(2,2);
\node[left,scale=0.7] at (-8,4) {$\rho=1$};
\node[left,scale=0.7] at (-8,2) {$\rho=0$};
\node[scale=1.2] at (-3,0) {$|$};
\node[scale=1.2] at (3,0) {$|$};
\draw[decorate,decoration={brace,mirror}] (-8,-0.5) -- (-3,-0.5);
\node[below,scale=1] at (-5.5,-0.5) {$P_{\omega,-}$};
\draw[decorate,decoration={brace,mirror}] (3,-0.5) -- (8,-0.5);
\node[below,scale=1] at (5.5,-0.5) {$P_{\omega,+}$};
\draw[decorate,decoration={brace,mirror}] (-3,-0.5) -- (3,-0.5);
\node[below,scale=1] at (0,-0.5) {$P_{\omega,loc}$};

\draw[line width=0.7mm, blue] (-8,0)--(-3,0);
\draw[line width=0.7mm, black] (8,0)--(3,0);
\draw[line width=0.7mm, green] (-3,0)--(3,0);
\end{tikzpicture}
    \caption{Support of projections}
    \label{fig_support_projection}
\end{figure}

{\color{blue}Step 3.1.} We first note that, by arguing in a similar way to  \eqref{eq_edge_to_bulk_HS_proof_5}, the averaged trace $\frac{1}{|\Lambda_L|}\text{Tr}_{\Lambda_L}\big( S_{\omega,ij}(z) \big)$ is uniformly bounded in $L$ up to a polynomial order of $|\Im z|^{-1}$. Hence, we can interchange the limit and the complex integral as follows:
\begin{equation} \label{eq_link_HS_bulk_index_proof_1}
\begin{aligned}
\lim_{L\to\infty}\frac{1}{|\Lambda_L|}\int_{\mathbb{C}}dm(z)\partial_{\overline{z}}\tilde{\rho}(z)\text{Tr}_{\Lambda_L}\big( S_{\omega,ij}(z) \big)
&=\int_{\mathbb{C}}dm(z)\partial_{\overline{z}}\tilde{\rho}(z)\lim_{L\to\infty}\frac{1}{|\Lambda_L|}\text{Tr}_{\Lambda_L}\big( S_{\omega,ij}(z) \big) \\
&=\int_{\mathbb{C}}dm(z)\partial_{\overline{z}}\tilde{\rho}(z)\mathcal{T}\big( S_{\omega,ij}(z) \big).
\end{aligned}
\end{equation}
Inserting $1=P_{\omega,+}+P_{\omega,-}+P_{\omega,loc}$ to the TPUV, recalling \eqref{eq_aux_projection_ae_localized}, Lemma \ref{lem_covariant_TPUV_cyclicity} and the fact that $P_{\omega,\pm},P_{\omega,loc},S_{\omega,ij}(z)$ are all covariant (for $\Im z\neq 0$), we see that the following holds $\mathbb{P}-$almost surely:
\begin{equation*}
\begin{aligned}
&\lim_{L\to\infty}\frac{1}{|\Lambda_L|}\int_{\mathbb{C}}dm(z)\partial_{\overline{z}}\tilde{\rho}(z)\text{Tr}_{\Lambda_L}\big( S_{\omega,ij}(z) \big) =\int_{\mathbb{C}}dm(z)\partial_{\overline{z}}\tilde{\rho}(z)\sum_{P\in \{P_{\omega,\pm},P_{\omega,loc}\}}\mathcal{T}\big(P S_{\omega,ij}(z) P\big).
\end{aligned}
\end{equation*}
For the same reason as \eqref{eq_link_HS_bulk_index_proof_1}, we can interchange TPUV and complex integral and obtain
\begin{equation} \label{eq_link_HS_bulk_index_proof_2}
\begin{aligned}
\lim_{L\to\infty}\frac{1}{|\Lambda_L|}\int_{\mathbb{C}}dm(z)\partial_{\overline{z}}\tilde{\rho}(z)\text{Tr}_{\Lambda_L}\big( S_{\omega,ij}(z) \big) =\sum_{P\in \{P_{\omega,\pm},P_{\omega,loc}\}}\mathcal{T}\big(P \int_{\mathbb{C}}dm(z)\partial_{\overline{z}}\tilde{\rho}(z)S_{\omega,ij}(z) P\big).
\end{aligned}
\end{equation}

{\color{blue}Step 3.2.} To further manipulate the right side of \eqref{eq_link_HS_bulk_index_proof_2}, the following decomposition is important:
\begin{equation} \label{eq_decompose_Sij}
\begin{aligned}
S_{\omega,ij}(z)
&=(\mathcal{H}_{\omega}-z)^{-1} [\mathcal{H}_{\omega},x_i] (\mathcal{H}_{\omega}-z)^{-1} [\mathcal{H}_{\omega},x_j] (\mathcal{H}_{\omega}-z)^{-1} \\
&=-(\mathcal{H}_{\omega}-z)^{-1} [\mathcal{H}_{\omega},x_i](\mathcal{H}_{\omega}-z)^{-1}x_j
+(\mathcal{H}_{\omega}-z)^{-1} [\mathcal{H}_{\omega},x_i]x_j(\mathcal{H}_{\omega}-z)^{-1} \\
&= [(\mathcal{H}_{\omega}-z)^{-1},x_i]x_j
+(\mathcal{H}_{\omega}-z)^{-1} [\mathcal{H}_{\omega},x_i]x_j(\mathcal{H}_{\omega}-z)^{-1} \\
&=:S_{\omega,ij}^{(1)}(z)+S_{\omega,ij}^{(2)}(z) ,
\end{aligned}
\end{equation}
which follows from the following identity: 
$$
(\mathcal{H}_{\omega}-z)^{-1} [\mathcal{H}_{\omega},x_j] (\mathcal{H}_{\omega}-z)^{-1}=-[(\mathcal{H}_{\omega}-z)^{-1},x_j].
$$
Moreover, we note that, for $P=P_{\omega,\pm}$ and any $\bm{n}\in \Lambda_L$,
\begin{equation} \label{eq_link_HS_bulk_index_proof_3}
\begin{aligned}
&\mathbbm{1}_{\{\bm{n}\}}P\Big( \int_{\mathbb{C}}dm(z)\partial_{\overline{z}}\tilde{\rho}(z)S_{\omega,ij}^{(2)}(z)\Big) P\mathbbm{1}_{\{\bm{n}\}} \\
& \ds = \int_{\mathbb{C}}dm(z)\partial_{\overline{z}}\tilde{\rho}(z) \Big(\mathbbm{1}_{\{\bm{n}\}}P (\mathcal{H}_{\omega}-z)^{-1} [\mathcal{H}_{\omega},x_i]x_j(\mathcal{H}_{\omega}-z)^{-1} P \mathbbm{1}_{\{\bm{n}\}}\Big) \\
&=0,
\end{aligned}
\end{equation}
because $(\mathcal{H}_{\omega}-z)^{-1} P=P(\mathcal{H}_{\omega}-z)^{-1}$ is analytic either in the support of $\tilde{\rho}$ or in its complement and hence, using the integration by parts, the complex integral \eqref{eq_link_HS_bulk_index_proof_3} can be calculated on the boundary $\partial (\text{supp}(\tilde{\rho}))$ on which the integrand vanishes.\footnote{A detailed proof follows the lines of \cite[Proposition 2.9]{qiu2025generalized}. Note that in \cite{qiu2025generalized} we deal with bounded operators but the position operator in \eqref{eq_link_HS_bulk_index_proof_3} is unbounded. Nonetheless, this is not essential for the proof as it does not hurt the analyticity. In fact, the application of the indicator function $\mathbbm{1}_{\{\bm{n}\}}$ and the localization property of $(\mathcal{H}_{\omega}-z)^{-1} P$ (in the sense of Lemma \ref{lem_composition_weak_localized})
implies $x_j(\mathcal{H}_{\omega}-z)^{-1} P \mathbbm{1}_{\{\bm{n}\}}$ is both analytic and $\ell^2$-bounded (up to a polynomial order of $|\Im z|^{-1}$). Then the proof of \cite[Proposition 2.9]{qiu2025generalized} applies readily to \eqref{eq_link_HS_bulk_index_proof_3}.} Thus, it holds $\mathbb{P}-$almost surely that
\begin{equation*}
\begin{aligned}
\lim_{L\to\infty}\frac{1}{|\Lambda_L|}\int_{\mathbb{C}}dm(z)\partial_{\overline{z}}\tilde{\rho}(z)\text{Tr}_{\Lambda_L}\big( S_{\omega,L,ij}(z) \big) 
&=\sum_{P\in \{P_{\omega,-},P_{\omega,+}\}}\mathcal{T}\big(P \int_{\mathbb{C}}dm(z)\partial_{\overline{z}}\tilde{\rho}(z)S_{\omega,ij}^{(1)}(z) P\big) \\
&\quad + \mathcal{T}\big(P_{\omega,loc} \int_{\mathbb{C}}dm(z)\partial_{\overline{z}}\tilde{\rho}(z)S_{\omega,ij}(z) P_{\omega,loc}\big) .
\end{aligned}
\end{equation*}
On the other hand, by the Helffer-Sjöstrand formula, we have
\begin{equation*}
\int_{\mathbb{C}}dm(z)\partial_{\overline{z}}\tilde{\rho}(z)S_{\omega,ij}^{(1)}(z)
=\int_{\mathbb{C}}dm(z)\partial_{\overline{z}}\tilde{\rho}(z)[(\mathcal{H}_{\omega}-z)^{-1},x_i]x_j =\pi i [\rho(\mathcal{H}_{\omega}),x_i]x_j .
\end{equation*}
Hence, we conclude that
\begin{equation} \label{eq_link_HS_bulk_index_proof_4}
\begin{aligned}
\lim_{L\to\infty}\frac{1}{|\Lambda_L|}\int_{\mathbb{C}}dm(z)\partial_{\overline{z}}\tilde{\rho}(z)\text{Tr}_{\Lambda_L}\big( S_{\omega,L,ij}(z) \big) 
&=\pi i\sum_{P\in \{P_{\omega,-},P_{\omega,+}\}}\mathcal{T}\big(P [\rho(\mathcal{H}_{\omega}),x_i]x_j P\big) \\
\nm
&\quad + \mathcal{T}\big(P_{\omega,loc} \int_{\mathbb{C}}dm(z)\partial_{\overline{z}}\tilde{\rho}(z)S_{\omega,ij}(z) P_{\omega,loc}\big) .
\end{aligned}
\end{equation}
However, we cannot obtain a precise expression for the term involving $P_{\omega,loc}$. In fact,
\begin{equation} \label{eq_link_HS_bulk_index_proof_7}
\begin{aligned}
&\mathcal{T}\big(P_{\omega,loc} \int_{\mathbb{C}}dm(z)\partial_{\overline{z}}\tilde{\rho}(z)S_{\omega,ij}(z) P_{\omega,loc}\big) \\
&=\mathcal{T}\Big\{P_{\omega,loc} \int_{\mathbb{C}}dm(z)\partial_{\overline{z}}\tilde{\rho}(z) [(\mathcal{H}_{\omega}-z)^{-1},x_i]x_j P_{\omega,loc} \\
&\quad\quad\quad\quad\quad + P_{\omega,loc} \int_{\mathbb{C}}dm(z)\partial_{\overline{z}}\tilde{\rho}(z) (\mathcal{H}_{\omega}-z)^{-1} [\mathcal{H}_{\omega},x_i]x_j(\mathcal{H}_{\omega}-z)^{-1} P_{\omega,loc} \Big\} \\
&=\mathcal{T}\Big\{\pi i\cdot P_{\omega,loc} [\rho(\mathcal{H}_{\omega}),x_i]x_jP_{\omega,loc} \\
&\quad\quad\quad\quad\quad + P_{\omega,loc} \int_{\mathbb{C}}dm(z)\partial_{\overline{z}}\tilde{\rho}(z) (\mathcal{H}_{\omega}-z)^{-1} [\mathcal{H}_{\omega},x_i]x_j(\mathcal{H}_{\omega}-z)^{-1}P_{\omega,loc} \Big\},
\end{aligned}
\end{equation}
where the Helffer-Sjöstrand formula is applied in the last equality. Note that the equalities in \eqref{eq_link_HS_bulk_index_proof_7} are exactly contained in the correction term of the edge index \eqref{eq_refined_bulk_index}. Compared to \eqref{eq_link_HS_bulk_index_proof_3}, one sees that the second term in \eqref{eq_link_HS_bulk_index_proof_7} does not vanish because $P_{\omega,loc}(\mathcal{H}_{\omega}-z)^{-1}$ is not analytic in the support of $\tilde{\rho}$ or of $1-\tilde{\rho}$. More importantly, the operator $[\rho(\mathcal{H}_{\omega}),x_i]x_j$ is not covariant, which prevents \eqref{eq_link_HS_bulk_index_proof_7} from being decomposed into two separate TPUV. This explains the complicated expression of the correction terms in \eqref{eq_refined_bulk_index}.

{\color{blue}Step 3.3.} By \eqref{eq_link_HS_bulk_index_proof_4}, \eqref{eq_link_HS_bulk_index_proof_7}, and the results of Sections \ref{sec_edge_index_HS_formula}-\ref{sec_HS_formula_edge_to_bulk}, we conclude that
\begin{equation*}
\begin{aligned}
\lim_{L\to\infty}\mathcal{E}_{edge,\omega,L} 
& - \frac{i}{2}\cdot \mathcal{T}\Big\{P_{\omega,loc} [\rho(\mathcal{H}_{\omega}),x_1]x_2P_{\omega,loc} \\
&\quad\quad\quad\quad\quad + \frac{1}{\pi}P_{\omega,loc} \int_{\mathbb{C}}dm(z)\partial_{\overline{z}}\tilde{\rho}(z) (\mathcal{H}_{\omega}-z)^{-1} [\mathcal{H}_{\omega},x_1]x_2(\mathcal{H}_{\omega}-z)^{-1}P_{\omega,loc} \Big\} \\
&  + \frac{i}{2}\cdot \mathcal{T}\Big\{P_{\omega,loc} [\rho(\mathcal{H}_{\omega}),x_2]x_1P_{\omega,loc} \\
&\quad\quad\quad\quad\quad + \frac{1}{\pi}P_{\omega,loc} \int_{\mathbb{C}}dm(z)\partial_{\overline{z}}\tilde{\rho}(z) (\mathcal{H}_{\omega}-z)^{-1} [\mathcal{H}_{\omega},x_2]x_1(\mathcal{H}_{\omega}-z)^{-1}P_{\omega,loc} \Big\} \\
&=\frac{i}{2}\sum_{P\in \{P_{\omega,\pm}\}}\mathcal{T}\big(P [\rho(\mathcal{H}_{\omega}),x_1]x_2 P\big) - \frac{i}{2}\sum_{P\in \{P_{\omega,\pm}\}}\mathcal{T}\big(P [\rho(\mathcal{H}_{\omega}),x_2]x_1 P\big).
\end{aligned}
\end{equation*}
Thus, it remains to prove the following identities to conclude the proof of Theorem \ref{thm_bec_disorder}:
\begin{equation}  \label{eq_link_HS_bulk_index_proof_12}
\begin{aligned}
& \mathcal{E}_{bulk,\omega,\lambda} =i \sum_{P\in \{P_{\omega,\pm}\}}\mathcal{T}\big(P [\rho(\mathcal{H}_{\omega}),x_1]x_2 P\big)  -i\mathcal{T}\big( P_{\omega,loc}\int d\lambda 
\rho^{\prime}(\lambda) T_{\lambda,12} P_{\omega,loc}\big) .
\end{aligned}
\end{equation}
and
\begin{equation}  \label{eq_link_HS_bulk_index_proof_12_alter}
\begin{aligned}
& \mathcal{E}_{bulk,\omega,\lambda} =-i \sum_{P\in \{P_{\omega,\pm}\}}\mathcal{T}\big(P [\rho(\mathcal{H}_{\omega}),x_2]x_1 P\big)  +i\mathcal{T}\big( P_{\omega,loc}\int d\lambda 
\rho^{\prime}(\lambda) T_{\lambda,21} P_{\omega,loc}\big) .
\end{aligned}
\end{equation}
We only prove \eqref{eq_link_HS_bulk_index_proof_12}, while \eqref{eq_link_HS_bulk_index_proof_12_alter} is derived following the same lines. To do this, we note that
\begin{equation*}
\begin{aligned}
\mathcal{T}\big(P_{\omega,\lambda}[P_{\omega,\lambda},x_1][P_{\omega,\lambda},x_2]\big)
&\overset{(i)}{=}\mathcal{T}\big(P_{\omega,\lambda}[P_{\omega,\lambda},x_1][P_{\omega,\lambda},x_2]P_{\omega,\lambda}\big)
=\mathcal{T}\big(P_{\omega,\lambda}[P_{\omega,\lambda}^{\perp},x_1][P_{\omega,\lambda}^{\perp},x_2]P_{\omega,\lambda}\big) \\
&\overset{(ii)}{=}\mathcal{T}\big(P_{\omega,\lambda}x_1P_{\omega,\lambda}^{\perp} x_2P_{\omega,\lambda}\big),
\end{aligned}
\end{equation*}
where the application of cyclicity in (i) is justified by the covariance and localization of $P_{\omega,\lambda}$ and $[P_{\omega,\lambda},x_i]$, and the identity $P_{\omega,\lambda}P_{\omega,\lambda}^{\perp}=0$ is applied in (ii). Similarly, we have
\begin{equation*}
\begin{aligned}
\mathcal{T}\big(P_{\omega,\lambda}[P_{\omega,\lambda},x_2][P_{\omega,\lambda},x_1]\big)
&=\mathcal{T}\big(P_{\omega,\lambda}[P_{\omega,\lambda}^{\perp},x_2][P_{\omega,\lambda}^{\perp},x_1]P_{\omega,\lambda}\big)
=-\mathcal{T}\big(P_{\omega,\lambda}x_2P_{\omega,\lambda}^{\perp} x_1P_{\omega,\lambda}\big) \\
&=\mathcal{T}\big([P_{\omega,\lambda},x_2]P_{\omega,\lambda}^{\perp}P_{\omega,\lambda}^{\perp} [P_{\omega,\lambda},x_1]\big)=\mathcal{T}\big(P_{\omega,\lambda}^{\perp} [P_{\omega,\lambda},x_1][P_{\omega,\lambda},x_2]P_{\omega,\lambda}^{\perp}\big) \\
&=-\mathcal{T}\big(P_{\omega,\lambda}^{\perp}x_1P_{\omega,\lambda} x_2P_{\omega,\lambda}^{\perp}\big).
\end{aligned}
\end{equation*}
Hence, it follows that
\begin{equation} \label{eq_link_HS_bulk_index_proof_18}
\begin{aligned}
\mathcal{E}_{bulk,\omega,\lambda}&=-i\cdot\mathcal{T}\big(P_{\omega,\lambda}\big[[P_{\omega,\lambda},x_1],[P_{\omega,\lambda},x_2]\big]\big) \\
&=i\cdot \mathcal{T}\big(P_{\omega,\lambda}x_1P_{\omega,\lambda}^{\perp} x_2P_{\omega,\lambda}-P_{\omega,\lambda}^{\perp}x_1P_{\omega,\lambda} x_2P_{\omega,\lambda}^{\perp}\big)=i\cdot \mathcal{T}\big(T_{\lambda,12}\big).
\end{aligned}
\end{equation}
On the other hand, since $\mathcal{E}_{bulk,\omega,\lambda}$ is independent of $\lambda\in \mathcal{I}$, we have the following equality by $1=-\int d\lambda\rho^{\prime}(\lambda)$:
\begin{equation*}
\begin{array}{lll}
\mathcal{E}_{bulk,\omega,\lambda} &=& -i\int d\lambda\rho^{\prime}(\lambda)\mathcal{T}\big(T_{\lambda,12}\big)\\
\nm 
&=& -i\cdot \mathcal{T}\big(\int d\lambda\rho^{\prime}(\lambda)T_{\lambda,12}\big).
\end{array}
\end{equation*}
Here, the interchange between the TPUV and the integral is justified by the Fubini theorem, the fact that $T_{\lambda,12}$ is covariant, and the estimate $\int_{\text{supp}(\rho^{\prime})} d\lambda \cdot\mathbb{E}\Big(\big|T_{\lambda,12}(0,0)  \big| \Big)<\infty$.\footnote{To prove this, we write the TPUV as expectation, and see that $$\mathbb{E}\Big( \big|\big(P_{\omega,\lambda}x_1P_{\omega,\lambda}^{\perp} x_2P_{\omega,\lambda}\big)(0,0) \big| \Big)\leq \mathbb{E}\big(\|xP_{\omega,\lambda}\mathbbm{1}_{\{0\}}\|_{\mathcal{X}}^2 \big)\leq \mathbb{E}\big(\sum_{\bm{n}\in \mathbb{Z}^2}|\bm{n}|^2\|P_{\omega,\lambda}(0,\bm{n})\| \big)\leq C<\infty,$$
with the constant $C$ being independent of $\lambda$ thanks to  \eqref{eq_expectation_all_B1_function}. Hence, the TPUV is interchangable with the $\lambda$-integral.                 } Next, we insert $1=P_{\omega,+}+P_{\omega,-}+P_{\omega,loc}$ into the TPUV and apply the cyclicity
\begin{equation*}
\begin{aligned}
\mathcal{E}_{bulk,\omega,\lambda}
&= -i\cdot \mathcal{T}\big(P_{\omega,-}\int d\lambda\rho^{\prime}(\lambda)T_{\lambda,12}P_{\omega,-}\big)-i\cdot \mathcal{T}\big(P_{\omega,+}\int d\lambda\rho^{\prime}(\lambda)T_{\lambda,12}P_{\omega,+}\big)\\
&\quad -i\cdot \mathcal{T}\big(P_{\omega,loc}\int d\lambda\rho^{\prime}(\lambda)T_{\lambda,12}P_{\omega,loc}\big) .
\end{aligned}
\end{equation*}
For the first two terms, we note that from$$P_{\omega,-}P_{\omega,\lambda}=P_{\omega,\lambda}P_{\omega,-}=P_{\omega,-}$$ and $$P_{\omega,-}P_{\omega,\lambda}^{\perp}=P_{\omega,\lambda}^{\perp}P_{\omega,-}=0,$$ 
it follows that
\begin{equation} \label{eq_T_lambda_decompose_1}
\begin{array}{lll}
P_{\omega,-}T_{\lambda,12}P_{\omega,-} 
=P_{\omega,-}x_1P_{\omega,\lambda}^{\perp}x_2P_{\omega,-} 
= P_{\omega,-}[P_{\omega,\lambda},x_1]x_2P_{\omega,-},
\end{array}
\end{equation}
and similarly,
\begin{equation} \label{eq_T_lambda_decompose_2}
P_{\omega,+}T_{\lambda,12}P_{\omega,+}=-P_{\omega,+}x_1P_{\omega,\lambda}x_2P_{\omega,+}=P_{\omega,+}[P_{\omega,\lambda},x_1]x_2P_{\omega,+}.
\end{equation}
Thus, we have
\begin{equation} \label{eq_link_HS_bulk_index_proof_13}
\begin{aligned}
\mathcal{E}_{bulk,\omega,\lambda}
&=-i\cdot \mathcal{T}\big(P_{\omega,-}\int d\lambda\rho^{\prime}(\lambda)[P_{\omega,\lambda},x_1]x_2P_{\omega,-}\big)-i\cdot \mathcal{T}\big(P_{\omega,+}\int d\lambda\rho^{\prime}(\lambda)[P_{\omega,\lambda},x_1]x_2P_{\omega,+}\big) \\
&\quad -i\cdot \mathcal{T}\big(P_{\omega,loc}\int d\lambda\rho^{\prime}(\lambda)T_{\lambda,12}P_{\omega,loc}\big) \\
&=i\cdot \mathcal{T}\big(P_{\omega,-}[\rho(\mathcal{H_{\omega}}),x_1]x_2P_{\omega,-}\big)+i\cdot \mathcal{T}\big(P_{\omega,+}[\rho(\mathcal{H_{\omega}}),x_1]x_2P_{\omega,+}\big) \\
&\quad -i\cdot \mathcal{T}\big(P_{\omega,loc}\int d\lambda\rho^{\prime}(\lambda)T_{\lambda,12}P_{\omega,loc}\big) ,
\end{aligned}
\end{equation}
where the identity $\rho(\mathcal{H})=-\int d\lambda\rho^{\prime}(\lambda)P_{\omega,\lambda}$ is applied in the last equality. This concludes the proof of \eqref{eq_link_HS_bulk_index_proof_12}.

\section{Aizenman-Molchanov Localization Under Random Gap Opening Potential}
\label{sec_anderson_loc}

We prove Theorem \ref{thm_anderson_localization} in this section. We first establish the Lifshitz-type estimate and the initial length estimate in Section \ref{sec_lif_initial}, which are the two key ingredients of the proof. Then, following the standard geometric decoupling argument, we complete the proof of Theorem \ref{thm_anderson_localization} in Section \ref{sec_geo_decouple}.

\subsection{Lifshitz Tail and Initial Length Estimate Near Gap Opening Point}
\label{sec_lif_initial}

We first introduce the box Hamiltonian with periodic boundary conditions, which is a counterpart of the one with simple boundary conditions defined in \eqref{eq_box_Hamiltonian_simple}, but is more suitable for studying the Lifshitz tail near the gap opening point. Define the gluing operator $\mathcal{F}_L\in \mathcal{B}(\mathcal{X}_{L})$ specified by the following kernel:
\begin{equation*}
\mathbb{C}^{d\times d}\owns \mathcal{F}_L(\bm{n},\bm{m}):=\left\{
\begin{aligned}
&\mathcal{H}_0(\bm{n},\bm{n}\pm\bm{e}_1),\quad n_1=\pm L,\, \bm{m}=\bm{n}\mp 2L\bm{e}_1, \\
&\mathcal{H}_0(\bm{n},\bm{n}\pm\bm{e}_2),\quad n_2=\pm L,\, \bm{m}=\bm{n}\mp 2L\bm{e}_2, \\
&0,\quad \text{else} .
\end{aligned}
\right.
\end{equation*}
With this notation, the box Hamiltonian with periodic boundary conditions is defined by adding the gluing operator to $\mathcal{H}_{\omega,L}^{sim}$ as follows:
\begin{equation} \label{eq_box_Hamiltonian_periodic}
\mathcal{H}_{\omega,L}^{\sharp}:=\iota^{*}_{\Lambda_L}\mathcal{H}_{\omega}\iota_{\Lambda_L}+\mathcal{F}_L\in \mathcal{B}(\mathcal{X}_{L}).
\end{equation}
Intuitively, the gluing operator $\mathcal{F}_L$ adds the hopping across the boundary to $\mathcal{H}_{\omega,L}^{sim}$, which makes the underlying domain a torus.

The first result estimates the probability of $\mathcal{H}_{\omega,L}^{\sharp}$ having a non-empty spectrum near the gap opening point $\lambda=\lambda_0$.

\begin{proposition}[Lifshitz-type estimate] \label{prop_lifshitz_tail}
For any $\beta >0$, there exist $C,L_0>0$ such that
\begin{equation} \label{eq_lifshitz_tail}
\mathbb{P}\big( \sigma(\mathcal{H}_{\omega,L}^{\sharp})\cap (\lambda_0-L^{-\beta},\lambda_0+L^{-\beta})\neq \emptyset \big)\leq CL^{2-\eta\beta}
\end{equation}
for any $L>L_0$. The exponent $\eta>0$ is the same as in Assumption \ref{asmp_distribution_density}.
\end{proposition}
\begin{remark} \label{rmk_beta_size}
Proposition \ref{prop_lifshitz_tail} demonstrates that eigenvalues near the gap opening point are rare when $\eta>0$ is sufficiently large, as a key element to prove the Anderson localization in the sense of the Aizenman-Molchanov regime.\footnote{A sufficient size of $\eta$ for proving Theorem \eqref{thm_anderson_localization}, i.e. $\eta>5$, is supposed in Assumption \ref{asmp_distribution_density} to guarantee the Anderson localization.} The availability of a scale $\beta <1$ is important. In fact, if one discards Assumption \ref{asmp_distribution_density} and assumes more information on the band structure of deterministic Hamiltonians $\mathcal{H}_0$ and $\mathcal{H}_{per}$, such as the effect of gap opening potential on the Bloch modes, one can prove an analogous estimate to \eqref{eq_lifshitz_tail} but only for $\beta\geq 2$.\footnote{This can indeed be done for example for the QWZ model by following the lines of \cite[Section 2.1]{stollmann2001caught}; see the discussion in Appendix \ref{appendixA}. Details are left for the interested reader.} However, the size of rare region for eigenvalues specified by $\beta\geq 2$ is not sufficient to prove Theorem \ref{thm_anderson_localization}. If the hopping part of the Hamiltonian is special, e.g., $\mathcal{H}_{0}$ being the discrete Laplacian, the Dirichlet-Neumann bracketing method is available to improve the estimates on the size of the rare region; see \cite[Section 2.2]{kirsch2007invitation}. Unfortunately, it is necessary to include more general hopping terms to make the Hamiltonian topologically non-trivial. That is the main reason we assume that \eqref{eq_poly_distribution} holds.
\end{remark}

\begin{proof}[Proof of Proposition \ref{prop_lifshitz_tail}]
{\color{blue}Step 1.} We introduce the following linear family of Hamiltonians:
\begin{equation*}
    \mathcal{H}_{\omega,L}^{\sharp}(t):=\iota_{\Lambda_L}^{*}(\mathcal{H}_{0}+1\otimes V+V_{\omega,t})\iota_{\Lambda_L}+\mathcal{F}_L
\end{equation*}
with
\begin{equation*}
(V_{\omega,t}\psi)(\bm{n}):=-t\cdot \text{diag}\big((1-\omega_{\bm{n}}^{(1)})V_1,\cdots ,(1-\omega_{\bm{n}}^{(d)}V_d)\big)\cdot \psi(\bm{n}) .
\end{equation*}
Apparently,
\begin{equation*}
  \mathcal{H}_{\omega,L}^{\sharp}(0)=\iota_{\Lambda_L}^{*}\mathcal{H}_{per}\iota_{\Lambda_L}+\mathcal{F}_L:=\mathcal{H}_{per,L}^{\sharp}
\end{equation*}
is the restriction of the periodic Hamiltonian $\mathcal{H}_{per}$ to $\Lambda_L$ with periodic boundary conditions and 
\begin{equation*}
  \mathcal{H}_{\omega,L}^{\sharp}(1)=\iota_{\Lambda_L}^{*}(\mathcal{H}_{0}+V_{\omega})\iota_{\Lambda_L}+\mathcal{F}_L=\mathcal{H}_{\omega,L}^{\sharp}.
\end{equation*}
In other words, $\mathcal{H}_{\omega,L}^{\sharp}(t)$ connects $\mathcal{H}_{per,L}^{\sharp}$ linearly to $\mathcal{H}_{\omega,L}^{\sharp}$. Since $\mathcal{H}_{\omega,L}^{\sharp}(t)$ is self-adjoint and has a discrete spectrum, each eigenvalue $\lambda\in \sigma(\mathcal{H}_{\omega,L}^{\sharp})$ is associated with an analytic path $\lambda(t)$ such that
\begin{equation*}
   \lambda(1)=\lambda\in \sigma(\mathcal{H}_{\omega,L}^{\sharp}),\quad \lambda(0)\in \sigma(\mathcal{H}_{per,L}^{\sharp}),
\end{equation*}
by analytic perturbation theory (cf. \cite[Chapter \rom{7}]{kato2013perturbation}).

{\color{blue}Step 2.} Suppose that there exists $\lambda\in \sigma(\mathcal{H}_{\omega,L}^{\sharp})\cap (\lambda_0-L^{-\beta},\lambda_0+L^{-\beta})$. By Assumption \ref{asmp_Dirac_eigenpair} and the discussion in Step 1, we can find an analytic family of eigenvalues $\lambda(t)$ that connects $(\lambda_0-\Delta,\lambda_0+\Delta)^c$ to $(\lambda_0-L^{-\beta},\lambda_0+L^{-\beta})$ (recall that $\sigma(\mathcal{H}_{per,L}^{\sharp})\subset \sigma(\mathcal{H}_{per})$ thanks to periodic boundary conditions), which means that
\begin{equation*}
\Delta-L^{-\beta}\leq \big|\int_{0}^{1}\lambda^{\prime}(t)dt \big|\leq \sup_{t\in [0,1]} |\lambda^{\prime}(t)|.
\end{equation*}
Moreover, if we denote the analytic family of normalized eigenmodes associated with $\lambda(t)$ by $\psi(t)$, the Feynman-Hellmann formula indicates that, for any $t\in [0,1]$,
\begin{equation*}
|\lambda^{\prime}(t)|=\Big|\big\langle (\partial_t\mathcal{H}_{\omega,L}^{\sharp}(t))\psi(t),\psi(t) \big\rangle_{\mathcal{X}_{L}}\Big|=\Big|\sum_{\bm{n}\in \Lambda_L,1\leq i\leq d}(1-\omega_{\bm{n}}^{(i)})V_i|\psi(t)_{\bm{n}}^{(i)}|^2\Big|
\leq \Delta \sup_{\bm{n}\in \Lambda_L,1\leq i\leq d}|1-\omega_{\bm{n}}^{(i)}|,
\end{equation*}
where $\sup_{1\leq i\leq d}|V_i|\leq \Delta$ (Assumption \ref{asmp_Dirac_eigenpair}) is applied to derive the last inequality. In conclusion, we see that
\begin{equation*}
\begin{aligned}
&\mathbb{P}\big( \sigma(\mathcal{H}_{\omega,L}^{\sharp})\cap (\lambda_0-L^{-\beta},\lambda_0+L^{-\beta})\neq \emptyset \big) \\
&\leq \mathbb{P}\big( \sup_{\bm{n}\in \Lambda_L,1\leq i\leq d}|1-\omega_{\bm{n}}^{(i)}|\geq 1-\Delta^{-1}L^{-\beta} \big) \\
&= \mathbb{P}\big( \omega_{\bm{n}}^{(i)}\leq \Delta^{-1}L^{-\beta}\quad \text{for some $\bm{n}\in \Lambda_L$ and $i\in\{1,2,\cdots,d\}$} \big) \\
&\leq 4dL^2 \mathbb{P}\big( \omega_{\bm{0}}^{(1)}\leq \Delta^{-1}L^{-\beta} \big).
\end{aligned}
\end{equation*}
Using \eqref{eq_poly_distribution} in Assumption \ref{asmp_distribution_density}, we obtain
\begin{equation*}
\mathbb{P}\big( \sigma(\mathcal{H}_{\omega,L}^{\sharp})\cap (\lambda_0-L^{-\beta},\lambda_0+L^{-\beta})\neq \emptyset \big) \leq CL^{2-\eta\beta},
\end{equation*}
whenever $L>L_0:=1/(\Delta\tau)^{1/\beta}$. This completes the proof of Proposition \ref{prop_lifshitz_tail}.
\end{proof}

The next key result is the so-called initial length estimate of the fractional moment of the box Green function, which characterizes the initial scale starting from which the off-diagonal part of the Green function becomes small. Before showing the result, we first recall the following a priori estimate of the fractional moment of the Green function.
\begin{lemma} \label{lem_apriori_bound_L}
For any $s\in (0,1)$, there exists $C>0$ such that
\begin{equation*}
\sup_{z\in\mathbb{C}\backslash \mathbb{R}}\mathbb{E}_{\bm{n},\bm{m}}\big(\|G_{\omega,L}^{\text{sim} /\sharp}(\bm{n},\bm{m};z)\|^s \big)\leq C<\infty,
\end{equation*}
which holds for all $L>0$ and $\bm{n},\bm{m}\in\Lambda_L$. Here, $\mathbb{E}_{\bm{n},\bm{m}}$ denotes the conditional expectation with $\{\omega_{\bm{\ell}}^{(i)}\}_{\bm{\ell}\in \mathbb{Z}^2\backslash \{\bm{n},\bm{m}\}}$ fixed, and $G_{\omega,L}^{\text{sim} /\sharp}(\bm{n},\bm{m};z):=(\mathcal{H}_{\omega,L}^{\text{sim} /\sharp}-z)^{-1}(\bm{n},\bm{m})$. Consequently,
\begin{equation*}
\sup_{z\in\mathbb{C}\backslash \mathbb{R}}\mathbb{E}\big(\|G_{\omega,L}^{\text{sim} /\sharp}(\bm{n},\bm{m};z)\|^s \big)\leq C<\infty.
\end{equation*}
\end{lemma}
We also need its whole-space version.
\begin{lemma} \label{lem_apriori_bound_whole}
For any $s\in (0,1)$, there exists $C>0$ such that
\begin{equation*}
\sup_{z\in\mathbb{C}\backslash \mathbb{R}}\mathbb{E}_{\bm{n},\bm{m}}\big(\|G_{\omega}(\bm{n},\bm{m};z)\|^s \big)\leq C<\infty,
\end{equation*}
which holds for all $\bm{n},\bm{m}\in \mathbb{Z}^2$. Consequently,
\begin{equation*}
\sup_{z\in\mathbb{C}\backslash \mathbb{R}}\mathbb{E}\big(\|G_{\omega}(\bm{n},\bm{m};z)\|^s \big)\leq C<\infty.
\end{equation*}
\end{lemma}
The proof of these lemmas can be found in \cite[Proposition 4.1]{prado2017dynamical}; see also \cite[Section 3]{stolz2011introduction}. Heuristically, Lemmas \ref{lem_apriori_bound_L} and \ref{lem_apriori_bound_whole} hold because the exponent $s<1$ regularizes the non-integrability where the Green function blows up (i.e., a resonance occurs). We refer the reader to \cite{aizenman1993localization_elementary_derivation,aizenman2001finite_volime,aizenman1994localization} for a more detailed discussion.

Now, we are in a position to state the second main result of this section.
\begin{proposition}[Initial length estimate]
\label{prop_initial_length}
For every $s,\beta\in (0,1)$ and $q>1$, there exist $C,L_1>0$ such that
\begin{equation} \label{eq_initial_length}
\mathbb{E}\big(\|G_{\omega,L}^{\text{sim}}(0,\bm{n};\lambda+i\epsilon)\|^s \big)\leq C L^{(2-\eta\beta)/q}
\end{equation}
for all $L>L_1$, $\bm{n}\in \partial^{int}\Lambda_L$, $\lambda\in [\lambda_0-L^{-\beta}/2,\lambda_0+L^{-\beta}/2]$, and $\epsilon>0$.
\end{proposition}

\begin{proof}
We define the following event:
\begin{equation*}
\Omega_{B}:=\{\omega:\, \sigma(\mathcal{H}_{\omega,L}^{\sharp})\cap (\lambda_0-L^{-\beta},\lambda_0+L^{-\beta})\neq \emptyset \}
\end{equation*}
and let $\Omega_{B}^{c}$ be its complement. Then the estimate of \eqref{eq_initial_length} is decomposed into two parts:
\begin{equation} \label{eq_initial_length_proof_1}
\mathbb{E}\big(\|G_{\omega,L}^{\text{sim}}(0,\bm{n};\lambda+i\epsilon)\|^s \big)=\mathbb{E}\big(\|G_{\omega,L}^{\text{sim}}(0,\bm{n};\lambda+i\epsilon)\|^s\mathbbm{1}_{\Omega_{B}} \big)+\mathbb{E}\big(\|G_{\omega,L}^{\text{sim}}(0,\bm{n};\lambda+i\epsilon)\|^s\mathbbm{1}_{\Omega_{B}^{c}} \big).
\end{equation}
Using Proposition \ref{prop_lifshitz_tail}, Lemma \ref{lem_apriori_bound_L} and Hölder's inequality, the first part of \eqref{eq_initial_length_proof_1} is estimated as
\begin{equation}
\label{eq_initial_length_proof_2}
\begin{aligned}
\mathbb{E}\big(\|G_{\omega,L}^{\text{sim}}(0,\bm{n};\lambda+i\epsilon)\|^s\mathbbm{1}_{\Omega_{B}} \big)
&\leq \mathbb{E}\big(\|G_{\omega,L}^{\text{sim}}(0,\bm{n};\lambda+i\epsilon)\|^{sp} \big)^{1/p} \big(\mathbb{P}(\Omega_{B})\big)^{1/q} \\
\nm
&\leq CL^{(2-\eta\beta)/q},
\end{aligned}
\end{equation}
where $q^{-1}+p^{-1}=1$ and $C=C(s,\beta,q)>0$. The estimate of the second part is more sophisticated. We first note that,  by the resolvent identity
\begin{equation*}
(\mathcal{H}_{\omega,L}^{\text{sim}}-z)^{-1}=(\mathcal{H}_{\omega,L}^{\sharp}-z)^{-1}+(\mathcal{H}_{\omega,L}^{\sharp}-z)^{-1}\mathcal{F}_{L}(\mathcal{H}_{\omega,L}^{\text{sim}}-z)^{-1},
\end{equation*}
it holds that
\begin{equation*}
G_{\omega,L}^{\text{sim}}(0,\bm{n};\lambda+i\epsilon)
=G_{\omega,L}^{\sharp}(0,\bm{n};\lambda+i\epsilon)+\sum_{(\bm{m},\bm{m}^{\prime})\in \text{supp} (\mathcal{F}_L)} G_{\omega,L}^{\sharp}(0,\bm{m};\lambda+i\epsilon)\mathcal{F}_L(\bm{m},\bm{m}^{\prime})G_{\omega,L}^{\text{sim}}(\bm{m}^{\prime},\bm{n};\lambda+i\epsilon).
\end{equation*}
Using the fact that $\|\mathcal{F}_L(\bm{m},\bm{m}^{\prime})\|\leq \sup_{\bm{n},\bm{m}}\|\mathcal{H}_{0}(\bm{n},\bm{m})\|<\infty$ and the elementary inequality $$(\sum|x_i|)^s\leq \sum |x_i|^{s} \quad  \text{for } s\in (0,1),$$ we obtain
\begin{equation}
\label{eq_initial_length_proof_3}
\begin{array}{lll}
\|G_{\omega,L}^{\text{sim}}(0,\bm{n};\lambda+i\epsilon)\|^s &\leq &\|G_{\omega,L}^{\sharp}(0,\bm{n};\lambda+i\epsilon)\|^s\\
\nm 
&& \ds +C\sum_{(\bm{m},\bm{m}^{\prime})\in \text{supp} (\mathcal{F}_L)}\|G_{\omega,L}^{\sharp}(0,\bm{m};\lambda+i\epsilon)\|^s\|G_{\omega,L}^{\text{sim}}(\bm{m}^{\prime},\bm{n};\lambda+i\epsilon)\|^s.
\end{array}
\end{equation}
Note that when $\sigma(\mathcal{H}_{\omega,L}^{\sharp})\cap (\lambda_0-L^{-\beta},\lambda_0+L^{-\beta})= \emptyset$, it holds that 
$$
\inf_{\lambda\in [\lambda_0-L^{-\beta}/2,\lambda_0+L^{-\beta}/2],\epsilon>0}\text{dist}(\lambda+i\epsilon,\sigma(\mathcal{H}_{\omega,L}^{\sharp}))\geq L^{-\beta}/2.
$$
Hence, we can bound $G_{\omega,L}^{\sharp}(0,\bm{n};\lambda+i\epsilon)$ using the standard Combes-Thomas estimate as in Lemma \ref{lem_resolv_CT_estimate}:
\begin{equation*}
\|G_{\omega,L}^{\sharp}(0,\bm{n};\lambda+i\epsilon)\|^s\mathbbm{1}_{\Omega_B^c}\leq \frac{C}{L^{-\beta}}e^{-\frac{1}{2}L^{-\beta}|\bm{n}|}\leq \frac{C}{L^{-\beta}}e^{-\frac{\sqrt{2}}{2}L^{1-\beta}}.
\end{equation*}
Therefore, taking the expectation over \eqref{eq_initial_length_proof_3} and recalling that $\text{supp}(\mathcal{F}_L)\subset \partial^{int}\Lambda_L\times \partial^{int}\Lambda_L$, we see that
\begin{equation}
\label{eq_initial_length_proof_4}
\begin{aligned}
\mathbb{E}\big(\|G_{\omega,L}^{\text{sim}}(0,\bm{n};\lambda+i\epsilon)\|^s\mathbbm{1}_{\Omega_B^c}\big)
&\leq \frac{C}{L^{-\beta}}e^{-\frac{\sqrt{2}}{2}L^{1-\beta}}\mathbb{E}\Big(1
+\sum_{(\bm{m},\bm{m}^{\prime})\in \text{supp} (\mathcal{F}_L)}\|G_{\omega,L}^{\text{sim}}(\bm{m}^{\prime},\bm{n};\lambda+i\epsilon)\|^s \Big) \\
&\leq CL^{1+\beta}e^{-\frac{\sqrt{2}}{2}L^{1-\beta}},
\end{aligned}
\end{equation}
where $C=C(s)>0$, and Lemma \ref{lem_apriori_bound_L} is applied to derive the last inequality. Thus, by \eqref{eq_initial_length_proof_1}-\eqref{eq_initial_length_proof_2}, \eqref{eq_initial_length_proof_4}, and choosing $L_1>L_0$ sufficiently large such that $$L^{1+\beta}e^{-\frac{1}{2}L^{1-\beta}}\leq L^{(2-\eta\beta)/q},$$ we conclude the proof of Proposition \ref{prop_initial_length}.
\end{proof}

\subsection{Geometric Decoupling: Proof of Theorem \ref{thm_anderson_localization}}
\label{sec_geo_decouple}

Based on the two key estimates obtained in the last section, we prove Theorem \ref{thm_anderson_localization} following the standard geometric decoupling argument. This elegant method, introduced by Aizenman and Molcanov \cite{aizenman1993localization_elementary_derivation}, has been applied to study Anderson localization in various settings; see, e.g., \cite[Section 7]{stolz2011introduction} and \cite{aizenman2001finite_volime,prado2017dynamical}.

By translation invariance, it suffices to prove Theorem \ref{thm_anderson_localization} for $\bm{n}=0$. To do this, we first introduce the Hamiltonian
\begin{equation*}
\mathcal{H}_{\omega}^{(L)}:=\mathcal{H}_{\omega,L}^{\text{sim}}\oplus \mathcal{H}_{\omega,L^c}^{\text{sim}} \quad \text{with} \quad  \mathcal{H}_{\omega,L^c}^{\text{sim}}:=\iota_{\Lambda_L^{c}}^{*} \mathcal{H}_{\omega}\iota_{\Lambda_L^{c}},
\end{equation*}
i.e., $\mathcal{H}_{\omega}^{(L)}$ ignores the hopping between the box $\Lambda_L$ and its exterior. It is linked to the original Hamiltonian $\mathcal{H}_{\omega}$ by the relation
\begin{equation*}
\mathcal{H}_{\omega}=\mathcal{H}_{\omega}^{(L)}+\mathcal{F}^{(L)},
\end{equation*}
where $\mathcal{F}^{(L)}$ is already introduced in Section \ref{sec_bec}. (Note that $\mathcal{F}^{(L)}$ is different from the gluing operator $\mathcal{F}_{L}$ defined in Section \ref{sec_lif_initial}). Using twice the resolvent identity, we see that
\begin{equation*}
\begin{aligned}
(\mathcal{H}_{\omega}-z)^{-1}
&=(\mathcal{H}_{\omega}^{(L)}-z)^{-1}-(\mathcal{H}_{\omega}^{(L)}-z)^{-1}\mathcal{F}^{(L)}(\mathcal{H}_{\omega}-z)^{-1} \\
&=(\mathcal{H}_{\omega}^{(L)}-z)^{-1}-(\mathcal{H}_{\omega}^{(L)}-z)^{-1}\mathcal{F}^{(L)}(\mathcal{H}_{\omega}^{(L+1)}-z)^{-1} \\
&\quad +(\mathcal{H}_{\omega}^{(L)}-z)^{-1}\mathcal{F}^{(L)}(\mathcal{H}_{\omega}-z)^{-1}\mathcal{F}^{(L+1)}(\mathcal{H}_{\omega}^{(L+1)}-z)^{-1}.
\end{aligned}
\end{equation*}
The above identity is rewritten in the form of Green's functions:
\begin{equation} \label{eq_anderson_loc_proof_1}
\begin{aligned}
G_{\omega}(0,\bm{m};z)
&=G^{(L)}_{\omega}(0,\bm{m};z)-\sum_{(\bm{u},\bm{u}^{\prime})\in \text{supp}(\mathcal{F}^{(L)})}G_{\omega}^{(L)}(0,\bm{u};z)\mathcal{F}^{(L)}(\bm{u},\bm{u}^{\prime})G_{\omega}^{(L+1)}(\bm{u}^{\prime},\bm{m};z) \\
&\quad +\sum_{\substack{(\bm{u},\bm{u}^{\prime})\in \text{supp}(\mathcal{F}^{(L)}) \\ (\bm{v},\bm{v}^{\prime})\in \text{supp}(\mathcal{F}^{(L+1)})}}G_{\omega}^{(L)}(0,\bm{u};z)\mathcal{F}^{(L)}(\bm{u},\bm{u}^{\prime})G_{\omega}(\bm{u}^{\prime},\bm{v};z)\mathcal{F}^{(L)}(\bm{v},\bm{v}^{\prime})G_{\omega}^{(L+1)}(\bm{v}^{\prime},\bm{m};z),
\end{aligned}
\end{equation}
where $G^{(L)}_{\omega}(\bm{n},\bm{m};z):=\langle\mathbbm{1}_{\{\bm{n}\}},(\mathcal{H}_{\omega}^{(L)}-z)^{-1} \mathbbm{1}_{\{\bm{m}\}}  \rangle_{\ell^2(\mathbb{Z}^2)}$.
Importantly, we note that if $\bm{m} \in  \Lambda_{L+1}^{c}$, then 
\begin{equation*}
G^{(L)}_{\omega}(0,\bm{m};z)=0,
\end{equation*}
and
\begin{equation*}
G_{\omega}^{(L+1)}(\bm{u}^{\prime},\bm{m};z)=0 \quad (\bm{u}^{\prime} \in  \Lambda_{L+1})
\end{equation*}
because $\Lambda_L$ decouples from its exterior $\Lambda_L^{c}$ under the action of $\mathcal{H}_{\omega}^{(L)}$. This implies that the first two terms in \eqref{eq_anderson_loc_proof_1} vanish. For the same reason, we can replace $G_{\omega}^{(L)}$ and $G_{\omega}^{(L+1)}$ in the third term of \eqref{eq_anderson_loc_proof_1} by $G_{\omega,L}^{\text{sim}}$ and $G_{\omega,L+1^c}^{\text{sim}}$, respectively (the latter is the Green function associated with $\mathcal{H}_{\omega,L+1^c}^{\text{sim}}$). Hence, we conclude that
\begin{equation*}
G_{\omega}(0,\bm{m};z)
=\sum_{\substack{(\bm{u},\bm{u}^{\prime})\in \text{supp}(\mathcal{F}^{(L)}) \\ (\bm{v},\bm{v}^{\prime})\in \text{supp}(\mathcal{F}^{(L+1)})}}G_{\omega,L}^{\text{sim}}(0,\bm{u};z)\mathcal{F}^{(L)}(\bm{u},\bm{u}^{\prime})G_{\omega}(\bm{u}^{\prime},\bm{v};z)\mathcal{F}^{(L)}(\bm{v},\bm{v}^{\prime})G_{\omega,L+1^c}^{\text{sim}}(\bm{v}^{\prime},\bm{m};z).
\end{equation*}
Using the fact that $\|\mathcal{F}^{(L)}(\bm{n},\bm{m})\|\leq \sup_{\bm{n},\bm{m}}\|\mathcal{H}_{0}(\bm{n},\bm{m})\|<\infty$, the above equation leads to
\begin{equation} \label{eq_anderson_loc_proof_2}
\|G_{\omega}(0,\bm{m};z)\|^{s}
\leq C\sum_{\substack{(\bm{u},\bm{u}^{\prime})\in \partial \Lambda_L \\ (\bm{v},\bm{v}^{\prime})\in \partial \Lambda_{L+1}}}\|G_{\omega,L}^{\text{sim}}(0,\bm{u};z) \|^{s} \|G_{\omega}(\bm{u}^{\prime},\bm{v};z)\|^{s} \|G_{\omega,L+1^c}^{\text{sim}}(\bm{v}^{\prime},\bm{m};z)\|^{s} ,
\end{equation}
for any $s\in (0,1)$. Remark that at the right side of \eqref{eq_anderson_loc_proof_2}, only the middle term $\|G_{\omega}(\bm{u}^{\prime},\bm{v};z)\|^{s}$ depends on the random variables supported on the sites $\bm{u}^{\prime}$ and $\bm{v}$, and the other two terms are stochastically independent of each other. Hence, when taking expectation over \eqref{eq_anderson_loc_proof_2}, we first calculate the conditional expectation $\mathbb{E}_{\bm{u}^{\prime},\bm{v}}$ and utilize Lemma \ref{lem_apriori_bound_whole} to bound the middle term, which leads to the result
\begin{equation*}
\mathbb{E}\big(\|G_{\omega}(0,\bm{m};z)\|^{s}\big)
\leq C\sum_{\bm{u}\in \partial^{int} \Lambda_L,\, \bm{v}^{\prime}\in \partial^{ext} \Lambda_{L+1}}\mathbb{E}\big(\|G_{\omega,L}^{\text{sim}}(0,\bm{u};z) \|^{s} \big) \mathbb{E}\big( \|G_{\omega,L+1^c}^{\text{sim}}(\bm{v}^{\prime},\bm{m};z)\|^{s} \big)
\end{equation*}
with $C=C(s)>0$. By bounding $G_{\omega,L}^{\text{sim}}$ with the initial length estimate presented in Proposition \ref{prop_initial_length}, we obtain the following:
\begin{equation} \label{eq_anderson_loc_proof_3}
\mathbb{E}\big(\|G_{\omega}(0,\bm{m};\lambda+i\epsilon)\|^{s}\big)
\leq CL^{1+(2-\eta\beta)/q}\sum_{\bm{v}^{\prime}\in \partial^{ext} \Lambda_{L+1}}\mathbb{E}\big( \|G_{\omega,L+1^c}^{\text{sim}}(\bm{v}^{\prime},\bm{m};\lambda+i\epsilon)\|^{s} \big)
\end{equation}
with $\beta\in (0,1)$ and $q>1$ to be determined and $C=C(s,\beta,q)>0$. Note that \eqref{eq_anderson_loc_proof_3} holds for all $L>L_1$ with $L_1=L_1(s,\beta,q)>0$ and $\lambda\in [\lambda_0-\frac{1}{2}L^{-\beta},\lambda_0+\frac{1}{2}L^{-\beta}]$. We want to use \eqref{eq_anderson_loc_proof_3} as the first step of the iteration to prove the exponential localization of $\mathbb{E}\big(\|G_{\omega}(0,\bm{m};\lambda+i\epsilon)\|^{s}\big)$. This requires the following result, which helps to bound the depleted Green function $G_{\omega,L+1^c}^{\text{sim}}$ in terms of the full Green function $G_{\omega}$. The proof follows the same lines as in \cite[Lemma 10.2]{prado2017dynamical} (see also \cite[Lemma 2.3]{aizenman2001finite_volime}).
\begin{lemma} \label{lem_depleted_Green_bound_with_full}
For any $s\in (0,1)$, there exists $C>0$ such that
\begin{equation*}
\|G_{\omega,L+1^c}^{\text{sim}}(\bm{v},\bm{m};z)\|^{s}\leq C\sum_{\bm{v}^{\prime}\in \partial^{ext}\Lambda_{L+1}}\|G_{\omega}(\bm{v}^{\prime},\bm{m};z)\|^{s}
\end{equation*}
for all $z\in\mathbb{C}\backslash \mathbb{R}$, $L\geq 1$, $\bm{v}\in  \partial^{ext}\Lambda_{L+1}$, and $\bm{m}\in\Lambda_{L+1}^{c}$.
\end{lemma}
With Lemma \ref{lem_depleted_Green_bound_with_full}, we derive from \eqref{eq_anderson_loc_proof_3} that
\begin{equation} \label{eq_anderson_loc_proof_4}
\mathbb{E}\big(\|G_{\omega}(0,\bm{m};\lambda+i\epsilon)\|^{s}\big)
\leq CL^{3+(2-\eta\beta)/q}\sup_{\bm{v}^{\prime}\in  \Lambda_{L+2}}\mathbb{E}\big( \|G_{\omega}(\bm{v}^{\prime},\bm{m};\lambda+i\epsilon)\|^{s} \big).
\end{equation}
Inequality \eqref{eq_anderson_loc_proof_4} facilitates the iteration to prove the exponential decay of the left side. Specifically, since $\eta>5$ by Assumption \ref{asmp_distribution_density}, we can choose $\beta\in (0,1)$ and $q>1$ such that $3+(2-\eta\beta)/q<0$ and fix $L>1$ such that $r:=CL^{3+(2-\eta\beta)/q}<1$. We also set $\delta:=\frac{1}{2}L^{-\beta}$ (the radius of the mobility gap). For $\lambda\in [\lambda_0-\delta,\lambda_0+\delta]$, we can use \eqref{eq_anderson_loc_proof_4} to start an iteration
\begin{equation*}
\mathbb{E}\big(\|G_{\omega}(\bm{v}^{\prime},\bm{m};\lambda+i\epsilon)\|^{s}\big)
\leq r\sup_{\bm{v}^{\prime\prime}\in  \Lambda_{2(L+2)}}\mathbb{E}\big( \|G_{\omega}(\bm{v}^{\prime\prime},\bm{m};\lambda+i\epsilon)\|^{s} \big),
\end{equation*}
and so forth. This iteration can be carried out in approximately $|\bm{m}|/L$ steps until the chain $\bm{v}^{\prime},\bm{v}^{\prime\prime},\cdots$ may hit $\bm{m}\in \Lambda_{L+1}^{c}$. Then, by bounding the Green function in the final step using the a priori bound in Lemma \ref{lem_apriori_bound_whole}, we conclude the exponential decay of $\mathbb{E}\big(\|G_{\omega}(\bm{v}^{\prime},\bm{m};\lambda+i\epsilon)\|^{s}\big)$ with the decay rate being $\alpha:=|\log(r)|/L$.

\section*{Acknowledgments} This work was supported in part by the Swiss National Science Foundation grant number 200021-236472. It was initiated while H.A. was visiting the Hong Kong Institute for Advanced Study as a Senior Fellow. We would like to thank Guillaume Bal and Gian Michele Graf for discussions on the extension of bulk and edge indices defined in this paper.

\appendix

\section{The Qi-Wu-Zhang Model} \label{appendixA}

\setcounter{equation}{0}
\setcounter{subsection}{0}
\setcounter{theorem}{0}
\renewcommand{\theequation}{A.\arabic{equation}}
\renewcommand{\thesubsection}{A.\arabic{subsection}}
\renewcommand{\thetheorem}{A.\arabic{theorem}}

As a prototypical model for (quantum) anomalous Hall physics \cite{ezawa2024nonlinear_phase_transition,qi2006qwz_model}, the Qi-Wu-Zhang model describes the motion of electrons on a square lattice consisting of two sublattices with nearest-neighbor hopping taken into account.\footnote{Or equivalently, one can write it into the next-nearest-neighbor hopping form by expanding the sublattices into a single lattice.} Specifically, the unperturbed Hamiltonian of the QWZ model is given by
\begin{equation} \label{eq_qwz_h0}
\begin{aligned}
(\mathcal{H}_{0}\psi)(\bm{n})=&\frac{1}{2}(\sigma_z-i\sigma_x)\psi(\bm{n}+\bm{e}_1)+\frac{1}{2}(\sigma_z+i\sigma_x)\psi(\bm{n}-\bm{e}_1) \\
&+\frac{1}{2}(\sigma_z-i\sigma_y)\psi(\bm{n}+\bm{e}_2)+\frac{1}{2}(\sigma_z+i\sigma_y)\psi(\bm{n}-\bm{e}_2),
\end{aligned}
\end{equation}
where $\sigma$ denotes the Pauli matrices. The operator $\mathcal{H}_{0}$ satisfies Assumption \ref{asmp_Hamiltonian_deter_part}. Moreover, it is easy to calculate that $\mathcal{H}_{0}$ attains two Dirac points at $\lambda=0$ and $\bm{K}=(0,\pi),(\pi,0)$, where the Bloch eigenmodes are
\begin{equation*}
\psi_1^{\bm{K}}(\bm{n})=(1,0)^{T}\cdot e^{i\bm{K}\cdot\bm{n}},\quad \psi_2^{\bm{K}}(\bm{n})=(0,1)^{T}\cdot e^{i\bm{K}\cdot\bm{n}}.
\end{equation*}
A gap lifting potential is given by the diagonal matrix $V=\delta\sigma_z$ with $0<|\delta|<2$. It associates opposite masses to $\psi_1^{\bm{K}}$ and $\psi_2^{\bm{K}}$:
\begin{equation} \label{eq_qwz_non_degenerate}
(V\psi_1^{\bm{K}},\psi_1^{\bm{K}})=\delta,\quad (V\psi_2^{\bm{K}},\psi_2^{\bm{K}})=-\delta,\quad
(V\psi_1^{\bm{K}},\psi_2^{\bm{K}})=(V\psi_2^{\bm{K}},\psi_1^{\bm{K}})=0.
\end{equation}
For $0<|\delta|<2$, it is directly verified that the periodic perturbed Hamiltonian $\mathcal{H}_{per}=\mathcal{H}_0+1\otimes V$ possesses a band gap near $\lambda_0=0$; hence, Assumption \ref{asmp_Dirac_eigenpair} is satisfied. Consequently, setting the random variables as in Assumptions \ref{asmp_random_potential} and \ref{asmp_distribution_density}, the analysis in the main text applies readily to the QWZ model. We also note that the Chern number is calculated explicitly \cite{ezawa2024nonlinear_phase_transition} as follows:
\begin{equation*}
\mathcal{C}_1(\mathbbm{1}_{(-\infty,0)}(\mathcal{H}_{per}))=
\left\{
\begin{aligned}
&1,\quad 0<\delta<2,\\
&-1,\quad -2<\delta<0.
\end{aligned}
\right.
\end{equation*}
\begin{remark}
As anticipated in Remark \ref{rmk_beta_size}, by identity \eqref{eq_qwz_non_degenerate} and analytic perturbation theory, one can follow the lines of \cite[Section 2.1]{stollmann2001caught} to prove an analogous Lifschitz-type estimate to Proposition \ref{prop_lifshitz_tail} without supposing Assumption \ref{asmp_distribution_density}, but only for $\beta\geq 2$. Unfortunately, the unperturbed Hamiltonian \eqref{eq_qwz_h0}, unlike the Laplacian operator, does not allow the Dirichlet-Neumann bracketing method to improve the Lifschitz estimate.
\end{remark}

\section{Nontriviality of the Bulk Index} \label{appendixB}

\setcounter{equation}{0}
\setcounter{subsection}{0}
\setcounter{theorem}{0}
\renewcommand{\theequation}{B.\arabic{equation}}
\renewcommand{\thesubsection}{B.\arabic{subsection}}
\renewcommand{\thetheorem}{B.\arabic{theorem}}

In this appendix, we prove that the topological number $\mathcal{E}_{bulk,\omega,\lambda}$ is equal to the first Chern number associated with the spectral projection $\mathbbm{1}_{(-\infty,\lambda_0)}(\mathcal{H}_{per})$. Specifically, we have the following result.
\begin{theorem} \label{thm_invariant_index}
Under Assumptions \ref{asmp_Dirac_eigenpair} and \ref{asmp_distribution_density} and for all $\lambda\in\mathcal{I}$ with $\mathcal{I}$ being the same as in Theorem \ref{thm_anderson_localization}, it holds $\mathbb{P}-$almost surely that
\begin{equation*}
\mathcal{E}_{bulk,\omega,\lambda}=\mathcal{C}_1(\mathbbm{1}_{(-\infty,\lambda_0)}(\mathcal{H}_{per}))\neq 0.
\end{equation*}
\end{theorem}
As one expects, the proof follows a standard deformation argument: as in \cite{Schulz-Baldes01homotopy}, we construct a continuous path from $\mathcal{H}_{\omega}$ to the deterministic Hamiltonian $\mathcal{H}_{per}$ such that \textit{the interval $\mathcal{I}=(\lambda_0-\delta,\lambda_0+\delta)$ introduced in Theorem \ref{thm_anderson_localization} is always an Aizenman-Molchanov mobility gap of the deformed Hamiltonian along the path}. Then, by the continuity and integrality of the bulk index, we conclude the proof of Theorem \ref{thm_invariant_index}.

\begin{remark}
Theorem \ref{thm_invariant_index} indicates that the bulk medium is topologically nontrivial even in the presence of disorder. We point out that this result only holds under a weak disorder. In fact, let us incorporate the strength parameter $s>0$ and consider the Hamiltonian
\begin{equation*}
\mathcal{H}_{\omega}^{s}:=\mathcal{H}+sV_{\omega}.
\end{equation*}
Then we can prove that there exists $s_0>0$ such that when $s>s_0$, the whole spectrum of $\mathcal{H}_{\omega}^{s}$ falls inside the Aizenman-Molchanov localization regime (see, e.g., \cite{stolz2011introduction,aizenman2001finite_volime,prado2017dynamical}). Following the same deformation argument as in \cite{Schulz-Baldes01homotopy}, one then sees that the bulk index $\mathcal{E}_{bulk,\omega}$ is equal to the one when the system is non-filling (or equivalently, full-filling), where the latter is apparently topologically trivial; in other words, we have $\mathcal{E}_{bulk,\omega}=0$ for $s>s_0$. That is, long-range propagation in the bulk medium, characterized by the topological index, is suppressed under strong disorder.
\end{remark}

Now, we embark on proving Theorem \ref{thm_invariant_index}. To facilitate the deformation argument, we introduce the following smooth family of Hamiltonians on $\mathcal{X}$,
\begin{equation*}
\mathcal{H}_{\omega,t}:=\mathcal{H}_{0}+1\otimes V+V_{\omega,t}=\mathcal{H}_{per}+V_{\omega,t}
\end{equation*}
with $(V_{\omega,t}\psi)(\bm{n}):=-t\cdot \text{diag}\big((1-\omega_{\bm{n}}^{(1)})V_1,\cdots ,(1-\omega_{\bm{n}}^{(d)}V_d)\big)\cdot \psi(\bm{n})$. (Note that it has already been proposed in the proof of Proposition \ref{prop_lifshitz_tail}.) Apparently, $\mathcal{H}_{\omega,0}=\mathcal{H}_{per}$ and $\mathcal{H}_{\omega,1}=\mathcal{H}_{\omega}$. The key observation is that the interval $\mathcal{I}$ introduced in Theorem \ref{thm_anderson_localization} is always an Aizenman-Molchanov mobility gap of $\mathcal{H}_{\omega,t}$.
\begin{theorem} \label{thm_anderson_localization_deform}
Let the constant $s\in (0,1),\alpha,C>0$ and the interval $\mathcal{I}$ be the same as in Theorem \ref{thm_anderson_localization}. Then
\begin{equation*}
\mathbb{E}(\|G_{\omega,t}(\bm{n},\bm{m};\lambda+i\epsilon)\|^s)\leq Ce^{-\alpha|\bm{n}-\bm{m}|}
\end{equation*}
holds for all $t\in [0,1]$, $\lambda\in \mathcal{I}$, $\epsilon>0$, and $\bm{n},\bm{m}\in\mathbb{Z}^2$. Here, $$G_{\omega,t}(\bm{n},\bm{m};z)=\langle\mathbbm{1}_{\{\bm{n}\}},(\mathcal{H}_{\omega,t}-z)^{-1}\mathbbm{1}_{\{\bm{m}\}}\rangle_{\ell^2(\mathbb{Z}^2)}.$$
\end{theorem}
Intuitively, this is because for $t<1$, $\mathcal{H}_{\omega,t}$ has a sparser spectrum near $\lambda=\lambda_0$ compared to $\mathcal{H}_{\omega,1}=\mathcal{H}_{\omega}$\footnote{One may be convinced by noticing that $\sigma(\mathcal{H}_{\omega,t})\subset \sigma(\mathcal{H}_{per})+\sigma(V_{\omega,t})$ while the right side does not include $\lambda=\lambda_0$ for $t<1$.}, and hence exhibits stronger localization. The detailed proof follows the same lines as in Section \ref{sec_anderson_loc}, which is left for the interested reader.

As one may expect, Theorem \ref{thm_anderson_localization_deform} implies that the spectral projection $P_{\omega,\lambda,t}:=\mathbbm{1}_{(-\infty,\lambda)}(\mathcal{H}_{\omega,t})$ ($\lambda\in\mathcal{I}$) stays localized and varies continuously as a function of $t$ because we stay in the mobility gap along the deformation. To be more specific, we recall the following result from \cite[Lemma 1]{Schulz-Baldes01homotopy}.
\begin{lemma} \label{lem_deform_projection_localization}
For any $\gamma\in (0,\frac{1}{2})$, there exists $C>0$ such that
\begin{equation} \label{eq_deform_projection_localization_3}
\mathbb{E}\big(\|P_{\omega,\lambda,t_1}(\bm{n},\bm{m})\|\big)\leq Ce^{-\alpha |\bm{n}-\bm{m}|}
\end{equation}
and 
\begin{equation} \label{eq_deform_projection_localization_4}
\mathbb{E}\big(\|P_{\omega,\lambda,t_1}(\bm{n},\bm{m})-P_{\omega,\lambda,t_2}(\bm{n},\bm{m})\| \big)\leq C|t_1-t_2|^{\gamma}e^{-\alpha |\bm{n}-\bm{m}|},
\end{equation}
for all $t_1,t_2\in [0,1]$ and $\bm{n},\bm{m}\in\mathbb{Z}^2$. The constant $\alpha$ is the same as in Theorem \ref{thm_anderson_localization_deform}.
\end{lemma}
Using the bound $\|P_{\omega,\lambda,t}(\bm{n},\bm{m})\|\leq 1$, we further improve these estimates as follows.
\begin{corollary} \label{corol_deform_projection_localization}
For any $\gamma\in (0,\frac{1}{2})$, there exists $C>0$ such that
\begin{equation} \label{eq_deform_projection_localization_1}
\mathbb{E}\big(\|P_{\omega,\lambda,t_1}(\bm{n},\bm{m})\|^{k}\big)\leq Ce^{-\alpha |\bm{n}-\bm{m}|}
\end{equation}
and 
\begin{equation} \label{eq_deform_projection_localization_2}
\mathbb{E}\big(\|P_{\omega,\lambda,t_1}(\bm{n},\bm{m})-P_{\omega,\lambda,t_2}(\bm{n},\bm{m})\|^{k} \big)\leq C|t_1-t_2|^{\gamma}e^{-\alpha |\bm{n}-\bm{m}|}, 
\end{equation}
for all $t_1,t_2\in [0,1]$, $k\in\mathbb{N}$ and $\bm{n},\bm{m}\in\mathbb{Z}^2$. The constant $\alpha$ is the same as in Theorem \ref{thm_anderson_localization_deform}.
\end{corollary}

Now, we are ready to prove Theorem \ref{thm_invariant_index}.

\begin{proof}[Proof of Theorem \ref{thm_invariant_index}]
It is sufficient to prove that, for any fixed $t_1,t_2\in [0,1]$ that are sufficiently close, 
\begin{equation} \label{eq_invariant_index_proof_1}
\mathbb{E}(\mathcal{E}_{bulk,\omega,\lambda,t_1})=\mathbb{E}(\mathcal{E}_{bulk,\omega,\lambda,t_2})
\end{equation}
with $\mathcal{E}_{bulk,\omega,\lambda,t}:=-i\cdot\mathcal{T}\big(P_{\omega,\lambda,t}\big[[P_{\omega,\lambda,t},x_1],[P_{\omega,\lambda,t},x_2]\big]\big)$.

To estimate the difference $\mathbb{E}\big(\mathcal{E}_{bulk,\omega,\lambda,t_1}-\mathcal{E}_{bulk,\omega,\lambda,t_2}\big)$, we note that the bulk index can be recast into a microscopic version thanks to the localization of spectral projections as in Lemma \ref{lem_deform_projection_localization}\footnote{Its original definition in the form of TPUV is referred to as the macroscopic version since it averages throughout the whole system.}:
\begin{equation*}
\mathbb{E}\big(\mathcal{E}_{bulk,\omega,\lambda,t_{j}}\big)=-i\cdot\mathbb{E}\Big(\text{Tr}_{\mathcal{X}}\big(P_{\omega,\lambda,t_j}\big[[P_{\omega,\lambda,t_j},\mathbbm{1}_{\{n_1\geq 0\}}],[P_{\omega,\lambda,t_j},\mathbbm{1}_{\{n_2\geq 0\}}]\big]\big)\Big) .
\end{equation*}
The proof can be found in \cite[Lemma 8]{graf2005equality}; see also \cite{Avron1994charge}. Hence, we can decompose the difference of indices as
\begin{equation} \label{eq_invariant_index_proof_2}
\begin{aligned}
\mathbb{E}\big(\mathcal{E}_{bulk,\omega,\lambda,t_1}-\mathcal{E}_{bulk,\omega,\lambda,t_2}\big)
&=-i\cdot\mathbb{E}\Big( \text{Tr}_{\mathcal{X}}\big(\tilde{P}_{\omega,\lambda,t_{1,2}}\big[[P_{\omega,\lambda,t_1},\mathbbm{1}_{\{n_1\geq 0\}}],[P_{\omega,\lambda,t_1},\mathbbm{1}_{\{n_2\geq 0\}}]\big]\big) \Big) \\
&\quad -i\cdot\mathbb{E}\Big( \text{Tr}_{\mathcal{X}}\big(P_{\omega,\lambda,t_2}\big[[\tilde{P}_{\omega,\lambda,t_{1,2}},\mathbbm{1}_{\{n_1\geq 0\}}],[P_{\omega,\lambda,t_1},\mathbbm{1}_{\{n_2\geq 0\}}]\big]\big) \Big) \\
&\quad -i\cdot\mathbb{E}\Big( \text{Tr}_{\mathcal{X}}\big(P_{\omega,\lambda,t_2}\big[[P_{\omega,\lambda,t_2},\mathbbm{1}_{\{n_1\geq 0\}}],[\tilde{P}_{\omega,\lambda,t_{1,2}},\mathbbm{1}_{\{n_2\geq 0\}}]\big]\big) \Big)
\end{aligned}
\end{equation}
with $\tilde{P}_{\omega,\lambda,t_{1,2}}:=P_{\omega,\lambda,t_1}-P_{\omega,\lambda,t_2}$. By Lemma \ref{lem_deform_projection_localization}, each of the three operators in \eqref{eq_invariant_index_proof_2} is trace-class almost everywhere.\footnote{Heuristically, these operators are trace-class because the commutator with the indicator function $\mathbbm{1}_{\{n_j\geq 0\}}$ in the $j$\textsuperscript{th} direction are localizes in $n_j$, and thus the whole composition is localized in $\mathbb{Z}^2$. The detailed justification is done by using the list of trace-class properties in Section \ref{sec_trace_class}; see also the calculation under \eqref{eq_invariant_index_proof_4}.} Moreover, we claim that the right side of \eqref{eq_invariant_index_proof_2} is bounded by $|t_1-t_2|^{\gamma}$. Here,  we only show the proof for the first term, while the other terms can be treated similarly. Calculating the trace in position basis, we see that for almost every $\omega$
\begin{equation} \label{eq_invariant_index_proof_3}
\begin{aligned}
&\text{Tr}_{\mathcal{X}}\big(\tilde{P}_{\omega,\lambda,t_{1,2}}\big[[P_{\omega,\lambda,t_1},\mathbbm{1}_{\{n_1\geq 0\}}],[P_{\omega,\lambda,t_1},\mathbbm{1}_{\{n_2\geq 0\}}]\big]\big) \\
&=\text{tr}\sum_{\bm{n}\in \mathbb{Z}^2}\sum_{\bm{m}\in \mathbb{Z}^2}\tilde{P}_{\omega,\lambda,t_{1,2}}(\bm{n},\bm{m})\sum_{\bm{k}\in \mathbb{Z}^2}[P_{\omega,\lambda,t_1},\mathbbm{1}_{\{n_1\geq 0\}}](\bm{m},\bm{k})[P_{\omega,\lambda,t_1},\mathbbm{1}_{\{n_2\geq 0\}}](\bm{k},\bm{n}) \\
&\quad -\text{tr}\sum_{\bm{n}\in \mathbb{Z}^2}\sum_{\bm{m}\in \mathbb{Z}^2}\tilde{P}_{\omega,\lambda,t_{1,2}}(\bm{n},\bm{m})\sum_{\bm{k}\in \mathbb{Z}^2}[P_{\omega,\lambda,t_1},\mathbbm{1}_{\{n_2\geq 0\}}](\bm{m},\bm{k})[P_{\omega,\lambda,t_1},\mathbbm{1}_{\{n_1\geq 0\}}](\bm{k},\bm{n}),
\end{aligned}
\end{equation}
where the decomposition into two traces is justified by the fact that the two operators at the right side are individually trace-class for almost every $\omega$. By the Hölder inequality, the expectation of the first trace at the right side of \eqref{eq_invariant_index_proof_3} is estimated as
\begin{equation} \label{eq_invariant_index_proof_4}
\begin{aligned}
&\mathbb{E}\Big(\big|\text{tr}\sum_{\bm{n}\in \mathbb{Z}^2}\sum_{\bm{m}\in \mathbb{Z}^2}\tilde{P}_{\omega,\lambda,t_{1,2}}(\bm{n},\bm{m})\sum_{\bm{k}\in \mathbb{Z}^2}[P_{\omega,\lambda,t_1},\mathbbm{1}_{\{n_1\geq 0\}}](\bm{m},\bm{k})[P_{\omega,\lambda,t_1},\mathbbm{1}_{\{n_2\geq 0\}}](\bm{k},\bm{n})\big| \Big) \\
&\leq d \sum_{\bm{n},\bm{m},\bm{k}\in\mathbb{Z}^2}
\mathbb{E}\big( \|\tilde{P}_{\omega,\lambda,t_{1,2}}(\bm{n},\bm{m})\|^2 \big)^{1/2}
\mathbb{E}\big( \|[P_{\omega,\lambda,t_1},\mathbbm{1}_{\{n_1\geq 0\}}] (\bm{m},\bm{k})\|^4 \big)^{1/4} 
\mathbb{E}\big( \|[P_{\omega,\lambda,t_1},\mathbbm{1}_{\{n_2\geq 0\}}](\bm{k},\bm{n}) \|^4 \big)^{1/4} \\
&\leq C \sum_{\bm{n},\bm{m},\bm{k}\in\mathbb{Z}^2}
|t_1-t_2|^{\frac{\gamma}{2}}e^{-\frac{\alpha}{2} |\bm{n}-\bm{m}|}
\mathbb{E}\big( \|[P_{\omega,\lambda,t_1},\mathbbm{1}_{\{n_1\geq 0\}}] (\bm{m},\bm{k})\|^4 \big)^{1/4} 
\mathbb{E}\big( \|[P_{\omega,\lambda,t_1},\mathbbm{1}_{\{n_2\geq 0\}}](\bm{k},\bm{n}) \|^4 \big)^{1/4},
\end{aligned}
\end{equation}
where \eqref{eq_deform_projection_localization_2} is applied for the last inequality. The expectation of commutators is bounded as follows. Note that
\begin{equation*}
\begin{aligned}
\|[P_{\omega,\lambda,t_1},\mathbbm{1}_{\{n_1\geq 0\}}] (\bm{m},\bm{k})\|
&= \Big\|\big(\mathbbm{1}_{\{n_1\geq 0\}^c}P_{\omega,\lambda,t_1}\mathbbm{1}_{\{n_1\geq 0\}}-\mathbbm{1}_{\{n_1\geq 0\}}P_{\omega,\lambda,t_1}\mathbbm{1}_{\{n_1\geq 0\}^c}\big)(\bm{m},\bm{k}) \Big\| \\
&\leq \Big\|\big(\mathbbm{1}_{\{n_1\geq 0\}^c}P_{\omega,\lambda,t_1}\mathbbm{1}_{\{n_1\geq 0\}}\big)(\bm{m},\bm{k}) \Big\|
+\Big\|\big(\mathbbm{1}_{\{n_1\geq 0\}}P_{\omega,\lambda,t_1}\mathbbm{1}_{\{n_1\geq 0\}^c}\big)(\bm{m},\bm{k}) \Big\| \\
&=\|P_{\omega,\lambda,t_1}(\bm{m},\bm{k})\| \mathbbm{1}_{\{n_1\geq 0\}^c}(\bm{m})\mathbbm{1}_{\{n_1\geq 0\}}(\bm{k}) \\
&\quad +\|P_{\omega,\lambda,t_1}(\bm{m},\bm{k})\|) \mathbbm{1}_{\{n_1\geq 0\}^c}(\bm{k})\mathbbm{1}_{\{n_1\geq 0\}}(\bm{m}).
\end{aligned}
\end{equation*}
Moreover,
\begin{equation*}
\begin{aligned}
\mathbb{E}\Big( \big( \|P_{\omega,\lambda,t_1}(\bm{m},\bm{k})\| \mathbbm{1}_{\{n_1\geq 0\}^c}(\bm{m})\mathbbm{1}_{\{n_1\geq 0\}}(\bm{k})\big)^4 \Big)
&=\mathbb{E}\big(\|P_{\omega,\lambda,t_1}(\bm{m},\bm{k})\|^4 \big) \mathbbm{1}_{\{n_1\geq 0\}^c}(\bm{m})\mathbbm{1}_{\{n_1\geq 0\}}(\bm{k}) \\
&\leq Ce^{-\alpha |\bm{m}-\bm{k}|} e^{-\alpha\cdot \text{dist}(\bm{m},\mathbb{Z}^{-}\times \mathbb{Z})} e^{-\alpha\cdot \text{dist}(\bm{k},\mathbb{Z}^{+}\times \mathbb{Z})} \\
&\leq Ce^{-\frac{\alpha}{3} |\bm{m}-\bm{k}|}e^{-\frac{\alpha}{3} |m_1|}e^{-\frac{\alpha}{3} |k_1|},
\end{aligned}
\end{equation*}
where \eqref{eq_deform_projection_localization_1} is applied to derive the first inequality, and the following is utilized for the second
\begin{equation*}
|\bm{m}-\bm{k}|+\text{dist}(\bm{m},\mathbb{Z}^{-}\times \mathbb{Z})+\text{dist}(\bm{k},\mathbb{Z}^{+}\times \mathbb{Z})
\geq \text{dist}(\bm{k},\mathbb{Z}^{-}\times \mathbb{Z}) + \text{dist}(\bm{k},\mathbb{Z}^{+}\times \mathbb{Z}) =|k_1| ,
\end{equation*}
\begin{equation*}
|\bm{m}-\bm{k}|+\text{dist}(\bm{m},\mathbb{Z}^{-}\times \mathbb{Z})+\text{dist}(\bm{k},\mathbb{Z}^{+}\times \mathbb{Z})
\geq \text{dist}(\bm{m},\mathbb{Z}^{+}\times \mathbb{Z}) + \text{dist}(\bm{m},\mathbb{Z}^{-}\times \mathbb{Z}) =|m_1| .
\end{equation*}
Hence, we see from \eqref{eq_invariant_index_proof_4} that
\begin{equation*}
\begin{aligned}
&\mathbb{E}\Big(\big|\text{tr}\sum_{\bm{n}\in \mathbb{Z}^2}\sum_{\bm{m}\in \mathbb{Z}^2}\tilde{P}_{\omega,\lambda,t_{1,2}}(\bm{n},\bm{m})\sum_{\bm{k}\in \mathbb{Z}^2}[P_{\omega,\lambda,t_1},\mathbbm{1}_{\{n_1\geq 0\}}](\bm{m},\bm{k})[P_{\omega,\lambda,t_1},\mathbbm{1}_{\{n_2\geq 0\}}](\bm{k},\bm{n})\big| \Big) \\
&\leq C|t_1-t_2|^{\frac{\gamma}{2}}\sum_{\bm{n},\bm{m},\bm{k}\in\mathbb{Z}^2}
e^{-\frac{\alpha}{2} |\bm{n}-\bm{m}|}e^{-\frac{\alpha}{12} |\bm{m}-\bm{k}|}e^{-\frac{\alpha}{12} |\bm{k}-\bm{n}|}e^{-\frac{\alpha}{12} |m_1|}e^{-\frac{\alpha}{12} |k_1|}e^{-\frac{\alpha}{12} |k_2|}e^{-\frac{\alpha}{12} |n_2|}
\end{aligned}
\end{equation*}
with $C>0$ being independent of $t_1,t_2$. The infinite sum at the right side, which is calculated elementarily using \eqref{eq_finite_volume_cyclicity_proof_2}, is also finite and independent of $t_1,t_2$. Thus, by estimating the other two terms in \eqref{eq_invariant_index_proof_2} similarly, we conclude that
\begin{equation*}
\mathbb{E}(\mathcal{E}_{bulk,\omega,\lambda,t_1}-\mathcal{E}_{bulk,\omega,\lambda,t_2})=\mathcal{O}(|t_1-t_2|^{\gamma}).
\end{equation*}
This, together with the quantization of bulk index $\mathbb{E}(\mathcal{E}_{bulk,\omega,\lambda,t})\in\mathbb{Z}$, completes the proof of \eqref{eq_invariant_index_proof_1}.
\end{proof}

\bigskip
\textbf{Conflict of Interest Statement}: On behalf of all authors, the corresponding author states that there is no conflict of interest.

\bigskip
\textbf{Data Availability Statement}: On behalf of all authors, the corresponding author confirms that the current paper does not have associated data.

\setcounter{equation}{0}
\setcounter{subsection}{0}
\setcounter{theorem}{0}
\renewcommand{\theequation}{B.\arabic{equation}}
\renewcommand{\thesubsection}{B.\arabic{subsection}}
\renewcommand{\thetheorem}{B.\arabic{theorem}}

\footnotesize
\bibliographystyle{plain}
\bibliography{ref}

@article {SWP1,
    AUTHOR = {Ammari, Habib and Davies, Bryn and Hiltunen, Erik Orvehed},
     TITLE = {Anderson localization in the subwavelength regime},
   JOURNAL = {Comm. Math. Phys.},
  FJOURNAL = {Communications in Mathematical Physics},
    VOLUME = {405},
      YEAR = {2024},
    NUMBER = {1},
     PAGES = {Paper No. 1, 20},
}

@article {SWP2,
    AUTHOR = {Ammari, Habib and Kosche, Thea},
     TITLE = {Topological phenomena in honeycomb {F}loquet metamaterials},
   JOURNAL = {Math. Ann.},
  FJOURNAL = {Mathematische Annalen},
    VOLUME = {388},
      YEAR = {2024},
    NUMBER = {3},
     PAGES = {2755--2785},
}

@article {SWP3,
    AUTHOR = {Ammari, Habib and Davies, Bryn and Hiltunen, Erik Orvehed},
     TITLE = {Robust edge modes in dislocated systems of subwavelength
              resonators},
   JOURNAL = {J. Lond. Math. Soc. (2)},
  FJOURNAL = {Journal of the London Mathematical Society. Second Series},
    VOLUME = {106},
      YEAR = {2022},
    NUMBER = {3},
     PAGES = {2075--2135},
}

@article {SWP4,
    AUTHOR = {Ammari, Habib and Davies, Bryn and Hiltunen, Erik Orvehed and
              Yu, Sanghyeon},
     TITLE = {Topologically protected edge modes in one-dimensional chains
              of subwavelength resonators},
   JOURNAL = {J. Math. Pures Appl. (9)},
  FJOURNAL = {Journal de Math\'ematiques Pures et Appliqu\'ees. Neuvi\`eme
              S\'erie},
    VOLUME = {144},
      YEAR = {2020},
     PAGES = {17--49},
}

@article{qi2006qwz_model,
  title={Topological quantization of the spin {H}all effect in two-dimensional paramagnetic semiconductors},
  author={Qi, Xiao-Liang and Wu, Yong-Shi and Zhang, Shou-Cheng},
  journal={Physical Review B—Condensed Matter and Materials Physics},
  volume={74},
  number={8},
  pages={085308},
  year={2006},
  publisher={APS}
}

@article{ezawa2024nonlinear_phase_transition,
  title={Nonlinearity-induced topological phase transition characterized by the nonlinear {C}hern number},
  author={Sone, Kazuki and Ezawa, Motohiko and Ashida, Yuto and Yoshioka, Nobuyuki and Sagawa, Takahiro},
  journal={Nature Physics},
  volume={20},
  number={7},
  pages={1164--1170},
  year={2024},
  publisher={Nature Publishing Group UK London}
}

@article{cornean2019gapped_dirac,
  title={Localization for gapped {D}irac Hamiltonians with random perturbations: Application to graphene antidot lattices},
  author={Barbaroux, Jean-Marie and Cornean, Horia D and Zalczer, Sylvain},
  journal={Documenta Mathematica},
  volume={24},
  pages={65--93},
  year={2019}
}

@book{stollmann2001caught,
  title={Caught by disorder: bound states in random media},
  author={Stollmann, Peter},
  volume={20},
  year={2001},
  publisher={Springer Science \& Business Media}
}

@article{aizenman1998localization,
  title={Localization bounds for an electron gas},
  author={Aizenman, Michael and Graf, Gian M.},
  journal={Journal of Physics A: Mathematical and General},
  volume={31},
  number={32},
  pages={6783},
  year={1998},
  publisher={IOP Publishing}
}

@article{najar2006band_edge,
  title={2-dimensional localization of acoustic waves in random perturbation of periodic media},
  author={Najar, Hatem},
  journal={Journal of mathematical analysis and applications},
  volume={322},
  number={1},
  pages={1--17},
  year={2006},
  publisher={Elsevier}
}

@article{combes1994localization,
  title={Localization for some continuous, random Hamiltonians in d-dimensions},
  author={Combes, Jean M. and Hislop, Peter D.},
  journal={Journal of Functional Analysis},
  volume={124},
  number={1},
  pages={149--180},
  year={1994},
  publisher={Elsevier}
}

@article{veselic2002localization,
  title={Localization for random perturbations of periodic {S}chr{\"o}dinger operators with regular {F}loquet eigenvalues},
  author={Veseli{\'c}, I.},
  journal={Annales Henri Poincar{\'e}},
  volume={3},
  issue={2},
  pages={389--409},
  year={2002},
  organization={Springer}
}

@article{klopp2003note_anderson,
  title={A note on {A}nderson localization for the random hopping model},
  author={Klopp, Fr{\'e}d{\'e}ric and Nakamura, Shu},
  journal={Journal of Mathematical Physics},
  volume={44},
  issue={11},
  pages={4975--4980},
  year={2003},
  publisher={American Institute of Physics}
}

@article{nakamura2003anderson,
  title={Anderson localization for 2D discrete {S}chr{\"o}dinger operators with random magnetic fields},
  author={Klopp, Fr{\'e}d{\'e}ric and Nakamura, Shu and Nakano, Fumihiko and Nomura, Yuji},
  journal={Annales Henri Poincar{\'e}},
  volume={4},
  number={4},
  pages={795--811},
  year={2003},
  organization={Springer}
}

@article{prado2017dynamical,
  title={Dynamical localization for discrete {A}nderson {D}irac operators},
  author={Prado, Roberto A. and de Oliveira, C{\'e}sar R. and Carvalho, Silas L.},
  journal={Journal of Statistical Physics},
  volume={167},
  number={2},
  pages={260--296},
  year={2017},
  publisher={Springer}
}

@article{prado2021density,
  title={Density of states and Lifshitz tails for discrete 1D random {D}irac operators},
  author={Prado, Roberto A. and de Oliveira, C{\'e}sar R. and de Oliveira, Edmundo C.},
  journal={Mathematical Physics, Analysis and Geometry},
  volume={24},
  number={3},
  pages={30},
  year={2021},
  publisher={Springer}
}

@article{carmona1987anderson,
  title={Anderson localization for {B}ernoulli and other singular potentials},
  author={Carmona, Ren{\'e} and Klein, Abel and Martinelli, Fabio},
  journal={Communications in Mathematical Physics},
  volume={108},
  number={1},
  pages={41--66},
  year={1987},
  publisher={Springer}
}

@article{aizenman1993localization_elementary_derivation,
  title={Localization at large disorder and at extreme energies: An elementary derivations},
  author={Aizenman, Michael and Molchanov, Stanislav},
  journal={Communications in Mathematical Physics},
  volume={157},
  number={2},
  pages={245--278},
  year={1993},
  publisher={Springer}
}

@article{chulaevsky2023anderson,
  title={Anderson Localization in Discrete Random Displacements Models},
  author={Chulaevsky, Victor},
  journal={Journal of Statistical Physics},
  volume={190},
  number={1},
  pages={5},
  year={2023},
  publisher={Springer}
}

@article{li2022anderson,
  title={Anderson--{B}ernoulli localization on the three-dimensional lattice and discrete unique continuation principle},
  author={Li, Linjun and Zhang, Lingfu},
  journal={Duke mathematical journal},
  volume={171},
  number={2},
  pages={327--415},
  year={2022},
  publisher={Duke University Press}
}

@article{ding2020localization,
  title={Localization near the edge for the {A}nderson {B}ernoulli model on the two dimensional lattice},
  author={Ding, Jian and Smart, Charles K},
  journal={Inventiones mathematicae},
  volume={219},
  number={2},
  pages={467--506},
  year={2020},
  publisher={Springer}
}

@article{aizenman1994localization,
  title={Localization at weak disorder: some elementary bounds},
  author={Aizenman, Michael},
  journal={Reviews in mathematical physics},
  volume={6},
  number={05a},
  pages={1163--1182},
  year={1994},
  publisher={World Scientific}
}

@article{kirsch2007invitation,
  title={An invitation to random {S}chr{\"o}dinger operators},
  author={Kirsch, Werner},
  journal={arXiv preprint arXiv:0709.3707},
  year={2007}
}

@article{stolz2011introduction,
  title={An introduction to the mathematics of {A}nderson localization},
  author={Stolz, G{\"u}nter},
  journal={Entropy and the quantum II. Contemp. Math},
  volume={552},
  pages={71--108},
  year={2011}
}

@article{graf2005equality,
  title={Equality of the bulk and edge {H}all conductances in a mobility gap},
  author={Elgart, Alexander and Graf, Gian M. and Schenker, Jeffrey H.},
  journal={Communications in mathematical physics},
  volume={259},
  number={1},
  pages={185--221},
  year={2005},
  publisher={Springer}
}

@article{chan2024topological_alloy,
  title={Topological photonic alloy},
  author={Qu, Tiantao and Wang, Mudi and Cheng, Xiaoyu and Cui, Xiaohan and Zhang, Ruo-Yang and Zhang, Zhao-Qing and Zhang, Lei and Chen, Jun and Chan, CT},
  journal={Physical Review Letters},
  volume={132},
  number={22},
  pages={223802},
  year={2024},
  publisher={APS}
}

@article{qiu2025bec_finite,
  title={Bulk-edge correspondence in finite photonic structure},
  author={Qiu, Jiayu and Zhang, Hai},
  journal={arXiv preprint arXiv:2501.15531},
  year={2025}
}

@article{qiu2025generalized,
  title={Generalized bulk-interface correspondence for non-quantized spin transport},
  author={Qiu, Jiayu and Zhang, Hai},
  journal={arXiv preprint arXiv:2505.16331},
  year={2025}
}

@article{aizenman2001finite_volime,
  title={Finite-Volume Fractional-Moment Criteria for {A}nderson Localization},
  author={Aizenman, Michael and Schenker, Jeffrey H. and Friedrich, Roland M. and Hundertmark, Dirk},
  journal={Communications in Mathematical Physics},
  volume={224},
  number={1},
  pages={219--253},
  year={2001},
  publisher={Springer}
}

@article{Schulz-Baldes01homotopy,
    author = {Richter, T. and Schulz-Baldes, H.},
    title = {Homotopy arguments for quantized {H}all conductivity},
    journal = {Journal of Mathematical Physics},
    volume = {42},
    number = {8},
    pages = {3439-3444},
    year = {2001},
    month = {08},
    issn = {0022-2488},
    doi = {10.1063/1.1379070},
    url = {https://doi.org/10.1063/1.1379070},
    eprint = {https://pubs.aip.org/aip/jmp/article-pdf/42/8/3439/19261580/3439_1_online.pdf},
}

@article{Bellissard94non_commutative_QHE,
    author = {Bellissard, J. and van Elst, A. and Schulz‐ Baldes, H.},
    title = {The noncommutative geometry of the quantum {H}all effect},
    journal = {Journal of Mathematical Physics},
    volume = {35},
    number = {10},
    pages = {5373-5451},
    year = {1994},
    month = {10},
    issn = {0022-2488},
    doi = {10.1063/1.530758},
    url = {https://doi.org/10.1063/1.530758},
    eprint = {https://pubs.aip.org/aip/jmp/article-pdf/35/10/5373/19099403/5373_1_online.pdf},
}

@article{barbaroux1997localization,
  title={Localization near band edges for random {S}chr{\"o}dinger operators},
  author={Barbaroux, Jean-Marie and Combes, Jean-Michel and Hislop, Peter D.},
  journal={Helvetica Physica Acta},
  volume={70},
  number={1},
  pages={16--43},
  year={1997},
  publisher={Basel: E. Birkhauser, 1928-c1999.}
}

@article{stolz1998anderson,
  title={Anderson localization for random {S}chr{\"o}dinger operators with long range interactions},
  author={Kirsch, Werner and Stollmann, Peter and Stolz, G{\"u}nter},
  journal={Communications in mathematical physics},
  volume={195},
  number={3},
  pages={495--507},
  year={1998},
  publisher={Springer}
}

@article{seelmann2020band,
  title={Band edge localization beyond regular {F}loquet eigenvalues},
  author={Seelmann, Albrecht and T{\"a}ufer, Matthias},
  journal={Annales Henri Poincar{\'e}},
  volume={21},
  issue={7},
  pages={2151--2166},
  year={2020},
  organization={Springer}
}

@book{kato2013perturbation,
  title={Perturbation theory for linear operators},
  author={Kato, Tosio},
  volume={132},
  year={2013},
  publisher={Springer Science \& Business Media}
}

@article{Avron1994charge,
author={Avron, Joseph E.
and Seiler, Ruedi
and Simon, Barry},
title={Charge deficiency, charge transport and comparison of dimensions},
journal={Communications in Mathematical Physics},
year={1994},
month={Jan},
day={01},
volume={159},
number={2},
pages={399-422},
issn={1432-0916},
doi={10.1007/BF02102644},
url={https://doi.org/10.1007/BF02102644}
}

@book{zworski2012semiclassical,
  title={Semiclassical Analysis},
  author={Zworski, M.},
  isbn={9780821883204},
  lccn={2012010649},
  series={Graduate studies in mathematics},
  url={https://books.google.com.hk/books?id=3Z0CAQAAQBAJ},
  year={2012},
  publisher={American Mathematical Society}
}

@article{Klein2005linear_disorder,
  title={Linear response theory for magnetic {S}chr{\"o}dinger operators in disordered media},
  author={Bouclet, Jean-Marc and Germinet, Francois and Klein, Abel and Schenker, Jeffrey H.},
  journal={Journal of Functional Analysis},
  volume={226},
  number={2},
  pages={301--372},
  year={2005},
  publisher={Elsevier}
}

@Article{vonKlitzing2020forty_years,
author={von Klitzing, Klaus
and Chakraborty, Tapash
and Kim, Philip
and Madhavan, Vidya
and Dai, Xi
and McIver, James
and Tokura, Yoshinori
and Savary, Lucile
and Smirnova, Daria
and Rey, Ana Maria
and Felser, Claudia
and Gooth, Johannes
and Qi, Xiaoliang},
title={40 years of the quantum {H}all effect},
journal={Nature Reviews Physics},
year={2020},
month={Aug},
day={01},
volume={2},
number={8},
pages={397-401},
issn={2522-5820},
doi={10.1038/s42254-020-0209-1},
url={https://doi.org/10.1038/s42254-020-0209-1}
}

@article{Klitzing80qhe,
  title = {New Method for High-Accuracy Determination of the Fine-Structure Constant Based on Quantized {H}all Resistance},
  author = {Klitzing, K. v. and Dorda, G. and Pepper, M.},
  journal = {Phys. Rev. Lett.},
  volume = {45},
  issue = {6},
  pages = {494--497},
  numpages = {0},
  year = {1980},
  month = {Aug},
  publisher = {American Physical Society},
  doi = {10.1103/PhysRevLett.45.494},
  url = {https://link.aps.org/doi/10.1103/PhysRevLett.45.494}
}

@article{qizhang11topo_insulator,
  title = {Topological insulators and superconductors},
  author = {Qi, Xiao-Liang and Zhang, Shou-Cheng},
  journal = {Rev. Mod. Phys.},
  volume = {83},
  issue = {4},
  pages = {1057--1110},
  numpages = {0},
  year = {2011},
  month = {Oct},
  publisher = {American Physical Society},
  doi = {10.1103/RevModPhys.83.1057},
  url = {https://link.aps.org/doi/10.1103/RevModPhys.83.1057}
}

@book{bernivig13topo_insulator,
 ISBN = {9780691151755},
 URL = {http://www.jstor.org/stable/j.ctt19cc2gc},
 author = {B. Andrei Bernevig and Taylor L. Hughes},
 publisher = {Princeton University Press},
 title = {Topological Insulators and Topological Superconductors},
 urldate = {2024-12-12},
 year = {2013}
}

@article{haldane08realization,
  title = {Possible Realization of Directional Optical Waveguides in Photonic Crystals with Broken Time-Reversal Symmetry},
  author = {Haldane, F. D. M. and Raghu, S.},
  journal = {Phys. Rev. Lett.},
  volume = {100},
  issue = {1},
  pages = {013904},
  numpages = {4},
  year = {2008},
  month = {Jan},
  publisher = {American Physical Society},
  doi = {10.1103/PhysRevLett.100.013904},
  url = {https://link.aps.org/doi/10.1103/PhysRevLett.100.013904}
}

@article{wang2009observation,
  title={Observation of unidirectional backscattering-immune topological electromagnetic states},
  author={Wang, Zheng and Chong, Yidong and Joannopoulos, John D. and Solja{\v{c}}i{\'c}, Marin},
  journal={Nature},
  volume={461},
  number={7265},
  pages={772--775},
  year={2009},
  publisher={Nature Publishing Group UK London}
}

@article{ozawa19topological_photonics,
  title = {Topological photonics},
  author = {Ozawa, Tomoki and Price, Hannah M. and Amo, Alberto and Goldman, Nathan and Hafezi, Mohammad and Lu, Ling and Rechtsman, Mikael C. and Schuster, David and Simon, Jonathan and Zilberberg, Oded and Carusotto, Iacopo},
  journal = {Rev. Mod. Phys.},
  volume = {91},
  issue = {1},
  pages = {015006},
  numpages = {76},
  year = {2019},
  month = {Mar},
  publisher = {American Physical Society},
  doi = {10.1103/RevModPhys.91.015006},
  url = {https://link.aps.org/doi/10.1103/RevModPhys.91.015006}
}

@article{Hatsugai93chern,
  title = {{C}hern number and edge states in the integer quantum {H}all effect},
  author = {Hatsugai, Yasuhiro},
  journal = {Phys. Rev. Lett.},
  volume = {71},
  issue = {22},
  pages = {3697--3700},
  numpages = {0},
  year = {1993},
  month = {Nov},
  publisher = {American Physical Society},
  doi = {10.1103/PhysRevLett.71.3697},
  url = {https://link.aps.org/doi/10.1103/PhysRevLett.71.3697}
}

@article{Kellendonk02landau+ktheory,
author = {Kellendonk, J. and Richter, T. and Schulz-Baldes, H.},
title = {Edge current channels and {C}hern numbers in the integer quantum {H}all effect},
journal = {Reviews in Mathematical Physics},
volume = {14},
number = {01},
pages = {87-119},
year = {2002},
doi = {10.1142/S0129055X02001107},
URL = {https://doi.org/10.1142/S0129055X02001107},
eprint = { https://doi.org/10.1142/S0129055X02001107}
}

@article{Kellendonk04landau+functional,
author = {Kellendonk, Johannes and Schulz-Baldes, Hermann},
year = {2004},
month = {06},
pages = {388-413},
title = {Quantization of edge currents for continuous magnetic operators},
volume = {209},
journal = {Journal of Functional Analysis},
doi = {10.1016/S0022-1236(03)00174-5}
}

@article{Kellendonk2004landau+ktheory,
  title={Boundary Maps for {C}*-Crossed Products with an Application to the Quantum {H}all Effect},
  author={Johannes Kellendonk and Hermann Schulz-Baldes},
  journal={Communications in Mathematical Physics},
  year={2004},
  volume={249},
  pages={611-637},
  url={https://api.semanticscholar.org/CorpusID:18638874}
}

@misc{taarabt2014landau+Ktheory,
      title={Equality of bulk and edge {H}all conductances for continuous magnetic random {S}chr\"odinger operators}, 
      author={Amal Taarabt},
      year={2014},
      eprint={1403.7767},
      archivePrefix={arXiv},
      primaryClass={math-ph},
      url={https://arxiv.org/abs/1403.7767}, 
}

@inproceedings{cornean2021landau+functional,
  title={From orbital magnetism to bulk-edge correspondence},
  author={Cornean, Horia D. and Moscolari, Massimo and Teufel, Stefan},
  booktitle={Annales Henri Poincar{\'e}},
  pages={1--55},
  year={2024},
  organization={Springer}
}

@article{shapiro2022shrodinger+functional,
  title={Tight-binding reduction and topological equivalence in strong magnetic fields},
  author={Shapiro, Jacob and Weinstein, Michael I.},
  journal={Advances in Mathematics},
  volume={403},
  pages={108343},
  year={2022},
  publisher={Elsevier}
}

@article{drouot2021microlocal,
  title={Microlocal analysis of the bulk-edge correspondence},
  author={Drouot, Alexis},
  journal={Communications in Mathematical Physics},
  volume={383},
  pages={2069--2112},
  year={2021},
  publisher={Springer}
}

@article{bal2019dirac+functional,
  title={Continuous bulk and interface description of topological insulators},
  author={Bal, Guillaume},
  journal={Journal of Mathematical Physics},
  volume={60},
  number={8},
  year={2019},
  publisher={AIP Publishing}
}

@article{bal2023dirac+microlocal,
  title={Topological charge conservation for continuous insulators},
  author={Bal, Guillaume},
  journal={Journal of Mathematical Physics},
  volume={64},
  number={3},
  year={2023},
  publisher={AIP Publishing}
}

@article{bal2023edge_curved,
  title={Edge state dynamics along curved interfaces},
  author={Bal, Guillaume and Becker, Simon and Drouot, Alexis and Kammerer, Clotilde Fermanian and Lu, Jianfeng and Watson, Alexander B.},
  journal={SIAM Journal on Mathematical Analysis},
  volume={55},
  number={5},
  pages={4219--4254},
  year={2023},
  publisher={SIAM}
}

@article{graf2018shortrange+transfer,
  title={The bulk-edge correspondence for disordered chiral chains},
  author={Graf, Gian Michele and Shapiro, Jacob},
  journal={Communications in Mathematical Physics},
  volume={363},
  pages={829--846},
  year={2018},
  publisher={Springer}
}

@article{avila2013shortrange+transfer,
  title={Topological invariants of edge states for periodic two-dimensional models},
  author={Avila, Julio Cesar and Schulz-Baldes, Hermann and Villegas-Blas, Carlos},
  journal={Mathematical Physics, Analysis and Geometry},
  volume={16},
  number={2},
  pages={137--170},
  year={2013},
  publisher={Springer}
}

@article{ludewig2020shortrange+coarse,
  title={Cobordism invariance of topological edge-following states},
  author={Ludewig, Matthias and Thiang, Guo Chuan},
  journal={arXiv preprint arXiv:2001.08339},
  year={2020}
}

@article{graf2013shortrange+scattering,
  title={Bulk-edge correspondence for two-dimensional topological insulators},
  author={Graf, Gian Michele and Porta, Marcello},
  journal={Communications in Mathematical Physics},
  volume={324},
  pages={851--895},
  year={2013},
  publisher={Springer}
}

@inproceedings{bourne2017ktheory,
  title={The K-theoretic bulk--edge correspondence for topological insulators},
  author={Bourne, Chris and Kellendonk, Johannes and Rennie, Adam},
  booktitle={Annales Henri Poincar{\'e}},
  volume={18},
  pages={1833--1866},
  year={2017},
  organization={Springer}
}

@article{kubota2017ktheory,
  title={Controlled topological phases and bulk-edge correspondence},
  author={Kubota, Yosuke},
  journal={Communications in Mathematical Physics},
  volume={349},
  number={2},
  pages={493--525},
  year={2017},
  publisher={Springer}
}

@article{prodan2016ktheory,
  title={Bulk and boundary invariants for complex topological insulators},
  author={Prodan, Emil and Schulz-Baldes, Hermann},
  journal={K},
  year={2016},
  publisher={Springer}
}

@article{lin2022transfer,
  title={Mathematical theory for topological photonic materials in one dimension},
  author={Lin, Junshan and Zhang, Hai},
  journal={Journal of Physics A: Mathematical and Theoretical},
  volume={55},
  number={49},
  pages={495203},
  year={2022},
  publisher={IOP Publishing}
}

@article{thiang2023transfer,
  title={Bulk-interface correspondences for one-dimensional topological materials with inversion symmetry},
  author={Thiang, Guo Chuan and Zhang, Hai},
  journal={Proceedings of the Royal Society A},
  volume={479},
  number={2270},
  pages={20220675},
  year={2023},
  publisher={The Royal Society}
}

@article{ammari2024toeplitz_1,
  title={Mathematical foundations of the non-Hermitian skin effect},
  author={Ammari, Habib and Barandun, Silvio and Cao, Jinghao and Davies, Bryn and Hiltunen, Erik Orvehed},
  journal={Archive for Rational Mechanics and Analysis},
  volume={248},
  number={3},
  pages={33},
  year={2024},
  publisher={Springer}
}

@article{ammari2024toeplitz_2,
  title={Applications of Chebyshev polynomials and Toeplitz theory to topological metamaterials},
  author={Ammari, Habib and Barandun, Silvio and Liu, Ping},
  journal={Reviews in Physics},
volume={13},
pages={Paper No. 100103},
  year={2025}
}

@article {ammari2024toeplitz_3,
    AUTHOR = {Ammari, Habib and Barandun, Silvio and Davies, Bryn and
              Hiltunen, Erik Orvehed and Kosche, Thea and Liu, Ping},
     TITLE = {Exponentially localized interface eigenmodes in finite chains
              of resonators},
   JOURNAL = {Stud. Appl. Math.},
  FJOURNAL = {Studies in Applied Mathematics},
    VOLUME = {153},
      YEAR = {2024},
    NUMBER = {4},
     PAGES = {Paper No. e12765, 25},
}

@misc{braverman2018spectralflowfamilytoeplitz,
      title={The spectral Flow of a family of Toeplitz operators}, 
      author={Maxim Braverman},
      year={2018},
      eprint={1803.11101},
      archivePrefix={arXiv},
      primaryClass={math.DG},
      url={https://arxiv.org/abs/1803.11101}, 
}

@misc{drouot2024bec_curvedinterfaces,
      title={The bulk-edge correspondence for curved interfaces}, 
      author={Alexis Drouot and Xiaowen Zhu},
      year={2024},
      eprint={2408.07950},
      archivePrefix={arXiv},
      primaryClass={math-ph},
      url={https://arxiv.org/abs/2408.07950}, 
}

@article{tauber2022chiral_finite_chain,
  title={Estimating bulk and edge topological indices in finite open chiral chains},
  author={Jezequel, Lucien and Tauber, Cl{\'e}ment and Delplace, Pierre},
  journal={J. Math. Phys.},
  volume={63},
  number={12},
  year={2022},
  publisher={AIP Publishing}
}

@article{bal2024continuous,
  title={Continuous topological insulators classification and bulk edge correspondence},
  author={Bal, Guillaume},
  journal={arXiv preprint arXiv:2412.00919},
  year={2024}
}

@article{mario19proof,
  title = {Proof of the Bulk-Edge Correspondence through a Link between Topological Photonics and Fluctuation-Electrodynamics},
  author = {Silveirinha, M\'ario G.},
  journal = {Phys. Rev. X},
  volume = {9},
  issue = {1},
  pages = {011037},
  numpages = {18},
  year = {2019},
  month = {Feb},
  publisher = {American Physical Society},
  doi = {10.1103/PhysRevX.9.011037},
  url = {https://link.aps.org/doi/10.1103/PhysRevX.9.011037}
}

\end{document}